\documentclass[twoside,11pt]{article}

\usepackage{jmlr2e}
\usepackage{amsmath}

\makeatletter
\let\ftype@table\ftype@figure
\makeatother

\ShortHeadings{Likelihood-free Inference with Jensen--Shannon Divergence}{Corander, Remes, Holopainen, and Koski}
\firstpageno{1}

\begin{document}

\title{ Non-parametric Likelihood-free Inference with Jensen--Shannon Divergence for Simulator-based Models with Categorical Output}
\author{\name Jukka Corander  \email jukka.corander@medisin.uio.no \\ 
\addr  Department of Mathematics and Statistics and Helsinki Institute of Information Technology (HIIT) \\ University of Helsinki  \\ Pietari Kalmin katu 5, 00014 Helsingin Yliopisto, Finland \\
\addr  Department of Biostatistics, Institute of Basic Medical Sciences \\ University of Oslo \\ Sognsvannsveien 9, 0372 Oslo, Norway \\
\addr Parasites and Microbes, Wellcome Sanger Institute \\ Cambridge, CB10 1SA, UK 
\AND \name Ulpu Remes \email  u.m.v.remes@medisin.uio.no\\  
\addr  Department of Biostatistics, Institute of Basic Medical Sciences \\ University of Oslo \\ Sognsvannsveien 9, 0372 Oslo, Norway 
\AND \name Ida Holopainen \email ida.x.holopainen@helsinki.fi \\ 
\addr  Department of Mathematics and Statistics \\  University of Helsinki \\ Pietari Kalmin katu 5, 00014 Helsingin Yliopisto, Finland
\AND  \name  Timo Koski \email tjtkoski@kth.se \\ 
\addr  Department of Mathematics and Statistics and Helsinki Institute of Information Technology (HIIT) \\  University of Helsinki \\ Pietari Kalmin katu 5, 00014 Helsingin Yliopisto, Finland \\
\addr  KTH Royal Institute of Technology \\ Lindstedtsv\"agen 25, 100 44 Stockholm, Sweden}

\editor{NN}

\maketitle

\begin{abstract}%   <- trailing '%' for backward compatibility of .sty file
Likelihood-free inference for simulator-based statistical models has recently attracted a surge of interest, both in the machine learning and statistics communities. The primary focus of these research fields has been to approximate the posterior distribution of model parameters, either by various types of Monte Carlo sampling algorithms or deep neural network -based surrogate models. Frequentist inference for simulator-based models has been given much less attention to date, despite that it would be particularly amenable to applications with big data where implicit asymptotic approximation of the likelihood is expected to be accurate and can leverage computationally efficient strategies. Here we derive a set of theoretical results to enable estimation, hypothesis testing and construction of confidence intervals for model parameters using asymptotic properties of the Jensen--Shannon divergence. Such asymptotic approximation offers a rapid alternative to more computation-intensive approaches and can be attractive for diverse applications of simulator-based models.  
\end{abstract}


\section{Introduction}
\label{ppo:sec:intro}

% \subsection{ Simulator-Based Models and the  Jensen--Shannon Divergence}

 There  are plenty of occasions   for statistical inference and learning in various disciplines,  in  natural and social science and in medicine,   where  the source of the observed data may have a scientific  description with  partly unknown  components    or   with   a prohibitively complex   analytical expression, leading to the need of using simulator-based models and likelihood-free inference, see  \citet{cranmer}  for a recent comprehensive review of this field.   Simulator-based models  are implemented as   computing  agents, which  specify how synthetic, e.g. simulated data are generated  as samples from a parametric statistical model to match the  intended  description. Let us consider  an observed data set   $\mathbf{D}$, seen as  $n_{o}$ independent and identically distributed (i.i.d.) samples  from  a finite discrete set or alphabet ${\cal A}$.     The physical data source that has emitted  $\mathbf{D}$  has a complex mathematical model with parameters $\theta$.  Frequently physical experimentation with the data generating source is prohibitively  expensive (or even impossible), 
whereas computer  simulations can be done to learn about $\theta$.  Several real life examples are found in   \citet{gutmann2016bayesian}  and  \citet{lintusaari2017fundamentals}.

The model   is assumed to be implemented  in a corresponding computer program  ${\mathbf M}_{C}$ called a  simulator model.  
Such simulator models specify  how synthetic, i.e.,  simulated data are generated  as samples for any given a value   ${\theta}$ of a model parameter.  
 A simulator-based  model  is  generative in the concrete sense that the functions  in  ${\mathbf M}_{C}$  may be as complex and flexible as needed as long  as exact sampling
is possible. This permits  researchers to develop  sophisticated models for data generating mechanisms without having to make strong simplifying assumptions to satisfy the requirements of  
computational and mathematical  tractability,    see \citet{gutmann2016bayesian} and \citet{lintusaari2017fundamentals}  for   theory  and  references,  and \citet{lintusaari2018elfi} for a platform of inference methods.

Let  $\mathbf{X}_{\theta}$  be    $n$  i.i.d. samples  generated by  ${\mathbf M}_{C}$  under  a given $\theta$, in short, samples from   ${\mathbf M}_{C}(\theta)$.  The primary question is then, how to use    $\mathbf{D}$ and $\mathbf{X}_{\theta}$  for statistical inference about 
real parameters $\theta \in \mathbb{R}^{d}$. 

 A simulator-based  statistical model   ${\mathbf M}_{C}$ can be interpreted  as an  implicit statistical model  in the sense of  \citet{diggle1984monte}, who were among the first to consider likelihood-free inference. To clarify this,  a probability distribution $P_{\theta} $  on  ${\cal A}$ is induced by  $P_{\theta} (a)= P\left( {\mathbf M}_{C}(\theta)=a \right) $ for any $a \in {\cal A}$.  Due to the   complexity of  the simulator-based model,   $P_{\theta} $   cannot be written by  means of an explicit  analytical expression  and   the likelihood  function for $\theta$  given  $\mathbf{D}$, ${\cal L}_{\mathbf{D}}(\theta)$,   cannot be written down explicitly. Thus simulator-based inference is  likelihood-free inference (LFI). 
Following   \citet{diggle1984monte}  the likelihood function ${\cal L}_{\mathbf{D}}(\theta)$ is called an implicit likelihood function.   In the LFI approaches typically considered in the statistics community the simulations are generally evoked in two ways: 
\begin{itemize}
\item to build an explicit surrogate parametric likelihood, or 
\item to accept/reject parameter values according to a measure of  discrepancy of data from the observations (Approximate Bayesian Computation). A  discrepancy between the observed $\mathbf{D}$ and the simulated  $\mathbf{X}_{\theta}$ is introduced  to approximate the implicit  likelihood.  
\end{itemize}
In both cases, simulations are adaptively tailored to the value of the observation data and our work combines some features of both of them. Inference about $\theta$  under  simulator-based  statistical models is based on  the computational minimization of a  measure of  discrepancy between the observed $\mathbf{D}$ and the simulated  $\mathbf{X}_{\theta} $ without 
an explicit surrogate likelihood but by quantities that under certain conditions turn out to be good estimates of the implicit likelihood.

Naturally, we  will by necessity be  making some tacit  assumptions about the structure of  ${\mathbf M}_{C}$  via assumptions on   $P_{\theta}$. In this regard the  Kennedy-O'Hagan  theory  (KOH),  see \citet [p.~2]{kennedy2000predicting} and 
 \citet{tuo2018prediction} for an up-date of it,  postulates   ${\mathbf M}_{C}$  with different levels of code, with  more complex slow codes being  achieved by  expanding  the simpler fast codes. Thus   simple natural laws can  be incorporated   in  ${\mathbf M}_{C}$, too. 
 KOH interprets the output of  ${\mathbf M}_{C}$  as a deterministic function observed with additive noise.  %As an attractive way to use our theoretical results in practice, we apply the optimization technique  in \citet{gutmann2016bayesian}  called  BOLFI  (Bayesian optimization for likelihood-free inference), and its  implementation  in the  ELFI software package  \citet{lintusaari2018elfi}. % ulpu comment: commented this out because (i think) rather than emphasise BOLFI, we should make it clear that the theoretical work presented here does not build on BOLFI. this is so that there is no confusion when we later observe that we cannot use BOLFI to calculate the test statistic values when there is variation between the simulated data sets, and at the same time we conclude (based on experiments that do not use BOLFI) that the proposed test statistic does work.

Our theoretical results stem from the use of  the  Jensen--Shannon divergence (JSD) between the observed and synthetic data, which was also considered as the basis for model training in the original work on generative adversarial neural networks (GANs), see \cite{good}.  It was  established  by    \citet{osterreicher1993statistical}  that  JSD  expresses  statistical information in the sense of  \citet{degroot1962uncertainty}  and that it is in fact   
a $\phi$-divergence,  as defined by   \citet{csiszar1967information}. This brings a major  advantage in  that   the  well known general properties of     $\phi$-divergences  are available  for  the study. A survey of the applications of $\phi$-divergences and other measures of statistical  divergence  is \citet{basseville2013divergence}. Further remarks about  the background  and origins of JSD can be found  in \citet{corremkos}.
 
%\subsection{  Organization of the Paper}
 
The remaining paper is organized as follows. Notations and basic properties  for  simulator-based models and categorical probability distributions are  introduced in Section \ref{sec:phiestimat11}.  The definition and first  properties of JSD are recapitulated in Section \ref{jsd}. By  applications of inequalities in  information theory we then prove the existence of JSD-based estimator that maximizes the implicit log-likelihood, and further derive Taylor expansions and statistics whose asymptotic properties are such that approximate confidence intervals and hypothesis tests can be directly obtained from them. Finally, we demonstrate our approach by application to multiple simulator models and discuss possible extensions of the theoretical framework to more general model classes. 

\section{ Categorical Distributions, Sampling and Sufficient Summaries}\label{sec:phiestimat11}  
In this section the definition and notations for categorical distributions  are stated. In this work  a categorical distribution is  in the first place   a map from  a finite and discrete alphabet  to the interval  $[0,1]$. The  map  will be   identified with the probability simplex,  a subset of a finite  dimensional vector  space.     Then we   give the definition and
notations  for simulator-based  categorical models and their representations as implicit models. 

It has been observed that approximate Bayesian computation (ABC)  requires computations based on low dimensional   summary statistics, rather than on the  full data sets.  It is desired to     find low dimensional summaries which are informative about the parameter inference. One  approach to dimension reduction in ABC has  focused on various 
approximations of the notion of  sufficiency concept   as found  in statistical inference, see the review in  \citet{prangle2015summary}.  In the case of categorical data we show in this section that the the vector of relative frequencies of the 
categories  in a sample is an exact  sufficient statistic.  The proof is based on the first  properties of the multinomial distribution. 

\subsection{Categorical Distributions }\label{sec:phiestimat} 
Let  ${\cal A}=\{a_{1},\ldots, a_{k}\}$ be a finite set, $k \geq 2$. We are concerned with a situation, where 
$k$  and  all  categories $a_{j}$ are known. This excludes the issues   of  very large ${\cal A}$, cf.,  \citet{kelly2012classification}.     $\mathbb{R}^{{\cal A}}$  denotes   the set of real valued functions on  ${\cal A}$. 
%With  the can  be written as 
%$$
%f(x)= \sum_{i=1}^{k} f_{i} [x=a_{i}], \quad   x \in {\cal A}.
%$$  
We introduce the set of categorical (probability)  distributions as 
\begin{equation}\label{allprob}
\mathbb{P} = \left\{ \mbox{all probability distributions on}   \quad \cal A \right\} \subset  \mathbb{R}^{{\cal A}}.
\end{equation}
The   Iverson bracket  $[ x =a_{i}]  \in  \mathbb{R}^{{\cal A}}$ is  defined for each $a_{i} \in {\cal A}$ by 
$
[x  =a_{i}]=  1$, if  $ x =a_{i}$, and 
$
[x  =a_{i}]=  0$, if  $ x \neq  a_{i}$. 
Any  $ P \in \mathbb{P} $ can then  be   written  as   
\begin{equation}\label{iverson}
P(x)= \prod_{i=1}^{k} p_{i}^{[x=a_{i}]}, x \in {\cal A},
\end{equation}
where  ($0^{0}=1, 0^{1}=0$), and $ p_{i} \geq 0$, $\sum_{i=1}^{k}p_{i}=1$. 
 The support 
of $P \in \mathbb{P}$  is  
$
{\rm supp}(P)= \{ a_{i} \in {\cal A}  |  \quad p_{i}=P(a_{i}) >0  \}$. 
If  $X$ is a random variable (r.v.) assuming values on ${\cal A}$, 
$X \sim P$ $ \in \mathbb{P}$  means that   $P(X=x)= P(x)$ for all $x \in {\cal A}$.  

Any  $P \in \mathbb{P}$  is naturally  understood   as a probability vector ${\bf p}$,  an element of the probability simplex  $\triangle_{k-1}$ defined by  
\begin{equation}\label{simplex}
\triangle_{k-1} :=\left \{ {\bf p} = \left(p_{1}, \ldots, p_{k} \right)\mid  p_{i} \geq 0 , i=1, \ldots,k;   \sum_{i=1}^{k}p_{i}=1 \right \}.
\end{equation}  
We write  this one-to-one correspondence between $\mathbb{P} \in  \mathbb{R}^{{\cal A}}$ and $\triangle_{k-1}\subset \mathbf{R}^{k} $ as 
\begin{equation}\label{trmap}
\triangle\left(P \right)= {\bf p}. 
\end{equation} 
The  $ i$-th face of $\triangle_{k-1}$  is defined as    $\partial_{i}\triangle_{k-1}=\{ {\bf p} \in \triangle_{k-1} | p_{i}=0 \}$.  Any  face is in fact a probability 
simplex in $\mathbf{R}^{k-1}$.

The simplicial boundary of  $\triangle_{k-1}$  is  $\partial \triangle_{k-1} =\cup_{i=1}^{p}\partial_{i}\triangle_{k-1}$ $= \{ {\bf p} \in \triangle_{k-1}  |p_{i} =0 
\quad  \mbox{for some } i \}$.
The simplicial or topological  interior of  $ \triangle_{k-1} $   is $\stackrel{o}{\triangle}_{k-1} := \triangle_{k-1} \setminus \partial \triangle_{k-1} $, i.e., 
\begin{equation}\label{topint}
\stackrel{o}{\triangle}_{k-1}= \{ {\bf p} \in \triangle_{k-1}  |p_{i} > 0,  i=1, \ldots,k;   \sum_{i=1}^{k}p_{i}=1  \}.
\end{equation}
We note that $ 
{\rm supp}(P)=  {\cal A}   \Leftrightarrow \triangle(P) \in   \stackrel{o}{\triangle}_{k-1}$.  
The assumption  
\begin{equation}\label{posass}
 \triangle(P) \in   \stackrel{o}{\triangle}_{k-1}
\end{equation}
will be required  of various members  of $\mathbb{P}$ in several situations in the sequel. 
 \begin{example}\label{standbas}{\bf The Standard Simplex:}
For $l=1, \ldots, k$  we define  $e_{l}  \in \mathbb{P} $   by    
\begin{equation}\label{iverson2}
e_{l}(x)= \prod_{i=1}^{k} e_{i,l}^{[x=a_{i}]}, x \in {\cal A},
\end{equation}
where for  $i=1, \ldots, k$  we have 
$$
e_{i,l}= \left\{ \begin{array}{cc}  1 &  i=l \\
0 & i \neq l. \end{array} \right.
$$
We have by Equation~(\ref{trmap})  the unit vector ${\bf e}_{l} =\triangle(e_{l} )$ 
\begin{equation}\label{bas}
{\bf e}_{l} = \left(  0, 0,  \ldots, 0, \underbrace{1}_{\mbox{position $l$}} , 0, \ldots, 0 \right) \in  \triangle_{k-1}.
\end{equation} 
Each  ${\bf e}_{l}$ is called a vertex of   $\triangle_{k-1}$, since  any  $ {\bf p} \in  \triangle_{k-1}$  can be written as 
\begin{equation}\label{vertex} 
 {\bf p} = \sum_{i=1}^{k} p_{i} {\bf e}_{i}. 
\end{equation} 
Hence $ \triangle_{k-1}$  is also known as the standard simplex.  
 \end{example}
\begin{example}\label{barycdef}{\bf The Barycenter of  $\triangle_{k-1}$:}
The discrete uniform distribution on ${\cal A}$  is
\begin{equation}\label{unif}
P_{U}(x)= \prod_{i=1}^{k} \left(  \frac{1}{k}\right)^{[x=a_i]},  \quad  x \in {\cal A}. 
\end{equation}
 $\triangle\left( P_{U} \right)$ is called  the barycenter of  $\triangle_{k-1}$. 
\end{example}

\subsection{Data, Simulator-based Categorical Models  and Implicit  Categorical Models; Sufficiency}\label{sec:phiestimat22} 
Consider  $P_{o} \in \mathbb{P}$, which is nominally   the   true distribution assumed to have generated the  observed data 
$ {\mathbf D}= (D_{1}, \ldots, D_{n_{o}}) $,  that are assumed to be  an i.i.d. $n_{o}$-sample  $\sim$ $P_{o}$.  Value of $n_{o}$  is fixed by external  circumstances   that are   independent  of $\mathbf{D}$.  We assume  Equation~(\ref{posass}) to hold for  $P_{o}$. 

 The summary statistics for $ {\mathbf D}$  will be the empirical distribution $\widehat{P}_{\mathbf D} \in \mathbb{P}$. This is  computed  
in terms of the relative frequencies of the categories $a_{i}$ in   ${\mathbf D}$. Formally, we write           
\begin{equation}\label{relfrekv}
\widehat{p}_{i}= \frac{n_{o,i}}{n_{o}}, \quad i=1, \ldots, k,
\end{equation}
where $n_{o,i} = $ the number of samples $D_{j}$ in $ {\mathbf D}$ such that $D_{j}=a_{i}$, and  following Equation~(\ref{iverson})
\begin{equation}\label{typen}
\widehat{P}_{\mathbf D}(x)=\prod_{i=1}^{k} \widehat{p}_{i}^{[x=a_{i}]}, \quad x \in{\cal A}. 
\end{equation}

Let  ${\mathbf M}_{C}$  be the simulator-model.  Citing   \citet{lintusaari2017fundamentals},  the functions in  ${\mathbf M}_{C}$   are computer programs, which we run $n$ times   taking  as input  random numbers  $V$ and the parameter $\theta  \in \Theta \subset  \mathbb{R}^{d}$, $d < k =\mid {\cal A} \mid$  and  in the special  case  under consideration here produce as output  $\mathbf{X}_{\theta}=\left(X_{1,\theta}, \ldots, X_{n,\theta}\right)$,  $n$ i.i.d. samples of categories in  ${\cal A}$.  
 We write the simulated outputs as  $\mathbf{X}_{\theta}$   and the corresponding function as 
 ${\mathbf M}_{C}(\theta)$.   

By this designation,
$
 p_{i}(\theta)=  P \left( {\mathbf M}_{C}(\theta) = a_{i}  \right) 
$
 for any  $\theta \in \Theta$  induces the category probabilities  so that there is the  distribution $P_{\theta} \in \mathbb{P}$ so that  the probability of a simulated  output $X$ being $ x \in {\cal A}$, i.e.,  $ P _{\theta}\left( X= x  \right) $, is 
\begin{equation}\label{pariversion}
P_{\theta}(x):= \prod_{i=1}^{k}p_{i}(\theta)^{[x=a_{i}]}, x \in {\cal A}, \theta \in \Theta.   
\end{equation}   
In this  $p_{i}(\theta)$ are  $k$ functions that  have no (fully)  explicit expression, i.e., they are implicit functions  of $\theta$ in the sense  of Diggle and  Gratton   satisfying  $p_{i}(\theta) \geq 0$, $\sum_{i=1}^{k} p_{i}(\theta)=1$ for all $\theta \in \Theta$. 
The   implicit model  representation  of 
 ${\mathbf M}_{C}$ in $\mathbb{P}$  is denoted  by   $\mathbb{M}= \left\{ P_{\theta} \mid  \theta \in \Theta \right \}\subset \mathbb{P}$, 
$$
  \mathbb{M}= \left\{ P_{\theta} \mid  \theta \in \Theta \right \}  \models {\mathbf M}_{C} = \left\{ {\mathbf M}_{C} (\theta) \mid  \theta \in \Theta \right \} . 
$$

We are going to use the customary notation   $ {\mathbf X}_{\theta} \sim   P_{\theta}$, or simpler  $ {\mathbf X} \sim   P_{\theta}$ , which is to be understood in the above sense of generative simulator-based sampling, not as sampling  from  a known categorical distribution.
This  statement  is fine-tuned   by  intractable   likelihoods  in  Example  \ref{machlear} below.  In some computations to follow we 
shall write  $ {\mathbf X}_{\theta} \sim   P^{(n)}_{\theta}$, where  $ P^{(n)}_{\theta}= \times_{i=}^{n} P_{\theta}$.
Of course, $ {\mathbf D} \sim P_{o} $ does not presume   simulation in any  tangible fashion, but we do not alter  the 
notation for this purpose.

%There is the  question of specification, whether $P_{o} \in   \mathbb{M}$ or not, that is,  if  there is in  $ {\mathbf M}_{C}$  a function  expressing   the true data generating mechanism. 

Let the relative frequencies of the categories $a_{j}$ in   ${\mathbf X}_{\theta}$ be
$
\widehat{q}_{i}, i=1, \ldots, k$.  This yields  
\begin{equation}\label{qhat2}
\widehat{Q}_{\theta}(x) =  \prod_{i=1}^{k} \widehat{q}_{i}^{[x=a_{i}]}, x \in {\cal A}
\end{equation}
as the summary statistics  of  a given sample  $ {\mathbf X}_{\theta}$.    Let  $\widehat{\mathbf{Q}}_{\theta} $ be  the  r.v.  defined by  the  summary statistics for  the r.v. $\mathbf{X} \sim   P_{\theta}$. A shorthand for this is  $\widehat{\mathbf{Q}}_{\theta} \in \mathbb{M}_{n}(\theta) $.
Consider    r.v.'s $\xi_{i}$ with  $\xi_{1}+\ldots +\xi_{k}=n$, $\xi_{i} \in \{0,\ldots, n\}$,  such that
the vector $\underline{ \xi}:=\left(\xi_{1},\ldots, \xi_{k} \right)$ has the multinomial distribution with parameters $n$ and $\left( p_{1}\left(\theta \right),\ldots,  p_{k}\left(\theta \right)\right) $.  Then
\begin{equation}\label{Qhat2}
\triangle\left(\widehat{\mathbf{Q}}_{\theta} \right) = \left(  \frac{\xi_{1}}{n},\ldots,  \frac{\xi_{k}}{n}\right)  = \frac{\underline{ \xi}}{n}.
\end{equation}
Summary of    categorical data set by  their   empirical distribution, i.e., by  outcomes $\widehat{Q}$  of $ \widehat{\mathbf{Q}}_{\theta} $  is next  justified next  by means of the Neyman factorization citerion and Bayes sufficiency w.r.t.  $\mathbb{M}$.  For the  definition and  structure  of  sufficiency  we refer to   \cite[Ch. 4, pp.~192--193]{bernardo2009bayesian}.

In the  sufficiency result  below,   the subindices $ \theta$ of  simulated data and  their summary have been dropped  as there is no explicit  formal dependence of the simulated empirical distributions on $\theta$. 
\begin{proposition}[Sufficiency]\label{sufffic}
Assume $\widehat{\mathbf{Q}} \in \mathbb{M}_{n}(\theta)  $.      Let $ \underline{x} =\left(x_{1}, \ldots,x_{k} \right) \in {\cal A}^{n}$   and let     $n_{i} =$ be the number of   $x_{j} \in $ $ \underline{x} $  such that $x_j=a_{i}$.  Let  $\underline{n} =\left(n_{1}, \ldots, n_{k} \right)$  so that
$\underline{n}/n: =\left(\frac{n_{1}}{n}, \ldots, \frac{n_{k}}{n}\right)$  $\in  \triangle_{k-1}$.
Then  it holds that  

\begin{description}
\item[i)]  There is a function  $G(\underline{x})$  such that   
\begin{equation}\label{suffstat}
  P_{\theta}  \left( \mathbf{X} = \underline{x}\right)= P_{\theta}  \left( \triangle \left(\widehat{\mathbf{Q}} \right) = \underline{n}/n \right)G\left( \underline{x}\right),
\end{equation}
where  $G\left( \underline{x}\right) >0$. 

\item[ii)]Under any prior density on $p(\theta)$  on $\Theta$ it holds that
\begin{equation}\label{basssuff} 
P\left(\theta | \mathbf{X}= \underline{x} \right)=   P\left(\theta |  \triangle \left(\widehat{\mathbf{Q}} \right) = \underline{n}/n  \right).  
\end{equation}
\end{description}
\end{proposition}
\begin{proof} 
\begin{description}
\item[i)] As    $\underline{ \xi}:=\left(\xi_{1},\ldots, \xi_{k} \right)$ is   multinomial  with parameters $n$ and $ p_{1}\left(\theta \right),\ldots,  p_{k}\left(\theta \right) $.
By multinomial probability, since $n$ is known in advance and does not depend on $\mathbf{X}$,   and   Equation~(\ref{Qhat2})
\begin{equation}\label{multinompid}
P_{\theta}  \left(\triangle \left(\widehat{\mathbf{Q}} \right)  = \underline{n}/n   \right)= P_{\theta}  \left(\underline{ \xi}=\left(n_{1},\ldots, n_{k} \right)\right)=
\frac{ n! }{ \prod_{j=1}^{k}n_{j}!}  \prod_{j=1}^{k}p_{j}^{n_{j}}(\theta).
\end{equation}
But   $ \prod_{j=1}^{k}p_{j}^{n_{j}}(\theta)=  P_{\theta}  \left( \mathbf{X} =\underline{x}\right)$,    the probability  of  $ \underline{x}$ as  $n$ i.i.d. outputs  of  ${\mathbf M}_{C} (\theta)$.   With $G\left( \underline{x}\right)=1/\frac{ n! }{ \prod_{j=1}^{k}n_{j}!}$   the Neyman factorization in  
Equation~(\ref{suffstat}) holds.

\item[ii)] 
Let $p(\theta)$ be a prior density on $\Theta$.  Then   
\begin{equation}\label{datallikelihood}
P\left(\mathbf{X}= \underline{x}  \right ) := \int_{\Theta} P_{\theta} \left(\mathbf{X}= \underline{x}  \right)p(\theta)d\theta.   
\end{equation} 
Here 
$d\theta$ is the Lebesgue measure induced on $\Theta$.
 By Bayes rule, Equation~(\ref{datallikelihood}) and Equation~(\ref{multinompid}) we obtain 
 \begin{eqnarray}
P\left(\theta | \mathbf{X}= \underline{x} \right) & =&  \frac{P_{\theta} (\mathbf{X}= \underline{x} )p(\theta)}{P\left(\mathbf{X}= \underline{x} \right ) }  \nonumber \\
& = & \frac{\frac{ n! }{ \prod_{j=1}^{k}n_{j}!} P_{\theta} (\mathbf{X}= \underline{x} )p(\theta)}{\int_{\Theta} \frac{ n! }{ \prod_{j=1}^{k}n_{j}!} P_{\theta} \left(\mathbf{X})= \underline{x}  \right)p(\theta)d\theta.    } \nonumber  \\
&= & 
\frac{
    P_{\theta}  \left( \triangle \left(\widehat{\mathbf{Q}} \right) = \underline{n}/n   \right) p(\theta)}{\int_{\Theta}   \  P_{\theta}  \left( \triangle \left(\widehat{\mathbf{Q}} \right)= \underline{n}/n   \right)p(\theta)d\theta.    }  \nonumber  \\
& = &    P\left(\theta |  \triangle \left(\widehat{\mathbf{Q}} \right)  = \underline{n}/n  \right), \nonumber 
\end{eqnarray}
where  Bayes rule was  invoked in  the final step.  This  proves Equation~(\ref{basssuff}) as claimed.  
 \end{description} \end{proof}
%In addition, $ P_{o}  \left( \mathbf{D}  = \underline{ n}_{o} \right)= P_{o}  \left( \widehat{P}_{\mathbf{D} }\right)G\left(\mathbf{D} \right)$, where $G\left(\mathbf{D} \right)=1/\frac{ n_{o}! }{ \prod_{j=1}^{k}n_{o,j}!}$  .

Let  ${\bf p}_{\theta} =\left(p_{1}(\theta), \ldots, p_{k}(\theta) \right)= \triangle\left(P_{\theta} \right)$ and 
$\widehat{{\bf q}}_{\theta}=  \triangle \left(\widehat{\mathbf{Q}}_{\theta} \right)  $. Then,
 since as $\xi_{i} \sim {\rm Bin}(n, p_{i}\left(\theta \right))$,  and by  Equation~(\ref{vertex})
\begin{equation}\label{vertexprop} 
{\bf p}_{\theta} = \sum_{i=1}^{k} p_{i}(\theta) {\bf e}_{i} =  \sum_{i=1}^{k} E_{P_{\theta}}\left[  \frac{\xi_{i}}{n} \right]{\bf e}_{i}= E_{P_{\theta}}\left[  \widehat{{\bf q}}_{\theta} \right].
\end{equation}

\begin{example}\label{machlear} Models that are not implicit are  by  Diggle and  Gratton called  prescribed  models.   
{\rm The inference about $\theta$ by JSD to be discussed here works also  even for  sampling  from  a prescribed model  with  an  intractable  likelihood.  To discuss   this    more precisely, consider   any real vector  $\theta \in   \mathbb{R}^{d}$
and $k$ known functions $g_{i}\left(\theta \right)$, $i=1, \ldots,k$.   The soft-max 
assignments 
$
p_{i}(\theta):= e^{g_{i}( \theta)} / C(\theta)$, $ i=1, \ldots, k
$
determine  by Equation~(\ref{pariversion})  a prescribed  model 
$ \mathbb{M}= \left\{ P_{\theta} \mid  \theta \in \Theta \right \} \subset \mathbb{P}$.  The normalizing constant $C(\theta)$  cannot frequently  be evaluated by  a closed-form expression. Hence  the likelihood is  intractable. But one can anyhow  generate synthetic samples 
$X_{\theta} \sim P_{\theta}$ and $\widehat{Q}_{\theta} $ by means of, e.g.,  the  Gumbel trick of \citet[Lemma 6, p. 123]{yellott1977relationship}. Hence nonparametric likelihood-free inference by means of JSD   is applicable. }
\end{example}

The suggested  notion of closeness  between  the summary statistics  $\widehat{Q}_{\theta} \in \mathbb{M}_{n}(\theta) $  and $\widehat{P}_{\mathbf D} $  is next  explicitly defined   by means of the Jensen--Shannon divergence.  
 
\section{The Jensen--Shannon Divergence}\label{jsd}
In this section we define the Jensen--Shannon divergence as a  $\phi$-divergence and recapitulate some of its properties, as found useful for  the present purposes.   Then we discuss an interpretation based on the early contributions  by  
N.~Jardine and R.~Sibson, for whom, however,  the later terminology of   Jensen--Shannon divergence was not available. We define the total  variation distance and cite a  useful   inequality for the distance between  the Shannon entropies of two categorical distributions  in terms of the total variation distance.  

\subsection{The Definition and the Range Property} 
We denote by $\phi$ a continuous convex function  $(0, + \infty)$ $\stackrel{\phi}{\mapsto} \mathbf{R}\cup {+\infty}$. 
The function $\phi$  has the properties $0 \phi \left( \frac{0}{0} \right) =0$ and $0\phi(x/0)= \lim_{\epsilon  \rightarrow 0} \epsilon
\phi(x/\epsilon)$. We require also that  $\phi(1)=0$ and that $\phi(x)$ is strictly convex at $x=1$.  We call $\phi$  a divergence function.   \citet[Ch. 3]{vajda1989theory}  presents   several  properties  specially valid for  convex functions  of  $[0, + \infty)$.

For two generic categorical probability   distributions $P:  P(x)= \prod_{i=1}^{k} p_{i}^{[x=a_{i}]}$  and $Q: Q(x)= \prod_{i=1}^{k} q_{i}^{[x=a_{i}]}$ in  $ \mathbb{P}$, we define the $\phi$-divergence $D_{\phi}( P, Q)$, also known as $\phi$-divergence  of  Csiszar,  between $P$ and $Q$ by means of a divergence function as  
\begin{equation}\label{informationone}
D_{\phi}( P, Q) := \sum_{x \in {\cal A}} Q(x)  \phi\left( \frac{P(x)}{Q(x)}\right)= \sum_{i=1}^{k} q_{i} \phi\left( \frac{ p_{i}}{q_{i}}\right).   
 \end{equation}
Since  $\phi(1)=0$, the Jensen inequality gives  $D_{\phi}( P, Q) \geq 0$.  The property   $D_{\rm JS}( P, Q)=0 \Leftrightarrow P=Q$ follows by   strictly convexity  at $x=1$, see  \citet{csiszar1967information}.
Comprehensive studies of $\phi$-divergences and generalizations on general abstract spaces are  \citet{liese2006divergences}  and \citet[Ch.  8 \& 9]{vajda1989theory}. Concise presentations  of the main properties  are  found  in  \citet[Ch. 1.2]{pardo2018statistical} and \citet{osterreicher2002csiszar}.  

For a first instance,   we select  in Equation~(\ref{informationone}) the divergence function $\phi(x) = x \ln x$ and  the resulting divergence becomes 
 the Kullback--Leibler divergence (KLD) given by $
 D_{\rm KL}( P, Q) := \sum_{i=1}^{k} p_{i} \ln \left( \frac{ p_{i}}{q_{i}}\right)$. 
Non-symmetry  $D_{\rm KL}( P, Q) \neq  D_{\rm KL}( Q,P)$, if $P \neq Q$ holds in general. 

Next,  the Jensen--Shannon divergence is   denoted by  $D_{\rm JS}( P, Q)$,    and  is defined with $M:= \pi P + (1-\pi)Q$ as 
 \begin{equation}\label{jsinformation}
D_{\rm JS}( P, Q) := \pi  D_{\rm KL}( P, M) + (1-\pi)  D_{\rm KL}( Q,  M), \quad 0 < \pi < 1.
 \end{equation}
It  is shown in  \citet{osterreicher1993statistical} that $D_{\rm JS}( P, Q)$ is a $\phi$-divergence with the divergence function 
\begin{equation}\label{jenshaphi}
\phi_{\rm JS}(x):=  \pi \cdot  x \ln x - (\pi x + (1- \pi)) \ln ( \pi x +(1- \pi)).
\end{equation}
It follows as a special case of the range property  of   any    $\phi$-divergence with  $\phi(1)=0$, \citet[Thm 5, p.~4399]{liese2006divergences}, that 
\begin{equation}\label{range2}
0 \leq D_{\rm JS}( P, Q) \leq B(\pi) \leq  B(1/2) =\ln(2),
\end{equation}  
where  $B(\theta)$  is   the binary entropy function in natural logarithms (nats) defined  for  $\theta  \in [0,1]$  by  
$
B(\theta) := -\theta \ln(\theta) -(1-\theta) \ln (1 -\theta)$. Hence $D_{\rm JS}( P, Q)$ is bounded, even  if the supports of $P$ and  $Q$ differ and   hence JSD is a smoothing of KLD, as KLD  can be equal to $+\infty$.  
The uniqueness property   can be  checked   by the uniqueness property of KLD.   
 
\subsection{An  Interpretation}
$D_{\rm JS}( P, Q)$ appears under the name  information radius of order one  in \citet{jardine1971mathematical}   and  \citet{sibson1969information}. 
In  \citet[Thm 2.8., p.~154]{sibson1969information} it  is   shown by an analytic proof that 
\begin{equation}\label{perpe}
D_{\rm JS}( P, Q) = B(\pi) \Leftrightarrow  {\rm supp}(P) \cap {\rm supp}(Q)= \emptyset.
\end{equation} 
In words, $D_{\rm JS}( P, Q)$ reaches its maximum as soon as we know with certainty that a sample of  $P$ and cannot be a sample of $Q$, and conversely.   There is thus  a  simple  proof of Equation~(\ref{perpe}) using the  properties of   the Bayesian optimal error in hypothesis testing.

In  \citet[Cor. 2.3, p. 153]{sibson1969information}   $D_{\rm JS}$  is  defined by 
 \begin{equation}\label{jsinformationsibs}
D_{\rm JS}( P, Q) = \inf_{\nu  } \left[ \pi  D_{\rm KL}( P, \nu) + (1-\pi)  D_{\rm KL}( Q,  \nu) \right], 
 \end{equation}
where $\nu$ is any probability in $\mathbb{P}$  dominating the convex combination  $ \pi   P + (1-\pi)  Q $.  

In \citet[pp.~13$-$16] {jardine1971mathematical}  the interpretation of    $D_{\rm JS}( P, Q) $ via Equation~(\ref{jsinformationsibs}) is as follows.   $ \pi  D_{\rm KL}( P, \nu)$ $+$ $ (1-\pi)  D_{\rm KL}( Q,  \nu)$ is  the amount of information per sample to discriminate against  $\nu$, when sampling from $P$   or $Q$ with probabilities  $\pi$ and $1-\pi$. 
Then  $\nu$ is chosen  so as to incorporate as much as possible of what is known about which  of $P$ and $Q$ should be chosen.  $D_{\rm JS}( P, Q) $  is 
the remaining  deficit in information.   \\

Next 
$
H(R):= -\sum_{i=1}^{k} r_{i} \ln r_{i} 
$
is the Shannon entropy of   $R \in \mathbb{P}$ in nats.
A special case of an  identity  in \citet[Lemma 4]{topsoe1979information}   is 
\begin{equation}\label{klinformationiudnet}
 D_{\rm JS}( P, Q) =H(M) - \pi H(P) - (1-\pi)H(Q).  
\end{equation}
The  right hand side of Equation~(\ref{klinformationiudnet})  is the Jensen--Shannon divergence  $D_{\rm JS}( P, Q)$  as  defined in \citet{lin1991divergence}.  

The divergence function   $\phi(x)=|x-1|$ gives the (total) variation distance denoted by $V(P,Q)$    and
equaling $V( P, Q)= \sum_{i=1}^{k} | p_{i} - q_{i}|$  by Equation~(\ref{informationone}).
 The  variation distance is special in the sense that  the triangle inequality 
holds for a $\phi$-divergence   if and only if for some constant $ \alpha >0 $
$
D_{\phi}( P, Q) = \alpha  V( P, Q) $, 
as shown in  \citet{khosravifard2006exceptionality}, see also \citet{vajda2009metric}.

Next    we cite  \citet[Lemma 2.7, p.~19] {csiszar2011information}. 
This lemma is needed for the asymptotics of  %the BOLFI statistic defined in Section \ref{thebolfistats}.
JSD between observed and simulated data.
\begin{lemma}\label{lemmakontbd}
$P \in \mathbb{P}$, $Q \in \mathbb{P}$. Assume  $V(P,Q) < 1/2$.  Then 
\begin{equation}\label{cskorner}
|H(P) - H(Q)| \leq  -V(P,Q) \cdot  \ln\left( \frac{V(P,Q)}{k}\right). 
\end{equation}
\end{lemma}

We note finally  a kind of fundamental   justification for JSD  in its  role  here.    

 The topology on $ \mathbb{P}$ induced by  $V(P,Q)$
is called the variation distance topology. For any $P \in \mathbb{P}$  we define  an  open JSD - neighborhood around $P$  as 
\begin{equation}
N(P, \epsilon):= \{ Q \in  \mathbb{P} |  D_{\rm JS}(P,Q)< \epsilon \}. 
\end{equation}
The following is \citet[Thm 2.7., p.153]{sibson1969information}.
\begin{theorem}\label{topbasis}  For varying $P$ and $\epsilon$, $N(P, \epsilon)$  form a
basis for  the variation distance topology.
 \end{theorem} 
 
%\section{Asymptotics of the  BOLFI Statistic}\label{thebolfistats} 
\section{Asymptotics of Simulator-Based JSD}\label{thebolfistats}

The present work focuses on comparison between observed and simulated data.
For $\widehat{Q}_{\theta} \in \mathbb{M}_{n} (\theta)$,    $D_{\rm JS}( \widehat{P}_{\mathbf{D}}, \widehat{Q}_{\theta})$ will be referred to as %the BOLFI statistic.
the simulator-based JSD statistic.
This is non-parametric, as $D_{\rm JS}( \widehat{P}, \widehat{Q}_{\theta})$    is  a symbol with $\theta$ as an argument    in the 
simulator-modeling  sense, whereby   the dependence on   $\theta$ is  through   the  statistical properties  mediated by   $\mathbb{M}_{n}\left(\theta \right)$, i.e., not through the explicit presence like, e.g.,  in a contrast function.

\begin{example}\label{berno} {\rm  ${\cal A}=\{0,1\}$,  $\mathbf{D}$ is an observed    sample of $n_{o}$ zeros and ones, $  \widehat{p}=$(number of  1's  in $\mathbf{D})/n_{o}$,  and $ \widehat {P}_{\mathbf{D}}(x) =    (1-\widehat{p})^{[x=0]} \cdot \widehat{p}^{[x=1]}   $. Let  $ \Theta=(0,1) $ and $\mathbb{M}= \left\{ P_{\theta} \mid  \theta \in  \Theta=(0,1) \right \}$$\models $ $ \mathbf{M}_{C}$. Then 
$
\widehat{Q}_{\theta}(x) =( 1-\widehat{q})^{[x=0]} \cdot \widehat{q}^{[x=1]}$, $ x \in {\cal A}
$,  where $\widehat{q}=$  $\left( \mbox{number of  1's  in } \mathbf{X}_{\theta} \right)/n$, where $ \mathbf{X}_{\theta} \sim \mathbf{M}_{C}(\theta)$.
We take   $\pi=1/2$, and denote   the corresponding JSD  by  $D_{\rm JS,1/2}$. Clearly, for this example, $H\left( \widehat{P}_{\mathbf{D}}\right)=B( \widehat{p}) 
$,  $H\left( \widehat{Q}_{\theta}\right) =B( \widehat{q}) 
$, $ H \left( M  \right)=B \left( \frac{1}{2} \widehat{p} +  \frac{1}{2}\widehat{q}  \right) 
$ and by Equation~(\ref{klinformationiudnet}), %the BOLFI statistic is 
the simulator-based JSD statistic
 \begin{equation}\label{jssymmen} 
D_{\rm JS,1/2}\left( \widehat{P}_{\mathbf{D}}, \widehat{Q}_{\theta} \right)  = B \left( \frac{1}{2} \widehat{p} +  \frac{1}{2}\widehat{q}  \right) - \frac{1}{2} B( \widehat{p}) - \frac{1}{2}B( \widehat{q}).
 \end{equation}
This  explicit expression is  also the JSD between two Bernoulli distributions, since here  $ \widehat{P}_{\mathbf{D}} \in  \mathbb{M}$  and $\widehat{Q}_{\theta} \in  \mathbb{M}$.  In general an empirical distribution is not inside the model that generated it. 
$D_{\rm JS,1/2}$ is applied  for numerical studies  in Appendix 1 of  \citet{jardine1971mathematical}, which does not, however, explicitly  recognize Equation~(\ref{jssymmen}). } 
\end{example}
We find  first   a simple exact  connection between 
$ D_{\rm JS}\left(\widehat{P}_{\mathbf{D}}, \widehat{Q}_{\theta}   \right)$ and $ D_{\rm JS} \left(  \widehat{P}_{\mathbf{D}},  P_{\theta}\right)$.  This will  yield a proof of the almost sure convergence of %the BOLFI-statistic to what we call a JSD statistic. 
the simulator-based JSD statistic to the JSD statistic $ D_{\rm JS} \left(  \widehat{P}_{\mathbf{D}},  P_{\theta}\right)$.

\subsection{An Exact  Representation  for the Simulator-Based JSD Statistic}\label{thebolfistatsrepr} 

\begin{lemma}\label{mellcomp0}
Assume  Equation~(\ref{posass})  for $P_{\theta}$ and $  \widehat{P}_{\mathbf{D}}$. Let   $\widehat{M}:= \pi \widehat{P}_{\mathbf{D}} + (1-\pi)\widehat{Q}_{\theta}$ and $  \widehat{M}_{\rm imp}:= \pi   \widehat{P}_{\mathbf{D}}  + (1- \pi) P_{\theta}$. $\widehat{Q}_{\theta}   \in \mathbb{M}_{n}(\theta)$. Then 
\begin{eqnarray}\label{melcomp2} 
D_{\rm JS}\left(\widehat{P}_{\mathbf{D}}, \widehat{Q}_{\theta}   \right) &=&   D_{\rm JS} \left(  \widehat{P}_{\mathbf{D}},  P_{\theta} \right)   +(1-\pi)D_{\rm KL}(\widehat{Q}_{\theta},  P_{\theta}) \nonumber \\ 
&  & \\
& &  +  D_{\rm KL} \left(  \widehat{M}_{\rm imp},  P_{\theta} \right)  -   D_{\rm KL}( \widehat{M},  P_{\theta}). \nonumber
\end{eqnarray} 
\end{lemma}
\begin{proof}
The compensation identity   \citet[Lemma 7]{topsoe1979information}, see also \citet[p.1603] {topsoe2000some},  tells  that for any 
$ R \in \mathbb{P}$   with ${\rm supp}\left( R \right)= {\cal A}$  and  $M= \pi P + (1-\pi)Q$  we have 
$
D_{\rm JS}( P, Q)= \pi D_{\rm KL}( P, R) + (1-\pi) D_{\rm KL}( Q, R)
-  D_{\rm KL}( M, R)$.
As ${\rm supp}\left(P_{\theta} \right)= {\cal A}$ we have  
$$
D_{\rm JS}\left(P, Q \right) =  \pi D_{\rm KL}(P,  P_{\theta}) + (1-\pi) D_{\rm KL}(Q,  P_{\theta})  -  D_{\rm KL}(M,  P_{\theta}). 
$$  
We take $P=\widehat{P}_{\mathbf{D}}$ and $ Q= $ $\widehat{Q}_{\theta} $  and  obtain 
\begin{equation}\label{melcomp} 
D_{\rm JS}\left(\widehat{P}_{\mathbf{D}}, \widehat{Q}_{\theta}   \right)  = \pi D_{\rm KL}(\widehat{P}_{\mathbf{D}},  P_{\theta})+(1-\pi)D_{\rm KL}(\widehat{Q}_{\theta},  P_{\theta})  - D_{\rm KL}( \widehat{M},  P_{\theta}). 
\end{equation}
   By Equation~(\ref{informationone})  and Equation~(\ref{jenshaphi})
\begin{eqnarray}\label{plan}
D_{\rm JS} \left(  \widehat{P}_{\mathbf{D}},  P_{\theta} \right)&=& D_{\phi_{\rm JS}}\left(  \widehat{P}_{\mathbf{D}},  P_{\theta} \right) = \sum_{i=1}^{k} p_{i} (\theta)\phi_{\rm JS}\left( \frac{ \widehat{p}_{i}}{p_{i} (\theta)}\right) \nonumber \\ 
   &=&\pi  \sum_{i=1}^{k} \widehat{p}_{i}\ln\left( \frac{ \widehat{p}_{i}}{p_{i} (\theta)}\right) \nonumber \\
& &  -
 \sum_{i=1}^{k}  \left(\pi  \widehat{p}_{i} + (1- \pi)p_{i} (\theta)\right) \ln \left(\frac{\pi  \widehat{p}_{i} + (1- \pi)p_{i} (\theta) }{p_{i} (\theta) }\right)\nonumber \\
& =& \pi D_{\rm KL} \left(  \widehat{P}_{\mathbf{D}},  P_{\theta} \right) - D_{\rm KL} \left(  \widehat{M}_{\rm imp}, P_{\theta} \right).
 \end{eqnarray}
 When we substitute  $\pi D_{\rm KL} \left(  \widehat{P}_{\mathbf{D}},  P_{\theta} \right)$  from  Equation~(\ref{plan})  in Equation~(\ref{melcomp}) we have  the equality in Equation~(\ref{melcomp2}) as claimed. \end{proof}

\begin{example}
For the barycenter in Equation~(\ref{unif})  we have 
$
D_{\rm KL}\left(  \widehat{P}_{\mathbf{D}},  P_{U} \right)=$ $ - H\left(   \widehat{P}_{\mathbf{D}} \right) -
\sum_{i=1}^{k}\frac{n_{o,i}}{n_{o}} \ln (1/k). 
$
Hence  we get from  Equation~(\ref{plan})
$$
  D_{\rm JS} \left(  \widehat{P}_{\mathbf{D}},  P_{\theta} \right)   -   \pi D_{\rm KL}\left(  \widehat{P}_{\mathbf{D}},  P_{U} \right) = \ln \frac{{\cal L}_{\mathbf{D}} (1/k)}{{\cal L}_{\mathbf{D}}(\theta)}  -
D_{\rm KL} \left(  \widehat{M}_{\rm imp},  P_{\theta} \right).$$
   Here  $ \ln \frac{{\cal L}_{\mathbf{D}} (1/k)}{{\cal L}_{\mathbf{D}}(\theta)}=\sum_{i=1}^{k}\frac{n_{o,i}}{n_{o}}\ln \frac{1/k}{p_{j}(\theta)} $ is a  log-likelihood ratio  between the uniform distribution  and  the  implicit likelihood.  
\end{example}
\subsection{Almost Sure Convergence of the Simulator-Based JSD Statistic to the JSD Statistic,  as $n \rightarrow +\infty$} \label{bolfianal}
	 
 \begin{theorem}\label{nskonv}
Assume  Equation~(\ref{posass})  for $P_{\theta}$ and $ \widehat{P}_{\mathbf{D}}$. $\widehat{Q} _{\theta} \in \mathbb{M}_{n}(\theta)$.  Then 
\begin{equation}\label{nskonv2}
 D_{\rm JS}\left(\widehat{P}_{\mathbf{D}}, \widehat{Q}_{\theta}   \right)  \rightarrow  D_{\rm JS} \left(  \widehat{P}_{\mathbf{D}},  P_{\theta}\right).
\end{equation} 
$P_{\theta}$-almost surely, as $ n\rightarrow +\infty$. 
\end{theorem}
\begin{proof} By Equation~(\ref{melcomp2}) in  Lemma \ref{mellcomp0} we have 
\begin{eqnarray}\label{melcomp22} 
\left | D_{\rm JS}\left(\widehat{P}_{\mathbf{D}}, \widehat{Q}_{\theta}   \right) -  D_{\rm JS} \left(  \widehat{P}_{\mathbf{D}},  P_{\theta} \right)\right | &\leq &    (1-\pi)D_{\rm KL}(\widehat{Q}_{\theta},  P_{\theta}) \nonumber \\ 
&  & \\
& +&  \left |  D_{\rm KL} \left(  \widehat{M}_{\rm imp},  P_{\theta} \right)  -   D_{\rm KL}( \widehat{M},  P_{\theta})\right |. \nonumber
\end{eqnarray} 
The first term in the right hand side of  Equation~(\ref{melcomp22})   is  shown to converge to $0$ $P_{\theta}$- a.s., as $n \rightarrow +\infty$  in Lemma \ref{dklkonv} of Appendix \ref{nasymptotiksec}. We consider the second term. 

By definition of  $D_{\rm KL}$  we have 
$$
D_{\rm KL} \left(  \widehat{M}_{\rm imp},  P_{\theta} \right)   = -H\left(  \widehat{M}_{\rm imp} \right) + 
\sum_{j=1}^{k} \widehat{m}_{j,\rm imp} \ln \frac{1}{p_{j}(\theta)}
$$
and analogously for  $D_{\rm KL}( \widehat{M},  P_{\theta})$.  By the triangle inequality we get 
\begin{eqnarray}\label{melcomp2245} 
\left |  D_{\rm KL} \left(  \widehat{M}_{\rm imp},  P_{\theta} \right)  -   D_{\rm KL}( \widehat{M},  P_{\theta})\right | &\leq &          \left |   H\left(  \widehat{M}\right)  - H\left(  \widehat{M}_{\rm imp} \right)   \right | \nonumber \\ 
&  & \\
& +&  \sum_{j=1}^{k}\left |  \widehat{m}_{j,\rm imp} -  \widehat{m}_{j} \right | \ln \frac{1}{p_{j}(\theta)},  \nonumber
\end{eqnarray} 
since  $\frac{1}{p_{j}(\theta)} >1$, and where we used the assumption in  Equation~(\ref{posass}) to ensure  that $ \ln \frac{1}{p_{j}(\theta)} < + \infty$. 

Next,  we shall invoke Lemma \ref{lemmakontbd}.  We have here 
$$
V\left(  \widehat{M}, \widehat{M}_{\rm imp} \right)  = \sum_{j=1}^{k}  \left |  \widehat{m}_{j,\rm imp} -  \widehat{m}_{j} \right |
$$
$$
=  \sum_{j=1}^{k} | \left( \pi \widehat{p}_{j=1}^{k} + (1-\pi) p_{j}(\theta)\right)-  \left( \pi \widehat{p}_{j=1}^{k} +  (1-\pi)\widehat{q}_{j} \right | = (1-\pi)  V \left(\widehat{Q}_{\theta}   , P_{\theta} \right).
$$
By Lemma \ref{totvarlemma}, $ V\left(\widehat{Q}_{\theta}, P_{\theta}\right)  \rightarrow 0$, as $n \rightarrow +\infty$. Hence there is an integer $N$ such that  $ (1-\pi)   V\left(\widehat{Q}_{\theta}, P_{\theta}\right) <   V\left(\widehat{Q}_{\theta}, P_{\theta}\right)<1/2$ for $ n >N$  $ P_{\theta}$-a.s..
Hence Equation~(\ref{cskorner})  entails 
$$
\left |   H\left(  \widehat{M}\right)  - H\left(  \widehat{M}_{\rm imp} \right)   \right | \leq   - V\left(\widehat{Q}_{\theta}, P_{\theta}\right)  \ln\left( \frac{ V\left(\widehat{Q}_{\theta}, P_{\theta}\right)}{k}\right).
$$
 Since $V \ln V  \rightarrow 0$, as $V \downarrow 0$, we get again  by Lemma  \ref{totvarlemma} that 
$$
\left |   H\left(  \widehat{M}\right)  - H\left(  \widehat{M}_{\rm imp} \right)   \right| \rightarrow 0, \quad  P_{\theta}  \mbox{-a.s.},
$$
as  $n \rightarrow +\infty$. 

The second term in the right hand side of Equation~(\ref{melcomp2245})  we have 
$$
 \sum_{j=1}^{k}\left |  \widehat{m}_{i,\rm imp} -  \widehat{m}_{i} \right | \ln \frac{1}{p_{j}(\theta)} \leq 
 \ln \left[ \frac{1}{ \min_{1 \leq j \leq k}p_{j}(\theta)}\right] V\left(  \widehat{M}_{\rm imp} , \widehat{M}\right), 
$$
where $\min_{1 \leq j \leq k}p_{j}(\theta) >0$ by the assumption in Equation~(\ref{posass}).
Hence, by  the computations done  for the first term in the right hand side of Equation~(\ref{melcomp2245}),  the sum 
in the left hand side of the inequality above  converges to $0$, $P_{\theta}$- a.s., as $n \rightarrow +\infty$  by  Lemma \ref{totvarlemma}. When these results are used in Equation~(\ref{melcomp22}), the asserted convergence follows as claimed.
\end{proof}

\begin{example}\label{bern11} {\rm We continue with  example  \ref{berno}.   $\widehat{Q}_{\theta} \in \mathbb{M}_{n}(\theta)$, so that 
$$
\widehat{Q}_{\theta}(x)  = \left (1-\sum_{i=1}^{n} \frac{ \xi_{i}}{n}\right)^{[x=0]} \left (\sum_{i=1}^{n} \frac{ \xi_{i}}{n}\right)^{[x=1]}, \quad x \in \{0,1 \}, 
$$
where  $\xi_{i}$ $\sim Be(\theta)$. 
We have by Equation~(\ref{jssymmen})
 \begin{equation}\label{jssymmen22} 
D_{\rm JS,1/2}( \widehat{P}_{\mathbf{D}},\widehat{Q}_{\theta})  = B \left( \frac{1}{2} \widehat{p} + \frac{1}{2}\sum_{i=1}^{n}\frac{ \xi_{i}}{n}\right) - \frac{1}{2} B\left( \sum_{i=1}^{n}\frac{ \xi_{i}}{n}\right) - \frac{1}{2}B( \widehat{p}).
 \end{equation}
When  $n$  is  large enough, Lemma \ref{totvarlemma} predicts  that
$\frac{1}{n}\sum_{i=1}^{n}  \xi_{i}  \approx \theta$.
By the proposition above,  (or as the binary  entropy function $B(\pi)$  is a  continuous function of a single variable), we get in  Equation~(\ref{jssymmen22})
 \begin{equation}\label{jssymmen222} 
D_{\rm JS,1/2}( \widehat{P}_{\mathbf{D}},\widehat{Q}_{\theta}) \rightarrow  B \left( \frac{1}{2} \widehat{p} + \frac{1}{2} \theta \right) - \frac{1}{2} B\left( \theta\right) - \frac{1}{2}B( \widehat{p})= D_{\rm JS,1/2}( \widehat{P}_{\mathbf{D}}, P_{\theta}). 
 \end{equation}
 Let us set $ -\frac{1}{2n_{o}}\ln {\cal L}_{\mathbf{D}}(\theta) =$ $  - \frac{1}{2} \left[ \widehat{p} \ln \theta +  (1-\widehat{p}) \ln(1-\theta)   \right] $. By maximum likelihood for ${\rm Be}(\theta)$,  $ -\frac{1}{2n_{o}}\ln {\cal L}_{\mathbf{D}}(\theta) \geq  -\frac{1}{2n_{o}}\ln {\cal L}_{\mathbf{D}}\left( \widehat{p} \right) =\frac{1}{2} B\left( \widehat{p} \right)$. Hence, by the right hand side of Equation~(\ref{jssymmen222}),
$ D_{\rm JS,1/2}\left( \widehat{P}_{\mathbf{D}}, P_{\widehat{p}}\right)= 0$. Hence,  
if $\theta$ is in a small neighborhood of  $ \widehat{p}$, 
$
D_{\rm JS,1/2}( \widehat{P}_{\mathbf{D}},\widehat{Q}_{\theta}) $  $\approx   0
$. This says that the simulator-based JSD statistic  $D_{\rm JS,1/2}( \widehat{P}_{\mathbf{D}},\widehat{Q}_{\theta})$ is approximately minimized with a high simulator probability.} \end{example}
\begin{theorem}\label{misspec11}
 Assume that Equation~(\ref{posass}) holds for
$P_{o} \in  \mathbb{P}$ and for  any $P_{\theta} \in \mathbb{M}$. Let  $ {\mathbf D}= (D_{1}, \ldots, D_{n_{o}}) $ be    an i.i.d. $n_{o}$-sample  $\sim$ $P_{o}$. Then  it holds that 
\begin{equation}\label{misspec1}
\lim_{n_{o} \rightarrow +\infty} D_{\rm JS} \left( \widehat{P}_{\mathbf D}, P_{\theta}\right) =  D_{\rm JS} \left( P_{o}, P_{\theta}\right). 
\end{equation}
$P_{o}$-a.s.. 
\end{theorem}
The proof is found  in  Appendix  \ref{stokka}.  Proposition \ref{misspec12} in the Appendix \ref{nasymptotiksec} is more general, as it covers two cases $i)$  $P_{o}=P_{\theta_{o}} \in \mathbb{M}$ and 
 $ii)$  $P_{o} \notin \mathbb{M}$.

By  arguments  similar to those used above we prove next a  continuity property of  JSD. 
\begin{theorem}\label{nskonv3}
Assume  Equation~(\ref{posass})  for  every  $P_{\theta} \in \mathbb{M}$. Assume that   $p_{j}(\theta)$  are   a continuous functions of $\theta \in \Theta$ for $j=1, \ldots,k$. Let    Then   $D_{\rm JS}\left(\widehat{P}_{\mathbf{D}},  P_{\theta}   \right)$ is a continuous function of 
$\theta \in \Theta$.
\end{theorem}
\begin{proof} Again by  Equation~(\ref{melcomp2}) in  Lemma \ref{mellcomp0} we obtain 
\begin{eqnarray}\label{melcomp2233} 
\left | D_{\rm JS}\left(\widehat{P}_{\mathbf{D}}, P_{\theta}   \right) -  D_{\rm JS} \left(  \widehat{P}_{\mathbf{D}},  P_{\theta_{o}} \right)\right | &\leq &    (1-\pi)D_{\rm KL}(P_{\theta},  P_{\theta_{o}}) \nonumber \\ 
&  & \\
& +&  \left |  D_{\rm KL} \left( M_{\theta},  P_{\theta} \right)  -   D_{\rm KL}\left(M_{\theta_{o}},  P_{\theta_{o}} \right)\right |, \nonumber
\end{eqnarray} 
where $M_{\theta}= \pi \widehat{P}_{\mathbf{D}} + (1-\pi) P_{\theta}$ and  $M_{\theta_{o}}= \pi \widehat{P}_{\mathbf{D}} + (1-\pi) P_{\theta_{o}}$.
The reverse Pinsker inequality $D_{\rm KL} (P,Q) \leq \frac{1}{Q_{min}}V( P, Q)^{2}$, cf. the discussion of  Equation~(\ref{sasonpinsker}), 
yields 
\begin{equation}\label{sasonverdu}
D_{\rm KL}(P_{\theta},  P_{\theta_{o}}) \leq  \frac{1}{\min_{j} p_{j}\left(\theta_{o}\right)}V\left(  P_{\theta} , P_{\theta_{o}} \right)^{2}.
\end{equation}
 As in the preceding proofs, counting here even   Appendix  \ref{stokka}, 
\begin{eqnarray}\label{melcomp5} 
\left |  D_{\rm KL} \left( M_{\theta},  P_{\theta} \right)  -   D_{\rm KL}\left(M_{\theta_{o}},  P_{\theta}\right)\right | &\leq &          \left |   H\left(  M_{\theta} \right)  - H\left( M_{\theta_{o}} \right)   \right | \nonumber \\ 
&  & \\
& +& \left| \sum_{j=1}^{k} m_{j,{\theta}}  \ln \frac{1}{p_{j}(\theta)} - \sum_{j=1}^{k}m_{j,{\theta_{o}}}  \ln \frac{1}{p_{j}\left(\theta_{o}\right)} \right|,  \nonumber
\end{eqnarray} 
where we used Equation~(\ref{posass}). Here we resort  again  to Lemma \ref{lemmakontbd}.  We have 
$$
V\left( M_{\theta}, M_{\theta_{o}} \right)  = (1-\pi)  V \left(P_{\theta}   , P_{\theta_{o}} \right).
$$
Thus Equation~(\ref{cskorner})  entails  for $ V\left( P_{\theta}, P_{\theta_{o}}\right)  <1/2$     that 
\begin{equation}\label{bound2}
\left |   H\left(  M_{\theta} \right)  - H\left( M_{\theta_{o}} \right)   \right |  \leq   - V\left(P_{\theta}, P_{\theta_{o}}\right)  \ln\left( \frac{  V\left(P_{\theta}, P_{\theta_{o}}\right) }{k}\right).
\end{equation}
For the second term in the right hand side of Equation~(\ref{melcomp5}) we have  by definitions of  $M_{\theta}$ and  $M_{\theta_{o}}$
$$
  \sum_{j=1}^{k} m_{j,{\theta}}  \ln \frac{1}{p_{j}\left(\theta\right)} - \sum_{j=1}^{k}m_{j,{\theta_{o}}}  \ln \frac{1}{p_{j}\left(\theta_{o}\right)} = 
$$
$$
\pi  \sum_{j=1}^{k}  \widehat{p}_{j} \ln \frac{p_{j}\left(\theta_{o} \right) }{p_{j}(\theta)} + 
(1-\pi) \left[  H\left( P_{\theta_{o}} \right)   - H\left( P_{\theta} \right)  \right].
$$
Hence 
$$
\left| \sum_{j=1}^{k} m_{j,{\theta}}  \ln \frac{1}{p_{j}(\theta)} - \sum_{j=1}^{k}m_{j,{\theta_{o}}}  \ln \frac{1}{p_{j}\left(\theta_{o}\right)} \right| \leq  
$$
$$
\pi  \sum_{j=1}^{k}  \widehat{p}_{j}  \left| \ln \frac{p_{j}\left(\theta_{o} \right) }{p_{j}\left(\theta\right)} \right|  + 
(1-\pi)   \left|  H\left( P_{\theta} \right)  - H\left( P_{\theta_{o}} \right) \right|. 
$$
Once more, if  $  V\left(P_{\theta}, P_{\theta_{o}}\right)  <1/2$,   Equation~(\ref{cskorner})  implies  
\begin{equation}\label{bound3}
   \left| H\left( P_{\theta_{o}} \right)   - H\left( P_{\theta} \right) \right| \leq -V\left(P_{\theta}, P_{\theta_{o}}\right)  \ln\left( \frac{  V\left(P_{\theta}, P_{\theta_{o}}\right) }{k}\right).
\end{equation}
Furthermore, 
$$
 \sum_{j=1}^{k}  \widehat{p}_{j}  \left| \ln \frac{p_{j}\left(\theta_{o} \right) }{p_{j}\left(\theta\right)} \right|=  
 \sum_{j=1}^{k}  \widehat{p}_{j}  \left| \ln\left[ 1  + \frac{p_{j}\left(\theta_{o} \right)-p_{j}\left(\theta\right) }{p_{j}\left(\theta\right)}\right] \right|.
$$ 
Since  $ \widehat{p}_{j}  <1$ and $x=(p_{j}\left(\theta_{o} \right)-p_{j}\left(\theta\right) )/ p_{j}\left(\theta\right)> -1$, so that $\ln(1+x) \leq x$,  the right hand side is bounded upwards by 
\begin{equation}\label{bound4}
\leq   \frac{1}{ \min_{ 1 \leq j \leq k} p_{j}\left(\theta\right) } \sum_{j=1}^{k} \left|p_{j}\left(\theta\right)  - p_{j}\left(\theta_{o} \right) \right|. 
\end{equation}
By Equations~(\ref{sasonverdu}),  (\ref{bound2}),  (\ref{bound3}) and  (\ref{bound4})  we have
\begin{eqnarray}\label{melcomp4} 
\left | D_{\rm JS}\left(\widehat{P}_{\mathbf{D}}, P_{\theta}   \right) -  D_{\rm JS} \left(  \widehat{P}_{\mathbf{D}},  P_{\theta_{o}} \right)\right | &\leq &    \frac{ (1-\pi) }{\min_{j} p_{j}\left(\theta_{o}\right)}V\left(  P_{\theta} , P_{\theta_{o}} \right)^{2} \nonumber \\ 
& -&  V\left(P_{\theta}, P_{\theta_{o}}\right)  \ln\left( \frac{  V\left(P_{\theta}, P_{\theta_{o}}\right) }{k}\right). \nonumber \\
& & \\
& - & (1-\pi)V\left(P_{\theta}, P_{\theta_{o}}\right)  \ln\left( \frac{  V\left(P_{\theta}, P_{\theta_{o}}\right) }{k}\right) \nonumber  \\  
& +&  \frac{1}{ \min_{ 1 \leq j \leq k} p_{j}\left(\theta\right) } \sum_{j=1}^{k} \left|p_{j}\left(\theta\right)  - p_{j}\left(\theta_{o} \right) \right|. \nonumber 
\end{eqnarray} 
We have  $V\left(  P_{\theta} , P_{\theta_{o}} \right)$ $  = \sum_{j=1}^{k}\left|p_{j}(\theta)-  p_{j}\left(\theta_{o} \right) \right|$.  Let next   $||\theta - \theta_{o}||$   be  any norm on $\Theta$,  all  norms in a finite dimensional space are equivalent, as is well known,  see, e.g.,   \citet{johnson2012notes}.  Since each  $p_{j}(\theta)$ is a  continuous functions of  $\theta$,  there exists  for every  $\epsilon >0$   a  $\delta_{1,\epsilon} >0$, such that as soon as  $||\theta - \theta_{o}|| < \delta_{1,\epsilon}$,
$ \frac{ (1-\pi) }{\min_{j} p_{j}\left(\theta_{o}\right)}V\left(  P_{\theta} , P_{\theta_{o}} \right)^{2}  \leq  \epsilon/4$.  By the same argument there exist  $\delta_{2,\epsilon}$ such that as soon as  $||\theta - \theta_{o}|| < \delta_{2,\epsilon}$,
$ 0 \leq   -  V\left(P_{\theta}, P_{\theta_{o}}\right)  \ln\left( \frac{  V\left(P_{\theta}, P_{\theta_{o}}\right) }{k}\right)
\leq  \epsilon/4$.  By continuing this line we find   $\delta_{3,\epsilon}$  and $\delta_{4,\epsilon}$ such that 
the two last terms in the right hand side of Equation~(\ref{melcomp4}) are both bounded by   $\epsilon/4$ for $||\theta - \theta_{o}||$ correspondingly small.  As soon as $||\theta - \theta_{o}|| < \min \left(  \delta_{1,\epsilon},  \delta_{2,\epsilon},   \delta_{3,\epsilon}, \delta_{4,\epsilon}\right)$ we get that 
$$
\left | D_{\rm JS}\left(\widehat{P}_{\mathbf{D}}, P_{\theta}   \right) -  D_{\rm JS} \left(  \widehat{P}_{\mathbf{D}},  P_{\theta_{o}} \right)\right |\leq \epsilon, 
$$
which proves the asserted continuity.
\end{proof}

 \section{The  Existence of the Minimum  JSD Estimate}\label{esitmateftn}

The preceding section leads to the study of  the minimum JSD estimate of $\theta$,
$$
\widehat{\theta}_{\rm JS}\left(\mathbf{D} \right)= {\rm argmin}_{\theta \in \Theta} D_{\rm JS}( \widehat{P}_{\mathbf{D}}, P_{\theta}).
$$
This is a special case of   the minimum $\phi$-divergence estimate  treated,  e.g., in 
\citet[Ch. 5.1--5.3]{pardo2018statistical}.    The minimum $\phi$-divergence estimate for discrete (incl. categorical) distributions is  studied  in \citet{morales1995asymptotic} and 
  \citet[pp. 388$-$391]{vajda1989theory}. The existence and measurability  result 
in the next proposition is not found in  \citet{morales1995asymptotic}. It is shown in \citet  {corremkos}  that 
$\widehat{\theta}_{\rm JS}\left(\mathbf{D} \right)$  and the maximum likelihood estimate $ \widehat{\theta}_{\rm ML}\left(\mathbf{D} \right)$ agree  asymptotically, as $n \rightarrow +\infty$.  

In information  theory   $\widehat{P}_{\mathbf{D}} $     is  called the type of $ {\mathbf D}$ on ${\cal A}$, see  
  \citet[Part I Ch.2]{csiszar2011information} and   \citet[Ch. 11.1]{cover2012elements}. 
  The type class of $ \widehat{P}_{\mathbf{D}}$  is defined , see   \citet[Ch. 11.1--11.3]{cover2012elements},  
  by  
\begin{equation}\label{typeclass}
{\cal T}_{n}\left(\widehat{P}_{\mathbf{D}}\right) := \{ {\mathbf X}= (X_{1}, \ldots,X_{n}) \in{\cal A}^{n} \mid  \widehat{P}_{\mathbf {X}}=  \widehat{P}_{\mathbf{D}} \}.
\end{equation} 
The set of  all types on  ${\cal  A}$ for $n$ samples
\begin{equation}\label{typesone}
{\cal P}_{n} := \left\{  P \in  \mathbb{P} \mid  {\cal T}_{n}\left( P \right)  \neq \emptyset \right\}.
\end{equation}
  The cardinality of the set of types is  $|{\cal P}_{n}|$ $=\left(\begin{array}{c} n+k-1 \\ k-1\end{array} \right)$  by a well known  combinatorial argument.
If  $S$ is   any  partition of  ${\cal P}_{n}$, then  it is   a sigma-field  ${\bf  P}_{n}:=S $ on ${\cal P}_{n}$ , i.e.,    $\left({\cal P}_{n}, {\bf{ P}}_{n}  \right)$
is a measurable space. \\

We require   the  following assumption in   \citet[(B), p.~817]{birch1964new}  known as the strong  identifiability  condition.  This is 
a mathematical expression for a property required of  the simulator  $\mathbf{M}_{C}$.  
Here $ {\bf p}(\theta)$ and ${\bf p}(\theta_{0})$  are given by the map in Equation~(\ref{trmap}) and  $|| {\bf x}||_{2, \mathbf{R}^{k}} =\sqrt{\bf{x} {\bf x}^{T}}$ is  the Euclidean norm on $ \mathbf{R}^{k}$.   The following assumption is   \citet[(B), p.~817]{birch1964new} and is known as the strong  identifiability  condition:  for any $\epsilon >0$ there exists $\delta >0$ such that, 
\begin{equation}\label{birchinvers}
\mbox{if } \quad   || {\bf p}(\theta)- {\bf p}(\theta_{0})||_{2, \mathbf{R}^{k}} > \epsilon,  \mbox{ then}  \quad 
 ||\theta - \theta_{0}||_{2, \mathbf{R}^{d}} > \delta.
\end{equation}
Clearly this implies the weak  identifiability  assumption
 \begin{equation}\label{identass} 
\theta \neq \theta^{'}  \Rightarrow  P_{\theta} \neq P_{\theta^{'}}. 
\end{equation} 
Under this assumption  $ {\bf p}_{\theta}= \triangle \left(P_{\theta} \right)$ is  a one-to-one  map between $\triangle\left( \mathbb{M}_{p} \right) $ and $\Theta$.

The next theorem  and its proof are based on the proof of   \citet[Lemma 2, p. 637]{jennrich1969asymptotic}.  The existence result in \citet[Thm 12.53 p.~391]{vajda1989theory} is less selfcontainced  and does not include measurability. 
\begin{theorem}\label{existenssats}
Assume that  $\Theta \subset {\mathbf{R}}^{d}$ is a compact set. Assume that   $p_{i}(\theta)$ in    $\mathbb{M}= \left\{ P_{\theta} \mid  \theta \in \Theta \right \}$  are continuous functions  on  $\Theta$.  Assume  identifiability as in Equation~(\ref{identass}). 
Then  there exists a ${\bf{ P}}_{n}  $ -measurable function  $\widehat{\theta}: {\cal P}_{n}  \mapsto \Theta$ 
such that  for every  $\widehat{P} \in {\cal P}_{n}$ 
\begin{equation}\label{minjsd}
 D_{\rm JS} \left(  \widehat{P},  P_{\widehat{\theta}\left(   \widehat{P} \right) }  \right)  =  
\inf_{\theta \in \Theta}   D_{\rm JS} \left(  \widehat{P},  P_{\theta}  \right).
\end{equation} 
\end{theorem} 
\begin{proof} Since  $\Theta $ is compact in $\mathbf{R}^{d}$,  there exists, see e.g.,  \citet[Ch. 14, pp.~88--90]{wu2017lecture} for any $\epsilon >0$  a finite  packing  set $ {\cal M}_{\epsilon}:= \{ \theta^{(1)}, \ldots, \theta^{(M(\epsilon))} \}$ with  the maximal  packing number  $M(\epsilon) < +\infty$, such that
$$
||\theta^{(i)} - \theta^{(j)}||_{2, \mathbf{R}^{d}} > \epsilon, \quad  \mbox{ for $i \neq j$}
$$
and    $\cap_{i=1}^{M(\epsilon)} B(\theta_{i},\epsilon/2)_{2, \mathbf{R}^{d}}=  \emptyset$.  By definition of  
 $M(\epsilon)$ it holds also that  for every $\theta \in \Theta$  there exists  $\theta^{(i)} \in {\cal M}_{\epsilon}$ such that 
$ ||\theta - \theta^{(i)}||_{2, \mathbf{R}^{d}} \leq \epsilon$, i.e., ${\cal M}_{\epsilon}$  
is also an $\epsilon$-covering, i.e.,  $\Theta   \subseteq \cup_{i=1}^{M(\epsilon)} B(\theta_{i},\epsilon)$, but  not necessarily a minimal such. 

Furthermore, $ \epsilon^{'} < \epsilon$  implies that $M(\epsilon) \leq M(\epsilon^{'})$  and since for   any two 
$\theta^{(i)}$ and $ \theta^{(j)}$ in  ${\cal M}_{\epsilon}$ with  $i \neq j$ it holds that  
$
||\theta^{(i)} - \theta^{(j)}||_{2, \mathbf{R}^{d}} > \epsilon^{'}$, we can include the points  of  ${\cal M}_{\epsilon}$ in ${\cal M}_{\epsilon^{'}}$  and thus    ${\cal M}_{\epsilon}$ is 
an increasing and dense sequence of finite sets. 

Let us  define  for  all $  \widehat{P}$  the map  $\bar{\theta}_{M(\epsilon)}\left(   \widehat{P}\right) $  on ${ \cal P}_{n}$ by
\begin{equation}\label{minjsdpac} 
 D_{\rm JS} \left(  \widehat{P},  P_{\bar{\theta}_{M(\epsilon)}\left(   \widehat{P}\right) }  \right)  =  
\inf_{ {\cal M}_{\epsilon} }   D_{\rm JS} \left(  \widehat{P},  P_{{\theta}^{(i)}} \right).
\end{equation} 
In this  $\bar{\theta}_{M(\epsilon)}$ is measurable, since  the sets  ${\bf { \cal P}}_{n}^{(i)}:=\{   \widehat{P} \in { \cal P}_{n}|\bar{\theta}_{M(\epsilon)}\left(   \widehat{P}\right)  = 
\theta^{(i)} \}$  form a partition of   ${ \cal P}_{n}$, and are thus in   ${\bf P}_{n}$. 
 
Let  $\bar{\theta}_{M(\epsilon),1}\left(    \widehat{P}\right)$ be the first of the $d$  components  of
 $\bar{\theta}_{M(\epsilon)}\left(\widehat{P}\right)$. Let  $\epsilon \downarrow 0$, and thus   $$\bar{\theta}_{1}\left(    \widehat{P}\right)=
\liminf_{M(\epsilon) \rightarrow +\infty  }\bar{\theta}_{M(\epsilon),1}\left(\widehat{P}\right).$$ Then  
 $\bar{\theta}_{1}\left(    \widehat{P}\right)$ is measurable.  Then for  every $ \widehat{P} \in {\cal P}_{n}$ 
there exists  due to  compactness  of $\Theta$  a subsequence $\bar{\theta}_{M(\epsilon_{i})}\left(\widehat{P}\right)$, which converges to  
$ \theta^{\ast}(\widehat{P}) \in \Theta$  of the form 
$$\theta^{\ast}(\widehat{P}) = \left( \bar{\theta}_{1}\left(    \widehat{P}\right),  \bar{\theta}^{\ast}_{2} , \ldots, \bar{\theta}^{\ast}_{d}\right).$$
Then 
\begin{equation}\label{delmal}
 \inf_{ {\cal M}_{\epsilon} }   D_{\rm JS} \left(  \widehat{P},  P_{ \left( \bar{\theta}_{1}\left(    \widehat{P}\right),  {\theta}_{2} , \ldots,  \theta_{d} \right)}  \right) \leq      D_{\rm JS} \left(  \widehat{P},  P_{\theta^{\ast}(\widehat{P}) }  \right). 
\end{equation}
Since  $p_{i}(\theta)$ in    are continuous functions, Theorem  \ref{nskonv3} implies 
$$
 D_{\rm JS} \left(  \widehat{P},  P_{\theta^{\ast}(\widehat{P}) }  \right) = \lim_{\epsilon_{i} \downarrow 0} D_{\rm JS} \left(  \widehat{P},  P_{\bar{\theta}_{M(\epsilon_{i})}\left(\widehat{P}\right)}  \right) 
$$
and by definition,  $D_{\rm JS} \left(  \widehat{P},  P_{\bar{\theta}_{M(\epsilon_{i})}\left(\widehat{P}\right)}  \right) $ $ = \inf_{{\cal M}_{\epsilon_{i}}}  D_{\rm JS} \left(  \widehat{P},  P_{\theta}  \right)$. Hence 
$$
 \lim_{\epsilon_{i} \downarrow 0} D_{\rm JS} \left(  \widehat{P},  P_{\bar{\theta}_{M(\epsilon_{i})}\left(\widehat{P}\right)}  \right) =\lim_{\epsilon_{i} \downarrow 0} \inf_{{\cal M}_{\epsilon_{i}}}  D_{\rm JS} \left(  \widehat{P},  P_{\theta}  \right).
$$
However   the covering sets are increasing and dense  in $\Theta$ and thus  the right hand  side equals 
$$
 = \inf_{\theta \in \Theta}   D_{\rm JS} \left(  \widehat{P},  P_{\theta}  \right).
$$
In view of Equation~(\ref{delmal})  we have  thus shown that 
$$
\inf_{ {\cal M}_{\epsilon} }   D_{\rm JS} \left(  \widehat{P},  P_{ \left( \bar{\theta}_{1}\left(    \widehat{P}\right),  {\theta}_{2} , \ldots,  \theta_{d} \right)}  \right)  \leq   \inf_{\theta \in \Theta}   D_{\rm JS} \left(  \widehat{P},  P_{\theta}  \right),
$$
which means that 
$$
\inf_{ {\cal M}_{\epsilon} }   D_{\rm JS} \left(  \widehat{P},  P_{ \left( \bar{\theta}_{1}\left(    \widehat{P}\right),  {\theta}_{2} , \ldots,  \theta_{d} \right)}  \right) = \inf_{\theta \in \Theta}   D_{\rm JS} \left(  \widehat{P},  P_{\theta}  \right).
$$ 
We  consider the map 
$$
   D_{\rm JS} \left(  \widehat{P},  P_{\left( \bar{\theta}_{1}\left(    \widehat{P}\right),  {\theta}_{2} , \ldots,  \theta_{d}  \right) }  \right)
$$
and  obtain $\bar{\theta}_{2}\left(    \widehat{P}\right)$ as above, and then repeat the subsequent computation.  When we continue in this manner component by component we get a measurable function such that 
$$
   D_{\rm JS} \left(  \widehat{P},  P_{ \left( \bar{\theta}_{1}\left(    \widehat{P}\right),  \bar{\theta}_{2}\left(    \widehat{P}\right), , \ldots,   \bar{\theta}_{d}\left(    \widehat{P}\right)\right) } \right) = \inf_{\theta \in \Theta}   D_{\rm JS} \left(  \widehat{P},  P_{\theta}  \right),
$$ 
which finishes the proof.
\end{proof}
The proposition proves  the existence of  $\inf_{\theta \in \Theta}   D_{\rm JS} \left(  \widehat{P}_{\mathbf{D}},  P_{\theta}  \right)$ as  a measurable map $\widehat{\theta}: {\cal P}_{n_{o}}  \mapsto \Theta$.

\section{ Taylor Expansions of the JSD statistic}\label{taylor}

In this section the Taylor expansions are performed on  real valued functions of  vectors  in  the interior $\stackrel{o}{\triangle}_{k-1}$. For   $ P \in \mathbb{P}  $  we consider  $ {\bf p}= \triangle  \left(P \right) \in \triangle_{k-1}$, see Equation~(\ref{trmap}),  and   
 for and   $ Q \in \mathbb{P}$  we have  $ {\bf q}=  \triangle \left(Q \right) \in  \triangle_{k-1}$. The map $D_{\rm JS} \left(  {\bf p},  {\bf q}   \right) =   D_{\rm JS}(P,Q)$ is then defined on  $\triangle_{k-1}  \times \triangle_{k-1}$. 
In addition,  $\widehat{{\bf p}}=\triangle\left( \widehat{P}_{\mathbf{D}}\right) $,  $\widehat{{\bf q}}_{\theta}= \triangle\left(\widehat{Q}_{\theta}\right) $ and   $ {\bf p}_{\theta}=  \triangle\left(P_{\theta}\right)$.  
By  Equation~(\ref{vertexprop})  we have  $ {\bf p}_{\theta} = E_{P_{\theta}}\left[  \widehat{{\bf q}}_{\theta} \right] $ .
Hence we have 
$$
 D_{\rm JS}\left( \widehat{{\bf p}} , \widehat{{\bf q}}_{\theta}
\right)  - D_{\rm JS} \left( \widehat{{\bf p}} ,  E_{P_{\theta}}\left[  \widehat{{\bf q}}_{\theta} \right] \right)= 
 D_{\rm JS}\left(\widehat{P}_{\mathbf{D}}, \widehat{Q}_{\theta}   \right)  - D_{\rm JS} \left(  \widehat{P}_{\mathbf{D}},  P_{\theta}\right).
$$
 Here 
$ \widehat{{\bf p}}$ is  held constant in  $\stackrel{o}{\triangle}_{k-1}$  and we expand w.r.t.  $ \widehat{{\bf q}}_{\theta}$ in 
an open neighborhood around $ {\bf p}_{\theta} = E_{P_{\theta}}\left[  \widehat{{\bf q}}_{\theta} \right] $.
 An open  neighborhood  is  in the topology of $\mathbb{R}^{k-1}$. 

We  find first  the partial derivatives  $\frac{\partial}{\partial q_{j}}D_{\rm JS} \left(  {\bf p},  {\bf q}   \right)$, $\frac{\partial^{2}}{\partial q_{j}^{2}} $  and  mixed partial derivatives  $\frac{\partial^{2}}{\partial q_{i} \partial q_{j}}D_{\rm JS} \left(  {\bf p},  {\bf q}   \right)$ for  $j=1, \ldots, k$  and $i=1, \ldots, k$. 

With the aid of  this expansion, we provide a uniform probability bound for the  distance between the simulator-based JSD statistic and 
its limiting JSD statistic, when the number of  synthetic samples increases. 

\subsection{  Taylor Expansions for Decomposable Functions on $\triangle_{k-1}  \times \triangle_{k-1}$} 

It is practical  to start by   formal partial  differentiations the general  $\phi$-divergence   in Equation~(\ref{informationone}),
$
D_{\phi}( P, Q) =D_{\phi}\left(  {\bf p},  {\bf q}   \right) = \sum_{i=1}^{k} q_{i} \phi\left( \frac{ p_{i}}{q_{i}}\right)$.    This is now  regarded  as  a function of  ${\bf q}$ for a fixed {\bf p}. Thereby   $D_{\phi}\left(  {\bf p},  {\bf q}   \right)$   is a decomposable function for the purposes of differentiation w.r.t.  $q_{j}$,  in the sense that  $D_{\phi}\left(  {\bf p},  {\bf q}   \right)= 
 \sum_{i=1}^{k} \Phi_{j}\left( q_{j}\right)$, where  $ \Phi_{j}\left( q_{j}\right) = q_{j} \phi\left( \frac{ p_{j}}{q_{j}}\right)$. 

First, 
\begin{equation}\label{forstaderiphi}
\frac{\partial}{\partial q_{j}} D_{\phi}\left(  {\bf p},  {\bf q}   \right) = \frac{\partial}{\partial q_{j}}  \Phi_{j}\left( q_{j}\right) = \phi\left(\frac{ p_{j}}{q_{j}}\right) - \frac{ p_{j}}{q_{j}}\phi^{(1)}\left(\frac{ p_{j}}{q_{j}}\right),
\end{equation}
$$
\frac{\partial^{2}}{\partial q_{j}^{2}} D_{\phi}\left(  {\bf p},  {\bf q}   \right)  = - \frac{ p_{j}}{q_{j}^{2}}\phi^{(1)}\left(\frac{ p_{j}}{q_{j}}\right) +\frac{ p_{i}}{q_{i}^{2}}\phi^{(1)}\left(\frac{ p_{i}}{q_{i}}\right)  + \frac{ p^{2}_{j}}{q_{j}^{3}}\phi^{(2)}\left(\frac{ p_{j}}{q_{j}}\right),
$$
i.e., 
\begin{equation}\label{andraderivphi}
\frac{\partial^{2}}{\partial q_{j}^{2}} D_{\phi}\left(  {\bf p},  {\bf q}   \right)  =  \frac{ p^{2}_{j}}{q_{j}^{3}}\phi^{(2)}\left(\frac{ p_{j}}{q_{j}}\right) . 
\end{equation} 
All mixed partial derivatives are zero. 
Let us now apply these to $\phi_{\rm JS}$  in Equation~(\ref{jenshaphi}).  
We have  for  $x \in (0,1)$  the first derivative
\begin{equation}\label{jsddivder1}
\phi^{(1)}_{\rm JS}(x)  =\pi \ln \left[\frac{x}{\pi x +(1-\pi)}  \right]
\end{equation}
and the second derivative
\begin{equation}\label{jsddivder3}
\phi^{(2)}_{\rm JS}(x)  = \frac{\pi (1-\pi)}{x (\pi x +(1-\pi))}.
\end{equation}
This checks also the strict convexity  at $x=1$. 
Let us  observe that
\begin{equation}\label{jsddivder4}
\phi_{\rm JS}(x)= x \phi^{(1)} _{\rm JS}(x)-(1-\pi) \ln \left(\pi x  + (1-\pi) \right).
\end{equation} 
When  we use  Equations~(\ref{jsddivder1}),(\ref{jsddivder3})  and (\ref{jsddivder4}) in  Equation~(\ref{forstaderiphi}) with $x=p_{j}/q_{j}$ we obtain 
with  $m_{j} = \pi p_{j} + (1-\pi) q_{j}$
\begin{equation}\label{forstaderivjsd}
\frac{\partial}{\partial q_{j}} D_{\phi_{\rm JS}}\left(  {\bf p},  {\bf q}   \right) =(1-\pi) 
  \ln \left( q_{j}/m_{j}\right), 
\end{equation}
and 
\begin{equation}\label{andraderivjsd}
 \frac{\partial^{2}}{\partial q_{j}^{2}} D_{\phi_{\rm JS}}\left(  {\bf p},  {\bf q}   \right) = \frac{ p_{j}}{q_{j}} \frac{\pi (1-\pi)}{m_{j}}.
\end{equation}
Obviously these partial derivatives are defined only for  $ {\bf q} \in  \stackrel{o}{\triangle}_{k-1}$. 

Since $ D_{\rm JS}\left( \widehat{{\bf p}} , {\bf q} \right)$  is three times continuously differentiable w.r.t. ${\bf q}$  in an open neighborhood of  $  E_{P_{\theta}}\left[  \widehat{{\bf q}}_{\theta}\right] \in \stackrel{o}{\triangle}_{k-1}$, we have the second  order Taylor series omitting the remainder term with third order partial  derivatives, see, e.g.,  \citet[Thm 4.7, p.~179]{lax2017multivariable}.  
\begin{eqnarray}\label{taylorjsdt1}
 D_{\rm JS}\left( \widehat{{\bf p}} , \widehat{{\bf q}}_{\theta}
\right)  - D_{\rm JS} \left( \widehat{{\bf p}} ,  E_{P_{\theta}}\left[  \widehat{{\bf q}}_{\theta} \right] \right)& =&\sum_{j=1}^{k} \frac{\partial}{\partial q_{j}} D_{\phi_{\rm JS}}\left(   \widehat{{\bf p}},   E_{P_{\theta}}\left[  \widehat{{\bf q}}_{\theta}  \right] \right)\left( \widehat{q}_{i,\theta} - p_{j}(\theta) \right) \nonumber \\
& &   \\
& &  +\frac{1}{2!}\sum_{j=1}^{k} \frac{\partial^{2}}{\partial q_{j}^{2}} D_{\phi_{\rm JS}}\left( \widehat{{\bf p}},    E_{P_{\theta}}\left[  \widehat{{\bf q}}_{\theta} \right]    \right) \left( \widehat{q}_{j,\theta} - p_{j}(\theta) \right)^{2}. \nonumber
\end{eqnarray}
\subsection{A Uniform  Probability Bound }\label{unifcocn}
The proposition in this section is a  kind  concentration  inequality,  cf.     \citet[Eq. (17),  p.~259]{massart2000some}, for  $D_{\rm JS}\left( \widehat{P}_{\mathbf{D}}, \widehat{Q}_{\theta} \right)$, meaning  that
 $D_{\rm JS,1/2}\left( \widehat{P}_{\mathbf{D}}, \widehat{Q}_{\theta} \right)$ for any $n$ closely concentrated around its mean for any $n$. 
 \begin{theorem}\label{unifns}
Assume that Equation~(\ref{posass}) holds for $P_{\theta}$ and that $\Theta$ is a compact subset of  $\mathbb{R}^{d}$.   Then there exists 
a positive number  $ K_{u}$ such that    for any $\epsilon >0$
\begin{equation}\label{unif2}
P_{\theta}\left(  \mid   D_{\rm JS}\left(\widehat{P}_{\mathbf{D}}, \widehat{Q}_{\theta}   \right) - D_{\rm JS} \left(  \widehat{P}_{\mathbf{D}},  P_{\theta}\right) \mid  \geq \epsilon \right) \leq   2k e^{-2 n\epsilon^{2}/\left( K_{u}^{2}k^{3} \right)}
\end{equation} 
\end{theorem}
\begin{proof}  
Let us expand 	up to the first  order. Then Equation~(\ref{taylorjsdt1})  and \citet[Thm 4.7, p.~179]{lax2017multivariable} entail 
  \begin{eqnarray}\label{taylorjsdt1and}
 D_{\rm JS}\left( \widehat{{\bf p}} , \widehat{{\bf q}}_{\theta}
\right)  - D_{\rm JS} \left( \widehat{{\bf p}} ,  E_{P_{\theta}}\left[  \widehat{{\bf q}}_{\theta} \right] \right)& =&\sum_{j=1}^{k} \frac{\partial}{\partial q_{j}} D_{\phi_{\rm JS}}\left(   \widehat{{\bf p}},   E_{P_{\theta}}\left[  \widehat{{\bf q}}_{\theta}  \right] \right)\left( \widehat{q}_{i,\theta} - p_{j}(\theta) \right) \nonumber \\
& &  \nonumber  \\
& &  +\frac{1}{2!}\sum_{j=1}^{k} \frac{\partial^{2}}{\partial q_{j}^{2}} D_{\phi_{\rm JS}}\left( \widehat{{\bf p}},    \tilde{\bf q}    \right) \left( \widehat{q}_{j,\theta} - p_{j}(\theta) \right)^{2}, \nonumber
\end{eqnarray}
where $\tilde{\bf q}= E_{P_{\theta}}\left[  \widehat{{\bf q}}_{\theta} \right]   +  \tau  \widehat{{\bf q}}_{\theta} $ 
for some $\tau \in(0,1)$.  Hence by Equations~(\ref{forstaderivjsd}) and  (\ref{andraderivjsd})
  \begin{eqnarray}\label{taylorjsdt1andra}
\left| D_{\rm JS}\left( \widehat{{\bf p}} , \widehat{{\bf q}}_{\theta}
\right)  - D_{\rm JS} \left( \widehat{{\bf p}} ,  E_{P_{\theta}}\left[  \widehat{{\bf q}}_{\theta} \right] \right) \right|& \leq &(1-\pi)  \sum_{j=1}^{k}  K_{j,1}(\theta)  \left | \widehat{q}_{i,\theta} - p_{j}(\theta) \right|  \nonumber \\
& &   \\
&+ & \frac{\pi (1-\pi)}{2!}\sum_{j=1}^{k}K_{j,2}(\theta)  \left( \widehat{q}_{j,\theta} - p_{j}(\theta) \right)^{2}. \nonumber
\end{eqnarray}
For space reasons,   $K_{j,1}(\theta) := \left| \ln \left(  p_{j}(\theta)/(\pi \widehat{p}+ (1-\pi) p_{j}(\theta)\right)\right|$	and   $K_{j,2}(\theta) :=\left|  \frac{\widehat{ p}_{j}}{\tilde{p}_{j}(\theta)} /\left(\pi \widehat{p} +(1-\pi) \tilde{p}_{j}(\theta)\right)\right| $, where  $  \tilde{p}_{j}(\theta) =  p_{j}(\theta)  +  \tau \widehat{ q}_{\theta,j}$. 	 	
  By the assumptions of the proposition there exist two finite constants $K_{1}$ and $K_{2}$ such that 
$$
\max_{j \in \{1, \ldots, k \}}\max_{\theta \in \Theta} \left| \ln \left(  p_{j}(\theta)/(\pi \widehat{p}+ (1-\pi) p_{j}(\theta))\right)\right|\leq K_{1}
$$
and 
$$
\max_{j \in \{1, \ldots, k \}}\max_{\theta \in \Theta}\left|  \frac{\widehat{ p}_{j}}{\tilde{p}_{j}(\theta)} \frac{1}{\left(\pi \widehat{p} +(1-\pi) \tilde{p}_{j}(\theta)\right)}\right| \leq K_{2}.
$$
In  Equation~(\ref{taylorjsdt1andra})  these bounds give 
\begin{eqnarray}
\left| D_{\rm JS}\left( \widehat{{\bf p}} , \widehat{{\bf q}}_{\theta}
\right)  - D_{\rm JS} \left( \widehat{{\bf p}} ,  E_{P_{\theta}}\left[  \widehat{{\bf q}}_{\theta} \right] \right) \right|& \leq &(1-\pi) K_{1} \sum_{j=1}^{k}  \left | \widehat{q}_{i,\theta} - p_{j}(\theta) \right|  \nonumber \\
& &   \nonumber \\
& &  +\frac{\pi (1-\pi)K_{2}}{2!}\sum_{j=1}^{k} \left( \widehat{q}_{j,\theta} - p_{j}(\theta) \right)^{2}. \nonumber
\end{eqnarray}		
Here  $ \left( \widehat{q}_{j,\theta} - p_{j}(\theta) \right)^{2}$ $=  \left|\widehat{q}_{j,\theta} - p_{j}(\theta) \right|  \left| \widehat{q}_{j,\theta} - p_{j}(\theta) \right| \leq  2 \left | \widehat{q}_{j,\theta} - p_{j}(\theta) \right|$. Thus 
$$
\frac{\pi (1-\pi)K_{2}}{2!}\sum_{j=1}^{k} \left( \widehat{q}_{j,\theta} - p_{j}(\theta) \right)^{2} \leq  
\frac{2\pi (1-\pi)K_{2}}{2!} \left | \widehat{q}_{i,\theta} - p_{j}(\theta) \right|. 
$$
Let us set 
$$
K_{u}:=(1-\pi) \left(K_{1}  +\pi  K_{2}  \right).
$$
%$$
%D_{\rm JS}\left( \widehat{{\bf p}} , \widehat{{\bf p}}_{\theta}
%\right)  - D_{\rm JS} \left( \widehat{{\bf p}} ,  E_{P_{\theta}}\left[  \widehat{{\bf p}}_{\theta} \right] \right) 
%=\sum_{j=1}^{k} \frac{\partial}{\partial q_{j}} D_{\phi_{\rm JS}}\left(   \widehat{{\bf p}}, \tilde{\bf q} \right)\left( \widehat{q}_{i,\theta} - p_{j}(\theta) \right)
%$$
%where   $\tilde{\bf q}= E_{P_{\theta}}\left[  \widehat{{\bf p}}_{\theta} \right]   +  \tau  \widehat{{\bf p}}_{\theta} $. As in the preceding proof 
%this entails 
%$$ 
%\mid D_{\rm JS}\left( \widehat{{\bf p}} , \widehat{{\bf p}}_{\theta}
%\right)  - D_{\rm JS} \left( \widehat{{\bf p}} ,  E_{P_{\theta}}\left[  \widehat{{\bf p}}_{\theta} \right] \right) \mid 
%\leq K_{u} \sum_{j=1}^{k} \mid  \widehat{q}_{i,\theta} - p_{j}(\theta)  \mid 
%$$
%
Thus  
$$
P_{\theta}\left(  \mid   D_{\rm JS}\left(\widehat{P}_{\mathbf{D}}, \widehat{P}_{\theta}   \right) - D_{\rm JS} \left(  \widehat{P}_{\mathbf{D}},  P_{\theta}\right) \mid  \geq \epsilon \right) \leq  P_{\theta}\left(  \sum_{j=1}^{k} \mid  \widehat{q}_{i,\theta} - p_{j}(\theta)  \mid  \geq \epsilon/K_{u} \right).
$$
We have   $n\widehat{q}_{i} \stackrel{d}{=}  \xi_{i}$. Then the bound  in Equation~(\ref{unif2}) follows as in the proof of Lemma \ref{totvarlemma} .
  \end{proof}

\section{   Expectation and Variance  of the Simulator-Based JSD Statistic }\label{bernvantevarde}

In this section we  first  compute  the exact expectation of the simulator-based JSD statistic  by means of Bernstein polynomials 
and  an  approximate expression  by means  of   Voronovskaya$^{,}$s Asymptotic Formula for Bernstein polynomials. 
The approximate expression plays an important role in the $\chi^{2}$- analysis in the sequel.  

Then we compute the variance of the simulator-based JSD statistic  using Voronovskaya$^{,}$s Asymptotic Formula 
and the mean squared error  of the distance between the simulator-based JSD statistic and the corresponding limiting  JSD statistic. Here the recent results  
   in  \citet{ouimet2021general} on the moments of multinomial distributions are applied.
\subsection{Expectation and Bernstein Polynomials }

Let us  define 
\begin{equation}\label{andraderivator}
V_{\rm F}\left(   p_{i}\left( \theta\right),  \widehat{p}_{i} \right):=(1-\pi) \left[  \frac{1 }{ p_{i}\left( \theta\right)} -\frac{(1-\pi)}{ \pi   \widehat{p}_{i} +(1-\pi) p_{i}\left( \theta\right)}\right],  i=1, \ldots, k,
\end{equation}
and set 
\begin{equation}\label{andraderivatorsum}
V_{\rm F}\left(  P_{\theta},  \widehat{P} _{\mathbf{D}}\right):=\sum_{i=1}^{k}  p_{i}\left( \theta\right)\left(1- p_{i}\left( \theta\right)\right)  V_{\rm F}\left(   p_{i}\left( \theta\right),  \widehat{p}_{i} \right).
\end{equation}

\begin{theorem}\label{jsdvantevarde} 
Assume Equation~(\ref{posass})   for all $P_{\theta} \in \mathbb{M}$.  For $\pi 
\in [0,1] $ and any $\theta \in \Theta$ and $ \widehat{Q}_{\theta} \in \mathbb{M}_{n}(\theta) $, it holds that
\begin{equation} \label{slututryckvor3}
E_{P_{\theta}}\left[D_{\rm JS}(  \widehat{P}_{\mathbf{D}},  \widehat{Q}_{\theta}) \right]    =  D_{\rm JS} \left(  \widehat{P}_{\mathbf{D}},  P_{\theta} \right)    
+ \frac{1}{2n}V_{\rm F}\left(  P_{\theta},  \widehat{P} _{\mathbf{D}}\right) +  o(1/n).  
\end{equation}
\end{theorem} 
 The proof is given after   Lemma \ref{vantevarde} below. 
By Equations~(\ref{Qhat2}) and  (\ref{multinompid}) we re-write  Equation~(\ref{andraderivator}) as
\begin{equation}\label{varianssum}
\frac{1}{2n}V_{\rm F}\left(  P_{\theta},  \widehat{P} _{\mathbf{D}}\right):=  \frac{1}{2}\sum_{i=1}^{k}  {\rm Var} \left( \frac{\xi_{i}}{n}\right) V_{\rm F}\left(  p_{i}\left( \theta\right), \widehat{p}_{i} \right),
\end{equation}
since  $\xi_{i}  \sim {\rm Bin}(n, p_{i}\left(\theta \right))$. In addition, by  Equation~(\ref{andraderivatorsum})  we simplify the remainder as 
\begin{equation}\label{andraderivatorsum2}
V_{\rm F}\left(  P_{\theta},  \widehat{P} _{\mathbf{D}}\right):=(1-\pi)  \left[ (k-1) - (1-\pi)\sum_{i=1}^{k} \frac{p_{i}\left( \theta\right)\left( 1-p_{i}\left( \theta\right)\right)}{ (\pi   \widehat{p}_{i} +(1-\pi)p_{i}\left( \theta\right))}\right].
\end{equation}
%This would seem to advise against using the current method  applied to  very large ${\cal A}$, where   $k$  is not taken as known  with  certainty, cf., \citet{kelly2012classification}.

For proof of  Theorem  \ref{jsdvantevarde}  we use  Bernstein polynomials. 
Suppose that   $u(x)$ is  a continuous function on $[0,1]$. 
 $b(r;  x)$ denotes   the binomial probability of 
$r$ successes for an  r.v.  $\sim {\rm Bin}\left(n, x \right)$, $i=1,\ldots,k$. The Bernstein polynomial 
(or operator), see, e.g.,  \citet[p.~222]{gut2013probability},  ${\rm Be}_{n}$ is  defined by
\begin{equation}\label{bernst1}
{\rm Be}_{n}\left(u,  x\right):=\sum_{r=0}^{n}u \left(\frac{r}{n} \right) b(r;x).
\end{equation} 
If $u$ is   bounded,  and $\frac{d^{2}}{dx} u(x)$ is its second derivative  and   exists  for  some  $x \in [0,1]$,  then  Voronovskaya$^{,}$s Asymptotic Formula, see, e.g.,   \citet[Thm 2.2., p.~19]{gupta2014convergence} says 
\begin{equation}\label{bernstvo}
{\rm Be}_{n}\left(u, x \right)= u(x) + \frac{x(1-x)}{2n}\frac{d^{2}}{dx} u(x) + o(1/n). 
\end{equation} 
Voronovskaya$^{,}$s formula    will  by means of  Equation~(\ref{slututryck}) be applied to  express  $E_{P_{\theta}}\left[D_{\rm JS}(  \widehat{P}_{\mathbf{D}},  \widehat{Q}_{\theta}) \right] $ as in  Equation~(\ref{slututryckvor3}).

\begin{lemma}\label{vantevarde}
Assume Equation~(\ref{posass})   for $P_{\theta} \in \mathbb{M}$. For $x \in [0,1]$ and $i=1, \ldots, k$ 
\begin{equation}\label{bernst2}
u_{i}^{(1)}(x) := \ln \left(\frac{ \widehat{p}_{i}}{\pi \widehat{p}_{i} + (1-\pi)x}\right), u_{i}^{(2)}(x):= x\ln\left(\frac{x}{\pi \widehat{p}_{i} + (1-\pi)x}\right).
\end{equation} Then for $\pi 
\in [0,1] $  and $\widehat{Q}_{\theta} \in \mathbb{M}_{n}(\theta) $,
\begin{eqnarray}\label{slututryck}
E_{P_{\theta}}\left[D_{\rm JS}( \widehat{P}_{\mathbf{D}}, \widehat{Q}_{\theta})\right]   & =& \pi\sum_{i=1}^{k} \widehat{p}_{i}{\rm Be}_{n}\left(u_{i}^{(1)},  p_{i}\left( \theta\right)\right) \nonumber \\
& & \\
&+ &    (1-\pi)\sum_{i=1}^{k}{\rm Be}_{n}\left(u_{i}^{(2)},  p_{i}\left( \theta\right)\right).
\nonumber
\end{eqnarray} 
\end{lemma}                
The proof is given in Appendix  \ref{odotusarvo}.

\begin{proof} of  Theorem  \ref{jsdvantevarde}.
We denote by  $v_{i}^{(1)}(x)$ and $v_{i}^{(2)}(x)$  the respective second derivatives  w.r.t. $x$    of  $u_{i}^{(1)}(x)$ and  $u_{i}^{(2)}(x)$ respectively.  
When we apply  the  Voronovskaya$^{,}$s Asymptotic Formula (Equation~\ref{bernstvo}) to  $u_{i}^{(1)}(x)$ and  $u_{i}^{(2)}(x)$ in Equation~(\ref{bernst2}) we obtain 
\begin{eqnarray}\label{bernstvor2}
{\rm Be}_{n}\left(u_{i}^{(1)},  p_{i}\left( \theta\right)\right)&=& \ln \left(\frac{ \widehat{p}_{i}}{\pi\widehat{p}_{i} + (1-\pi) p_{i}\left( \theta\right)}\right) \nonumber \\
& & \\
& &+ \frac{ p_{i}\left( \theta\right)\left(1- p_{i}\left( \theta\right)\right)}{n} v_{i}^{(1) }\left( p_{i}\left( \theta\right) \right)+ o(1/n) \nonumber
\end{eqnarray} 
and 
\begin{eqnarray}\label{bernstvor3}
{\rm Be}_{n}\left(u_{i}^{(2)},  p_{i}\left( \theta\right)\right)&= & p_{i}\left( \theta\right)\ln \left(\frac{ p_{i}\left( \theta\right)}{\pi\widehat{p}_{i} + (1-\pi)p_{i}\left( \theta\right)}\right)  \nonumber \\
&& \\
&&+ \frac{ p_{i}\left( \theta\right)\left(1- p_{i}\left( \theta\right)\right)}{n}  v_{i}^{(2) }\left( p_{i}\left( \theta\right)\right) + o(1/n),\nonumber
\end{eqnarray} 
 respectively. Hence we get  in Equation~(\ref{slututryck}) 
\begin{eqnarray}\label{slututryckvor2}
E_{P_{\theta}}\left[D_{\rm JS}( \widehat{P}_{\mathbf{D}}, \widehat{Q}_{\theta})\right]  & =&\pi\sum_{i=1}^{k} \widehat{p}_{i}  \ln \left(\frac{ \widehat{p}_{i}}{\pi\widehat{p}_{i} +(1-\pi) p_{i}\left( \theta\right)}\right)  \nonumber \\
& & \\
&+ &  (1-\pi)\sum_{i=1}^{k}p_{i}\left( \theta\right)\ln \left(\frac{  p_{i}\left( \theta\right)}{\pi\widehat{p}_{i} +(1-\pi)p_{i}\left( \theta\right)}\right)   
\nonumber \\
& +&  \frac{1}{n}\sum_{i=1}^{k}  p_{i}\left( \theta\right)\left(1- p_{i}\left( \theta\right)\right)  \left[\pi \widehat{p}_{i}   v_{i}^{(1) }\left( p_{i}\left( \theta\right)\right)  + (1-\pi) v_{i}^{(2) }\left( p_{i}\left( \theta\right)\right)  \right]. \nonumber 
\end{eqnarray} 
It is shown in Appendix \ref{voronder} that 
$$
  \left[ \pi \widehat{p}_{i}   v_{i}^{(1) }\left( p_{i}\left( \theta\right)\right)  + (1-\pi) v_{i}^{(2) }\left( p_{i}(\theta) \right)\right]=V_{\rm F}\left(   p_{i}\left( \theta\right),  \widehat{p}_{i} \right),
$$
where $V_{\rm F}\left(   p_{i}\left( \theta\right),  \widehat{p}_{i} \right)$ is given in Equation~(\ref{andraderivator}). 
By definition of JSD    we   obtained Equation~(\ref{slututryckvor3}) as claimed. \end{proof}

This would seem to be a good point to   recall    the brief treatment of categorical  implicit models  in  \citet[pp.~195$-$196]{diggle1984monte}. Actually this work  has  $k=+\infty$, which entails, as they point out,  problems of convergence in the approximating expressions. 
However,  written in   our notations and for finite $k$,  Diggle and Gratton  estimate in \citet[Eq. (2.1)]{diggle1984monte} the implicit likelihood function by 
$$
\widehat{\cal L}(\theta) = \sum_{i=1}^{k} n_{i} \ln \widehat{q}_{i},
$$
where $\widehat{q}_{i} $ are in Equation~(\ref{qhat2}), i.e., $\widehat{Q}_{\theta} \in \mathbb{M}_{n}(\theta)$.   Then Diggle and Gratton  compute  
$$
E_{P_{\theta}}\left( \widehat{\cal L}(\theta) \right) = \sum_{i=1}^{k} n_{i} E_{P_{\theta}} \left(\ln \frac{\xi_{i}}{n} \right)
$$
by a second order Taylor series expansion around $\ln \left(\frac{ E_{P_{\theta}}(\xi_{i})}{n}\right)$ and obtain,  \citet[Eq. (2.2)]{diggle1984monte},
$$
E_{P_{\theta}}\left( \widehat{\cal L}(\theta) \right) = {\cal L}(\theta) -\sum_{i=1}^{k} \widehat{p}_{i} \frac{(1- p_{i}(\theta) )}{2p_{i}(\theta)} + o(1/n^{2}).
$$
When we apply Voronovskaya  Asymptotic Formula (Equation~\ref{bernstvo})  we obtain 
$$
 \sum_{i=1}^{k} n_{i} E_{P_{\theta}} \left(\ln \frac{\xi_{i}}{n} \right)= \sum_{i=1}^{k} n_{i}{\rm Be}_{n}\left( \ln, p_{i}\left( \theta\right)\right) $$  $$=  {\cal L}(\theta)  -
  \sum_{i=1}^{k} \widehat{p}_{i} \frac{(1- p_{i}(\theta) )}{2p_{i}(\theta)}+ o(1/n).
$$
As pointed out by  Diggle and Gratton themselves, this means instability for the estimate, when there are small values of some 
$p_{i}(\theta)$. This is not  a difficulty   in Equation~(\ref{andraderivatorsum2}), which is bounded and well defined even if there were  empty cell frequency  $\widehat{p}_{i} =0$  and small $p_{i}(\theta)$ for some $i$.

\subsection{Variance of  the  Simulator-Based JSD  Statistic: Preliminary Squares of the Taylor Expansions} \label{expvarexp}
 The notation in Equation~(\ref{taylorjsdt1}) is abbreviated  by   $\frac{\partial}{\partial q_{j}} F_{j}\left( p_{j}(\theta) \right) : = \frac{\partial}{\partial q_{j}} D_{\phi_{\rm JS}}\left( \widehat{{\bf p}},    E_{P_{\theta}}\left[  \widehat{{\bf q}}_{\theta} \right]    \right) $ and 
$ \frac{\partial^{2}}{\partial q_{j}^{2}} F_{j}\left( p_{j}(\theta)\right)  :=\frac{\partial^{2}}{\partial q_{j}^{2}}D_{\phi_{\rm JS}}\left( \widehat{{\bf p}},    E_{P_{\theta}}\left[  \widehat{{\bf q}}_{\theta} \right]\right)$. Then 
\begin{eqnarray}\label{taylorjsdt2} 
\left(\sum_{j=1}^{k} \frac{\partial}{\partial q_{j}} D_{\phi_{\rm JS}}\left(   \widehat{{\bf p}},   E_{P_{\theta}}\left[  \widehat{{\bf q}}_{\theta}  \right] \right)\left( \widehat{q}_{i,\theta} - p_{j}(\theta) \right) \right)^{2} = 
\sum_{j=1}^{k}\left[\frac{\partial}{\partial q_{j}} F_{j}\left( p_{j}(\theta) \right)\right]^{2}\left( \widehat{q}_{i,\theta} - p_{j}(\theta) \right)^{2} & & \nonumber \\
& & \\
+ 2\sum_{i=1}^{k-1} \sum_{j=i+1}^{k}\frac{\partial}{\partial q_{i}} F_{i}\left( p_{i}(\theta) \right) \frac{\partial}{\partial q_{j}} F_{j}\left( p_{j}(\theta) \right) \left( \widehat{q}_{i,\theta} - p_{i}(\theta) \right) \left( \widehat{q}_{j,\theta} - p_{j}(\theta) \right). & &  \nonumber 
\end{eqnarray} 
Next 
\begin{eqnarray}\label{taylorjsdt3} 
\left(\frac{1}{2!}\sum_{j=1}^{k} \frac{\partial^{2}}{\partial q_{j}^{2}} D_{\phi_{\rm JS}}\left( \widehat{{\bf p}},    E_{P_{\theta}}\left[  \widehat{{\bf q}}_{\theta} \right]    \right) \left( \widehat{q}_{j,\theta} - p_{j}(\theta) \right) \right)^{2} = 
\frac{1}{4}\sum_{j=1}^{k}\left[\frac{\partial^{2}}{\partial q_{j}^{2}} F_{j}\left( p_{j}(\theta) \right)\right]^{4}\left( \widehat{q}_{i,\theta} - p_{j}(\theta) \right)^{4} & & \nonumber \\
& & \\
+\frac{1}{2} \sum_{i=1}^{k-1} \sum_{j=i+1}^{k}\frac{\partial^{2}}{\partial q_{i}^{2}} F_{i}\left( p_{i}(\theta) \right) \frac{\partial^{2}}{\partial q_{j}^{2}} F_{j}\left( p_{j}(\theta) \right) \left( \widehat{q}_{i,\theta} - p_{i}(\theta) \right)^{2}\left( \widehat{q}_{j,\theta} - p_{j}(\theta) \right)^{2}. & &  \nonumber 
\end{eqnarray} 
Finally, 
\begin{eqnarray}\label{taylorjsdt4}
\sum_{i=1}^{k} \frac{\partial}{\partial q_{i}} D_{\phi_{\rm JS}}\left(   \widehat{{\bf p}},   E_{P_{\theta}}\left[  \widehat{{\bf q}}_{\theta}  \right] \right)\left( \widehat{q}_{i,\theta} - p_{j}(\theta) \right) \cdot \frac{1}{2!}\sum_{j=1}^{k} \frac{\partial^{2}}{\partial q_{j}^{2}} D_{\phi_{\rm JS}}\left( \widehat{{\bf p}},    E_{P_{\theta}}\left[  \widehat{{\bf q}}_{\theta} \right]    \right) \left( \widehat{q}_{j,\theta} - p_{j}(\theta) \right)^{2}  & = & \nonumber \\
& & \\
 \frac{1}{2}\sum_{i=1}^{k}  \sum_{j=1}^{k} \frac{\partial}{\partial q_{}} F_{j}\left( p_{i}(\theta) \right) \cdot\frac{\partial^{2}}{\partial q_{j}^{2}} F_{j}\left( p_{j}(\theta)\right)\left( \widehat{q}_{i,\theta} - p_{i}(\theta) \right) \left( \widehat{q}_{j,\theta} - p_{j}(\theta)\right)^{2}. \nonumber 
\end{eqnarray}
\begin{lemma}\label{jsdtaylorlem} 
Assume   $ {\bf p}_{\theta} = E_{P_{\theta}}\left[  \widehat{{\bf q}}_{\theta} \right]  \in  \stackrel{o}{\triangle}_{k-1}$. Then for  $ \widehat{{\bf q}}_{\theta}$  in an open neighborhood of  $  E_{P_{\theta}}\left[  \widehat{{\bf q}}_{\theta}\right] \in \stackrel{o}{\triangle}_{k-1}$
\begin{equation}\label{jsdstatmse} 
E_{P_{\theta}}\left[ \left( D_{\rm JS}\left( \widehat{{\bf p}} , \widehat{{\bf q}}_{\theta}
\right)  - D_{\rm JS} \left( \widehat{{\bf p}} ,  E_{P_{\theta}}\left[  \widehat{{\bf q}}_{\theta} \right] \right)
\right)^{2}\right] = M_{1}+M_{2} + M_{3}, 
\end{equation} 
where 
\begin{eqnarray}\label{taylorjsdt21} 
M_{1} = \sum_{j=1}^{k}\left[\frac{\partial}{\partial q_{j}} F_{j}\left( p_{j}(\theta) \right)\right]^{2}E_{P_{\theta}} \left( \widehat{q}_{i,\theta} - p_{j}(\theta) \right)^{2} 
& & \nonumber \\
& & \\
 + 2\sum_{i=1}^{k-1} \sum_{j=i+1}^{k}\frac{\partial}{\partial q_{i}} F_{i}\left( p_{i}(\theta) \right) \frac{\partial}{\partial q_{j}} F_{j}\left( p_{j}(\theta) \right)E_{P_{\theta}}\left[\left( \widehat{q}_{i,\theta}- p_{i}(\theta) \right) \left( \widehat{q}_{j,\theta} - p_{j}(\theta) \right)\right], & & \nonumber 
\end{eqnarray} 
\begin{equation}\label{taylorjsdt41}
M_{2}=
 \sum_{i=1}^{k}  \sum_{j=1}^{k} \frac{\partial}{\partial q_{}} F_{j}\left( p_{i}(\theta) \right) \cdot\frac{\partial^{2}}{\partial q_{j}^{2}} F_{j}\left( p_{j}(\theta)\right)E_{P_{\theta}}\left[\left( \widehat{q}_{i,\theta} - p_{i}(\theta) \right) \left( \widehat{q}_{j,\theta} - p_{j}(\theta)\right)^{2}\right],
\end{equation}
and 
\begin{eqnarray}\label{taylorjsdt31} 
M_{3}= 
\frac{1}{4}\sum_{j=1}^{k}\left[\frac{\partial^{2}}{\partial q_{j}^{2}} F_{j}\left( p_{j}(\theta) \right)\right]^{4}E_{P_{\theta}}\left[ \left( \widehat{q}_{i,\theta} - p_{j}(\theta) \right)^{4} \right]& & \nonumber \\
& & \\
+\frac{1}{2} \sum_{i=1}^{k-1} \sum_{j=i+1}^{k}\frac{\partial^{2}}{\partial q_{i}^{2}} F_{i}\left( p_{i}(\theta) \right) \frac{\partial^{2}}{\partial q_{j}^{2}} F_{j}\left( p_{j}(\theta) \right)E_{P_{\theta}}\left[ \left( \widehat{q}_{i,\theta} - p_{i}(\theta) \right)^{2}\left( \widehat{q}_{j,\theta} - p_{j}(\theta) \right)^{2}\right]. & &  \nonumber 
\end{eqnarray} 
\end{lemma}
\begin{proof} The proof follows  by inserting  Equation~(\ref{taylorjsdt1}) in Equation~(\ref{jsdstatmse})  and   by squaring, which gives by elementary algebra  a sum 
of the right hand sides  of   Equations~(\ref{taylorjsdt2}) and (\ref{taylorjsdt3}) plus two times the right  hand side of Equation~(\ref{taylorjsdt4}), whereafter the pertinent expectations  are  taken. \end{proof} 

\begin{lemma}\label{jsdvantevardetaylor} 
Let   $\widehat{{\bf p}}= \triangle  \left(  \widehat{P}_{\mathbf{D}} \right) $ and 
 $\widehat{{\bf q}}_{\theta}= \triangle  \left(  \widehat{Q}_{\theta} \right) $ with $\widehat{Q}_{\theta} \in \mathbb{M}_{n}(\theta)$.  $m_{i}= \pi \widehat{p}_{i} + (1-\pi) p_{i}(\theta)$ for $i=1, \ldots, k$.  $\sigma_{j}^{2}(\theta):= p_{j}(\theta)(1-p_{j}(\theta))$.
Assume Equation~(\ref{posass})   for  $P_{\theta} \in \mathbb{M}$ and let  ${\bf p}_{\theta}= \triangle  \left(  P_{\theta} \right) $. 
Set 
\begin{equation}\label{msedef}
{\rm MSE}\left( \widehat{P}_{\mathbf{D}} ,   \widehat{Q}_{\theta}\right):=E_{P_{\theta}}\left[ \left( D_{\rm JS}\left( \widehat{{\bf p}} , \widehat{{\bf q}}_{\theta}
\right)  - D_{\rm JS} \left( \widehat{{\bf p}} ,  E_{P_{\theta}}\left[  \widehat{{\bf q}}_{\theta} \right]\right)\right)^{2} \right].
\end{equation}
Then   for $ \widehat{{\bf q}}_{\theta}$  in an open neighborhood of   ${\bf p}_{\theta}  $, omitting third and higher order terms,
\begin{eqnarray}\label{taylorjsdt211fin} 
{\rm MSE}\left( \widehat{P}_{\mathbf{D}} ,   \widehat{Q}_{\theta}\right)  & = &\frac{(1-\pi) }{n}\sum_{j=1}^{k}\left[ 
  \ln \left( p_{j}(\theta) /m_{j}\right)\right]^{2}
 \sigma_{j}^{2}(\theta)  \nonumber \\
& & 
 -  \frac{2(1-\pi) ^{2}}{n}\sum_{i=1}^{k-1}  \ln \left( p_{i}(\theta)/m_{j}\right) p_{i}(\theta)\sum_{j=i+1}^{k}
 \ln \left(p_{j}(\theta)/m_{j}\right)p_{j}(\theta) \nonumber \\
& &+\frac{(1-\pi)^{2}\pi}{n^{2}} \sum_{i=1}^{k}  \ln \left( p_{i}(\theta)/m_{i}\right)p_{i}(\theta)  \sum_{j=1}^{k} \cdot  
\frac{ \widehat{p}_{j}}{ m_{j}} \left( 2p_{j}(\theta)-1 \right) \nonumber \\
& &+ \frac{\pi (1-\pi)^{4}}{4n^{3}}\sum_{j=1}^{k}\frac{\widehat{p}^{4}_{j}}{p_{j}(\theta)^{4}m_{j}^{4}}
 \sigma_{j}^{2}(\theta)\left( 1+ 3 \sigma_{j}^{2}(\theta) (n-2) \right) \nonumber \\
& &+ \frac{3(n-2)(\pi(1-\pi))^{2}}{n^{3}}  \sum_{i=1}^{k-1}\frac{ \widehat{p}_{i} p_{i}(\theta)}{ m_{i}} \sum_{j=i+1}^{k}\frac{ \widehat{p}_{j}}{m_{j}}p_{j}(\theta)\nonumber \\
& & + \frac{(n-2)(\pi(1-\pi))^{2}}{n^{3}}\sum_{i=1}^{k-1} \frac{ \widehat{p}_{i}}{p_{i}(\theta) m_{i}}\sum_{j=i+1}^{k}\frac{ \widehat{p}_{j}}{p_{j}(\theta) m_{j}}\left[1-(p_{i}(\theta)+p_{j}(\theta))\right] 
\nonumber \\
& &  +\frac{(\pi(1-\pi))^{2}}{n^{3}} \sum_{i=1}^{k-1}\frac{ \widehat{p}_{i}}{m_{i}} \sum_{j=i+1}^{k}\frac{ \widehat{p}_{j}}{ m_{j}}.
\end{eqnarray} 
\end{lemma} 
%\end{document}
The proof is found in Appendix  \ref{mse}. The   lemma  relies crucially on the recent work by F.~Ouimet  in  \citet{ouimet2021general} on moments of multinomial distributions.

%\citet[p.270]{griffiths2013raw}, \citet[Eq. (79) p.~25]{ouimet2021general} and \citet[Table 2., p.5]{skorski2020handy}
%
%$E_{P_{\theta}} \left( \xi_{j} -n p_{j}(\theta) \right)^{3} = n \sigma_{j}^{2}(\theta) ( 1-2 p_{j}(\theta))$
\subsection{The Variance}
By definition 
$$
{\rm Var}_{P_{\theta}}\left[ D_{\rm JS}(  \widehat{P}_{\mathbf{D}},  \widehat{Q}_{\theta})   \right] = 
E_{P\theta} \left[\left\{D_{ \rm JS}(  \widehat{P}_{\mathbf{D}},  \widehat{Q}_{\theta})  -E_{P_{\theta}}\left[D_{\rm JS}(  \widehat{P}_{\mathbf{D}},  \widehat{Q}_{\theta}) \right] \right\}^{2} \right].
$$
We insert from Theorem   \ref{jsdvantevarde}, i.e., Equation~(\ref{slututryckvor3}),  and get 
$$
= E_{P\theta} \left[\left\{\left( D_{ \rm JS}(  \widehat{P}_{\mathbf{D}},  \widehat{Q}_{\theta})  -D_{\rm JS} \left(  \widehat{P}_{\mathbf{D}},  P_{\theta} \right) \right)  - \left(\frac{1}{n}V_{\rm F}\left(  P_{\theta},  \widehat{P} _{\mathbf{D}}\right) +  o(1/n) \right)\right\}^{2} \right].
$$
Thus  there emerges  the sum of 
${\rm MSE}\left( \widehat{P}_{\mathbf{D}} ,   \widehat{Q}_{\theta}\right) $ defined in Equation~(\ref{msedef}),  of 
$$
-2\left(\frac{1}{n}V_{\rm F}\left(  P_{\theta},  \widehat{P} _{\mathbf{D}}\right) +  o(1/n) \right)E_{P\theta} \left[ D_{ \rm JS}(  \widehat{P}_{\mathbf{D}},  \widehat{Q}_{\theta})  -D_{\rm JS} \left(  \widehat{P}_{\mathbf{D}},  P_{\theta} \right)  \right]
$$
and of $\left(\frac{1}{n}V_{\rm F}\left(  P_{\theta},  \widehat{P} _{\mathbf{D}}\right) +  o(1/n)\right)^{2}$. Here 
$$
E_{P\theta} \left[ D_{ \rm JS}(  \widehat{P}_{\mathbf{D}},  \widehat{Q}_{\theta})  -D_{\rm JS} \left(  \widehat{P}_{\mathbf{D}}, P_{\theta} \right)\right]= E_{P\theta} \left[ D_{ \rm JS}(  \widehat{P}_{\mathbf{D}},  \widehat{Q}_{\theta}) \right]-D_{\rm JS} \left(  \widehat{P}_{\mathbf{D}}, P_{\theta}\right)
$$
and again by  Equation~(\ref{slututryckvor3}) this equals $
=\frac{1}{2n}V_{\rm F}\left(  P_{\theta},  \widehat{P} _{\mathbf{D}}\right) +  o(1/n) 
$.
Hence we have reached the   formula in the next proposition. 
\begin{theorem}\label{variansjsdstatsim}
\begin{equation}\label{variansfin}
{\rm Var}_{P_{\theta}}\left[ D_{\rm JS}(  \widehat{P}_{\mathbf{D}},  \widehat{Q}_{\theta})   \right] = 
{\rm MSE}\left( \widehat{P}_{\mathbf{D}} ,   \widehat{Q}_{\theta}\right)- \left[\frac{1}{2n}V_{\rm F}\left(  P_{\theta},  \widehat{P} _{\mathbf{D}}\right) +  o(1/n)\right]^{2}.
\end{equation}
Here   ${\rm MSE}\left( \widehat{P}_{\mathbf{D}} ,   \widehat{Q}_{\theta}\right) $  is  in Equation~(\ref{taylorjsdt211fin})  and 
$V_{\rm F}\left(  P_{\theta},  \widehat{P} _{\mathbf{D}}\right)$  is given in  Equation~(\ref{andraderivatorsum}). 
\end{theorem}
The first, third and fifth terms  in right hand side of  Equation~(\ref{taylorjsdt211fin}) contain  category probabilities  $p_{i}(\theta)$ as  argument  of the natural logarithm   or  in the  denominator in such manner that  small category probabilities  will contribute to  a  high variance. Hence 
guaranteed  stable ${\rm Var}_{P_{\theta}}\left[ D_{\rm JS}(  \widehat{P}_{\mathbf{D}},  \widehat{Q}_{\theta})   \right] $  are found in  suitable  neighborhoods of the barycenter.

\section{The  JSD-statistic and the Pearson    $\chi^{2}$-  Statistic} \label{secttjitvaa}
The above derived Taylor expansions provide an analytically tractable way to summarize the asymptotic behavior of the JSD-statistic in terms of its expected value and variance. In this section we expand the theory further by showing how a function of JSD-statistic can be characterized by the Pearson    $\chi^{2}$-  statistic under a given null hypothesis. 
\subsection{The Simulator-based  Symmetric JSD Test of Hypothesis of Fit}\label{secttjitvaa2}
It holds by  Theorem \ref{nskonv} that $ D_{\rm JS}\left(\widehat{P}_{\mathbf{D}}, \widehat{Q}_{\theta}   \right)  \rightarrow  D_{\rm JS} \left(  \widehat{P}_{\mathbf{D}},  P_{\theta}\right)$ 
with $P_{\theta}$ probability one, as $ n\rightarrow +\infty$,  where $n_{o}$ is fixed. Furthermore, by  Equation~(\ref{plan}) 
\begin{equation}\label{ukorr}
D_{\rm JS} \left(  \widehat{P}_{\mathbf{D}},  P_{\theta} \right) = \pi D_{\rm KL} \left(  \widehat{P}_{\mathbf{D}},  P_{\theta} \right) - D_{\rm KL} \left(  \widehat{M}_{\rm imp},  P_{\theta} \right), 
\end{equation}
where  $   \widehat{M}_{\rm imp}=   \pi \widehat{P}_{\mathbf{D}} + (1-\pi) P_{\theta}  $. 
 Here a well known  exercise, see \citet[Exercise 11.2]{cover2012elements}  or   \citet[Section~16.3.4 p.~597]{agresti2003categorical}, cf., Equation~(\ref{shieldscsiszar}) with $\phi(x)=x \ln x$,     tells that
$$
 n_{o}D_{\rm KL} \left(  \widehat{P}_{\mathbf{D}},  P_{\theta} \right) = \frac{1}{2}\sum_{i=1}^{k} \frac{\left(n_{i}  -n_{o}p_{i}(\theta) \right)^{2}}{ n_{o}p_{i}(\theta) } + O_{P_{\theta}}(1/n_{o}^{1/2}).
$$
Here  $ O_{P_{\theta}}$  is the probabilistic order as defined  in \citet[Section~16.1.1. p.~588]{agresti2003categorical}.  By  the  cited  exercise we get also 
$$
 n_{o} D_{\rm KL} \left(  \widehat{M}_{\rm imp},  P_{\theta} \right) =  \frac{1}{2}\sum_{i=1}^{k} \frac{\left( \left( \pi  n_{i} + (1-\pi)  n_{o} p_{i}(\theta) \right) - n_{o}p_{i}(\theta) \right)^{2}}{ n_{o}p_{i}(\theta) } +O_{P_{\theta}}(1/n_{o}^{1/2}).
$$
A simplification yields  here 
$$
 n_{o} D_{\rm KL} \left(  \widehat{M}_{\rm imp},  P_{\theta} \right) = \frac{\pi^{2}}{2}\sum_{i=1}^{k} \frac{\left( n_{i}  -n_{o}p_{i}(\theta) \right)^{2}}{ n_{o}p_{i}(\theta) } +O_{P_{\theta}}(1/n_{o}^{1/2}).
$$
 Hence we have by Equation~(\ref{ukorr})
\begin{equation}\label{dkltjitv} 
\frac{ 2 n_{o}}{\pi(1-\pi)}D_{\rm JS} \left(  \widehat{P}_{\mathbf{D}},  P_{\theta} \right) = \sum_{i=1}^{k} \frac{\left( n_{i}  -n_{o}p_{i}(\theta) \right)^{2}}{ n_{o}p_{i}(\theta) }  + O_{P_{\theta}}(1/n_{o}^{1/2}),
\end{equation} 
as follows by the rules of  computing  $ O_{P_{\theta}}$ for a sum of two sequences.

We shall next put the results above in a context. 
If $P_{\theta} $ were  explicit, we would  have the following  standard situation. 
 We have a set of observed frequencies  $n_{1}, \ldots, n_{k}$, $n_{1}+ \ldots + n_{k}=n_{o}$ computed from $ {\mathbf D}= (D_{1}, \ldots, D_{n_{o}}) $,   an i.i.d. $n_{o}$-sample  $\sim$ $P_{o}$.  Let  ${\bf p}_{\theta}= \triangle \left(P_{\theta} \right)$   and  ${\bf p}_{o}= \triangle \left(P_{o} \right)$.      If $\theta$ is known, there is the simple 
null hypothesis 
$$
{\rm H}_{0}:  \quad  {\bf p}_{o}= {\bf p}_{\theta}
$$
and the composite alternative hypothesis 
$$
{\rm H}_{1}:  \quad  {\bf p}_{o}  \neq {\bf p}_{\theta}. 
$$
The  Pearson statistic $ \chi^{2}_{n_{o}}:=\sum_{i=1}^{k} \frac{\left(n_{i}  -n_{o}p_{i}(\theta) \right)^{2}}{ n_{o}p_{i}(\theta) } $   is  the familiar  goodness-of-fit statistic for test of the null hypothesis  using the given observed frequencies. 
One of the most well known pioneering results of statistics tells that,  as $n_{o}$  becomes large, $ \chi^{2}_{n_{o}}$   follows  asymptotically,  under the null hypothesis,  the 
$ \chi^{2}$  distribution with $k -1$ degrees of freedom, $  \chi^{2}(k-1)$, for a proof  see, e.g.,  \citet[Section1.5.7 p.~22]{agresti2003categorical}. The limiting distribution   $\chi^{2}(k-1)$  does not depend on the explicit  form or  complexity of  the  functions $p_{i}(\theta) $. 

  By simulator-modeling we have chosen from  $\Theta$ a value $\theta$, which is    thus  known to us.  This defines  a program function $\mathbf{M}_{C}(\theta)$, which induces the implicit  $P_{\theta} \in \mathbb{P}$. 
The Pearson statistic   $ \chi^{2}_{n_{o}}$  in Equation~(\ref{dkltjitv}) contains  implicit model category probabilities, but the quoted limit theorem is valid  mathematically under the simulator-model null hypothesis, and limiting  distribution $\chi^{2}(k-1)$  does not depend on  these implicit functions. 

% We have shown that  $ D_{\rm JS}\left(\widehat{P}_{\mathbf{D}}, \widehat{Q}_{\theta}   \right)$ converges to   $D_{\rm JS} \left(  \widehat{P}_{\mathbf{D}},  P_{\theta}\right)$  almost surely w.r.t. $P_{\theta}$, it converges also in distribution, as  $n$  grows to infinity. Thereby $n_{o}$ is   the given number of  non-synthetic samples in $\mathbf{D}$.  % ulpu: commented out because rather than assume that the simulation count approaches infinite so that we can calculate JSD between observed and simulated data based on one simulation, we use its expected value over simulations

%By the preceding, i.e., proposition  \ref{jsdvantevarde},  Equation~(\ref{andraderivatorsum2}), and  Equation~(\ref{dkltjitv}), omitting remainder terms
%\begin{equation}\label{jsdtjitvaae}
% \frac{2 n_{o}}{\pi (1-\pi)} D_{\rm JS}\left(\widehat{P}_{\mathbf{D}}, \widehat{Q}_{\theta}   \right)  -  \frac{n_{o}(k-1) }{2n \pi(1-\pi)}  \stackrel{.}{=} \sum_{i=1}^{k} \frac{\left( n_{i}  -n_{o}p_{i}(\theta) \right)^{2}}{ n_{o}p_{i}(\theta) }.
%\end{equation} % ulpu: commented out because this does not match the expression in equation 10.1 ie when we substitute pi=1/2 here the second term would be 2 x n_o/n x (k-1) whereas in equation 10.1 we have n_o/n x (k-1)

% ulpu: added this ->

To summarize, we have shown that the Pearson statistic can be approximated as
\begin{equation}
\chi^2_{n_o} \approx \frac{ 2 n_{o}}{\pi(1-\pi)}D_{\rm JS} \left(  \widehat{P}_{\mathbf{D}},  P_{\theta} \right).
\end{equation}
However this is not computable when the mapping between model parameters and category probabilities is unknown.
To arrive at a computable approximation, we propose to approximate $D_{\rm JS} \left(  \widehat{P}_{\mathbf{D}},  P_{\theta} \right)$ using $E_{P_{\theta}}\left[D_{\rm JS}(  \widehat{P}_{\mathbf{D}},  \widehat{Q}_{\theta}) \right]$ which can be computed based on repeated simulations.
Proposition \ref{jsdvantevarde} (Equation \ref{slututryckvor3}) suggests
\begin{equation}
\chi^2_{n_o}\approx \frac{ 2 n_{o}}{\pi(1-\pi)}\left[E_{P_{\theta}}\left[D_{\rm JS}(  \widehat{P}_{\mathbf{D}},  \widehat{Q}_{\theta}) \right] - \frac{1}{2n}V_{\rm F}\left(  P_{\theta},  \widehat{P} _{\mathbf{D}}\right)\right],
\end{equation}
where $n$ is the simulated data set size.
$V_{\rm F}\left(P_{\theta}, \widehat{P} _{\mathbf{D}}\right)$  cannot be computed when the mapping between model parameters and category probabilities is unknown, but we observe based on Equation (\ref{andraderivatorsum2})  that when $P_\theta\approx \widehat P_{\mathbf{D}}$, $V_{\rm F}\left(P_{\theta}, \widehat{P} _{\mathbf{D}}\right)\approx \pi(1-\pi)(k-1)$.
Hence we approximate the test statistic as
\begin{equation}\label{jsdtjitvaae}
\chi^2_{n_o} \approx \frac{2 n_{o}}{\pi (1-\pi)} E_{P_{\theta}}\left[D_{\rm JS}(  \widehat{P}_{\mathbf{D}},  \widehat{Q}_{\theta}) \right]  
 -  \frac{n_{o}}{n}(k-1).
\end{equation}

\subsection{ The $\chi^{2}$-divergence Options}

The $\phi$-divergence with divergence function  $\phi(x)=(x-1)^{2}$ is called   $\chi^{2}$-divergence    and is denoted by $\chi^{2}( P, Q)$     see,   e.g.  \citet{osterreicher2002csiszar}.  Let $P \in \mathbb{P}$ and $Q$ $\in \mathbb{P}$, and  assume that ${\rm supp}(Q)={\cal A}$.  By  
Equation~(\ref{informationone})  
\begin{equation}\label{tjitvaa}
\chi^{2}( P,Q):= \sum_{i=1}^{k} \frac{\left( p_{i} -q_{i}   \right)^{2}}{q_{i}}.
\end{equation}
There is the universal  "asymptotic equivalence'' of
$\phi$-divergences subject to the differentiability hypotheses, see \citet[Thm 4.1., pp.~448$-$449]{csiszar2004information}  meaning that  if $\phi(x)$ is  twice differentiable at $x$ = 1 and  the second derivative  $\phi^{(2)}(1) > 0$,  then   for $P$ close to $Q$   
\begin{equation}\label{shieldscsiszar}  
D_{\phi}(P,Q) \approx \frac{\phi^{(2)}(1)}{2}\chi^{2}( P,Q). 
\end{equation} 
From  Equation~(\ref{jsddivder3})  we have $\phi^{(2)}_{\rm JS}(1)  = \pi (1-\pi)$, and hence Equation~(\ref{shieldscsiszar}) holds for the JSD.   
Let us set 
$ \widehat{\varepsilon} = \max_{1 \leq i \leq k}\mid \frac{p_{i}}{q_{i}} -1  \mid$. 
 Then   \citet[Thm 3]{pardo2003asymptotic} says that 
\begin{equation}\label{vajdaparduasymp22}
\left |  \frac{2}{\phi^{(2)}(1)} D_{\phi}(P, Q)  -  \chi^{2}(P, Q) \right | \leq c \widehat{\varepsilon}  \chi^{2}(P,Q).
\end{equation}
Let $\widehat{Q}_{\theta} \in \mathbb{M}_{n}(\theta)$.  Then  $\widehat{q}_{i} = \frac{\xi_{i}}{n}$ and  in view of  Equation~(\ref{tjitvaa}) we obtain , if  $n=n_{o}$,  the  Pearson
$\chi^{2}$-statistic
\begin{equation}\label{neymandiv3}
n \chi^{2}(\widehat{P}_{\mathbf{D}},\widehat{Q}_{\theta})=\sum_{i=1}^{k} \frac{\left(n_{i} -n\widehat{q}_{i} \right)^{2}}{ n\widehat{q}_{i}}, 
\end{equation}
which  is computable from observed and simulated data but requires $n=n_{o}$, and  yields an awkward goodness-of-fit 
thinking.

 Due to  the  symmetry, 
$ D_{\rm JS, 1/2}\left(\widehat{P}_{\mathbf{D}},\widehat{Q}_{\theta}\right)$ can also be approximated by 
the Neyman modified $\chi^{2}$-statistic, i.e., 
\begin{equation}\label{neymandiv4}
n \chi^{2}(\widehat{Q}_{\theta}, \widehat{P}_{\mathbf{D}})=\sum_{i=1}^{k} \frac{\left(\xi_{i}  -n\widehat{p}_{i} \right)^{2}}{ n\widehat{p}_{i}}, 
\end{equation}
which  requires $n=n_{o}$   and  can be seen as  situation of  a  misspecified null hypothesis, cf.,  \citet[Section16.3.5 p.~598]{agresti2003categorical} or  \citet[Thm 3.1, p. 446]{cressie1984multinomial} requiring  further   assumptions.   \\

 \citet[pp.~364$-$365]{pardo2003asymptotic}  derive  a  non-central $\chi^{2}$-distribution for  $D_{\phi}\left( \widehat{P}_{\mathbf{D}},\widehat{Q}_{\mathbf{D}^{\ast} }\right)$ in  an explicit model  $\mathbb{M}= \{P_{\theta} |  \theta \in \Theta \}$, when   $\mathbf{D}=(X_{1}, \ldots, X_{n_{o}} ) \sim P_{\theta_{o}}$  and there  are  $n^{\ast}$ additional data $\mathbf{D}^{\ast}=  \left(X_{n_{o}+1}, \ldots, X_{n_{o}+n^{\ast}} \right) \sim P_{\theta_{o}}$.  
$\widehat{Q}_{\mathbf{D}^{\ast}} $  is given by $p_{j}\left(  \widehat{\theta} \right)$, where  $ \widehat{\theta}$ is   an estimate of  based on $\mathbf{D}^{\ast}$, i.e.,  $\widehat{\theta}= \widehat{\theta}\left( \mathbf{D}^{\ast} \right)$.

\subsection{Assumptions}

\newcommand{\param}{\theta}
\newcommand{\ncls}{k} % num classes
\newcommand{\icls}{i} % class index
\newcommand{\nobs}{n_o} % num observations
\newcommand{\nsms}{n} % num simulations
\newcommand{\nrep}{m} % num repeated simulations
\newcommand{\pcls}{p_\icls}

% multinomial distribution

The statistical test in which $\chi^2_{\nobs}$ is compared to $\chi^2(\ncls-1)$ distribution is considered reliable when (1) the observed data can be modeled as a random sample from a fixed multinomial distribution and (2) the sample size is large and all categories are associated with adequate observation counts. 
The same applies when the test statistic is calculated based on simulated data as proposed in Equation (\ref{jsdtjitvaae}).
However the multinomial distribution assumption can be a particular concern because we assume that the dependencies between the model parameters and observed data cannot be captured with an explicit or even  tractable likelihood function.
This means that the observed and simulated data are not necessarily expected to strictly follow a multinomial distribution, and as a consequence, we cannot assume that the null distribution associated with the proposed test statistic is $\chi^2(\ncls-1)$.
However we demonstrate in the simulation experiments (Section \ref{sec:experiments}) how effective sample size (ESS) can be used to correct the test statistic distribution.

% sample size

The minimum observation count and compensation for small sample size have received more attention in the past. For studies and discussion, see for example \cite{yates1934contingency} and \cite{yarnold1970minimum}.
In practice it is common to assume that the test is reliable when most categories have expected observation counts above 5 and all categories have expected observation count above 1.

\section{Confidence Intervals for $\theta$ with Test Inversion and the Symmetric JSD-statistic}\label{testinv}

From here on we shall  specialize to   $\pi=1/2$, unless otherwise stated   and  use  the notation 
$D_{\rm JS,1/2}( P, Q) $ for the corresponding JSD.  This is  a  symmetrized and smoothed version of KLD,  since  $D_{\rm JS,1/2}( P, Q)= D_{\rm JS,1/2}( Q,P).$  
The special case  $D_{\rm JS,1/2}\left( P, Q \right)$ is often  used in applying the JSD in machine learning. By Equation~(\ref{range2}) we get  
  $ 0 \leq D_{\rm JS,1/2}\left( P, Q \right) \leq \ln (2)$. 
    The paper \citet{topsoe2000some} provides several   expressions (e.g., in terms of infinite series)  and  bounds for  $ D_{\rm JS, 1/2}\left( P, Q \right)$.  \citet[Lemma 4, p.~787]{kelly2012classification} derives 
$D_{\rm JS,1/2}( P, Q) $ by a version of  Equation~(\ref{jsinformationsibs}) starting from a generalized likelihood ratio. 
 It is shown in \citet{endres2003new}, that the symmetric $  \sqrt{D_{\rm JS,1/2}\left( P,Q \right)}$ satisfies the triangle inequality  and is thus a metric  on $\mathbb{P} $. 
 By Equations~(\ref{andraderivatorsum2})  and (\ref{jsdtjitvaae}) we get for $\pi=1/2$
\begin{equation}\label{jsdtjitvaae2}
   8 n_{o} E_{P_{\theta}}\left[D_{\rm JS, 1/2}\left(\widehat{P}_{\mathbf{D}}, \widehat{Q}_{\theta}   \right)\right] -  \frac{n_{o} (k-1)}{n} 
\stackrel{d}{\approx}  \chi^{2}(k-1).
\end{equation}
 Note that $\phi^{(2)}_{\rm JS, 1/2}(1)  = \frac{1}{4}$, so this agrees in the pertinent term  with Equations~(\ref{shieldscsiszar}) and   (\ref{vajdaparduasymp22}).  Of course, the issue of how small  $n_{o}p_{j}(\theta)$ can be for the approximation by 
$\chi(k-1)$ to  be valid, arises here.

% ulpu: added this->
 
Let us denote the test statistic in Equation (\ref{jsdtjitvaae2}) as $T(\theta)$.
The test statistic and its asymptotic distribution under the null hypothesis are compared to statistically assess the observed support to candidate parameters $\theta$.
In practice we can use the chi-squared distribution to determine critical values $h(\alpha)$ such that
\begin{equation}
P(T(\theta_o) < h(\alpha)) = 1-\alpha, 
\end{equation}
where $\theta_o$ are parameter values that mimic the true distribution $P_{o}$ that produced the observed data.
This means that when $T(\theta)$ exceeds $h(\alpha)$, the hypothesis that parameters $\theta$ mimic the true category probabilities can be rejected at significance level $\alpha$.
In addition the parameter values for which this hypothesis cannot be rejected at significance level $\alpha$ constitute a $100(1-\alpha)$~\% confidence interval or confidence set,
\begin{equation}
A(\mathbf{D}) = \left\{\theta: T(\theta)\leq h(\alpha) \right\}.
\end{equation}
While a maximum likelihood or minimum JSD estimate indicates the model parameters that best replicate the observed data,
a confidence set takes into account the expected random variation between observations and provides information about alternative explanations to the observed data.
Since the test statistic studied in this work measures absolute rather than relative support to the candidate parameters, the confidence sets can also be empty.
This is expected when the simulator model cannot replicate the observed frequencies due to model mismatch or when the observed frequencies represent an outcome that is overall rare compared to the selected significance level $\alpha$.

\section{Simulation Experiments}\label{sec:experiments}

\newcommand*{\vcenteredhbox}[1]{\begingroup
\setbox0=\hbox{#1}\parbox{\wd0}{\box0}\endgroup}

This section describes simulation experiments carried out to evaluate the test statistic proposed in this work.
The section is organized as follows.
Section~\ref{sec:methods} describes how the test statistic values are calculated based on observed and simulated data while Section~\ref{sec:evaluation} introduces the evaluation measure used in the experiments.
Each experiment is then described in more detail and the evaluation results are reported in Sections \ref{sec:experiment1}-\ref{sec:experiment3}.

\subsection{Methods}\label{sec:methods}

\newcommand{\avejsd}{E_{P_{\theta}}\left[D_{\rm JS,1/2}( \widehat{P}_{\mathbf{D}}, \widehat{Q}_{\theta})\right] }

The proposed test statistic is calculated based on the expected JSD between observed and simulated data.
The expected value $\avejsd$ can be estimated as average calculated over JSD between the observed data and $\nrep$ simulated data sets generated with parameters $\param$,
\begin{equation}
\avejsd=\frac{1}{m}\sum_{j=1}^{m} D_{\rm JS,1/2}( \widehat{P}_{\mathbf{D}}, \widehat{Q}_{\theta}^{(j)}),
\end{equation}
where $\widehat{Q}_{\theta}^{(j)}$ denotes the category proportions in simulated set $j$.
This method is expected work well in hypothesis testing, since the estimates are accurate when $\nrep$ is large.
However when we want to estimate a confidence set over several candidate parameters, running $\nrep$ simulations with each candidate is not practical in case the candidate set is large and individual simulations are computationally expensive.
Hence we also examine using Bayesian optimization (BO) as proposed by \citet{gutmann2016bayesian}.

% idea overview

BO is a sequential optimization method that utilizes a model fitted on the previous evaluations to predict evaluation outcomes and decide the next parameter values to evaluate \citep{shahriari2015taking}.
In simulator-based inference, the optimization task is to find parameter values that minimize the expected distance or discrepancy between observed and simulated data \citep{gutmann2016bayesian}.
We measure the distance using JSD and use the surrogate model fitted in optimization to calculate $\avejsd$ at selected candidate parameter values.
%
% setup
%
The experiments presented in this work were carried out with the BOLFI implementation available in ELFI \citep{lintusaari2018elfi}.
Gaussian process regression with a normal likelihood and squared exponential kernel was used to model the dependencies between simulator parameters and JSD between observed and simulated data.
The kernel associates each parameter dimension with a lengthscale estimated based on the available data, and we used gamma prior distributions to bias the estimates towards reasonable values in the initial iterations when observed data is scarce.
%To bias the estimates towards reasonable values in the initial iterations when observed data is scarce, we used gamma prior distributions with shape $a=2$ and rate $b=10/h$, where $h$ denotes the parameter optimization range.
%
In addition we used the theoretical bounds in Equation~(\ref{range2}) to normalize the output range to $[-1,1]$, and assumed zero mean in the regression model.
The parameter values evaluated in each iteration were selected based on the lower confidence bound acquisition rule.

% ess

When we cannot assume that the observed and simulated data are sampled from a fixed multinomial distribution, we also use simulated data to estimate an effective sample size.
The estimates used in this work are calculated based on $\nrep$ simulated data sets generated with parameters $\param$ as 
\begin{equation}
\textrm{ESS}=\frac{\sum_{\icls=1}^\ncls \bar q_\icls(1-\bar q_\icls)}{\frac{1}{m}\sum_{\icls=1}^\ncls\sum_{j=1}^\nrep (\hat q_\icls^{(j)}-\bar q_\icls)^2},
\end{equation}
where $\hat q_\icls^{(j)}$ denotes the observed category proportion in simulated set $j$ and $\bar q_\icls=\frac{1}{\nrep}\sum_{j=1}^\nrep \hat q_\icls^{(j)}$.
The idea is that the expected variance calculated based on the estimated ESS and multinomial distribution assumption should match with the variance observed between simulated samples.
For alternative approaches, see for example \cite{candy2008estimation}.

\subsection{Evaluation}\label{sec:evaluation}

% models

We carried out experiments with three simulator models that produce categorical observation data as discussed in Section \ref{sec:phiestimat11}.
The present work focuses on hypothesis testing and confidence set estimation in simulator-based models with intractable likelihoods and unknown mapping between parameter values and category probabilities.
However, to evaluate the proposed test statistic, we ran experiments with two simulator models where the mapping between parameter values and category probabilities is known (Sections~\ref{sec:experiment1}--\ref{sec:experiment2}).
This allows comparison between the proposed simulator-based approximation and the standard Pearson statistic discussed in Section \ref{secttjitvaa2}.
In addition we %compare test statistics calculated based on the BOLFI model to test statistics calculated based on $\nrep=1000$ simulations, and
run experiments with the simulator model studied by \citet{corander2017frequency} which has an intractable likelihood.

% simulations

We used the simulator models to generate 1000 observation sets with fixed parameter values and calculated the proposed test statistic based on each observation set as discussed in the previous section.
The proportion of experiments where true parameter values are rejected should approach $\alpha$.
Hence to evaluate the proposed test statistic, we count the experiments where the true parameter value is accepted and included in the confidence set at significance levels $\alpha=\{0.50, 0.10, 0.05, 0.01\}$ that correspond to confidence levels $\{50, 90, 95, 99\}$~\%.
The observed proportions, referred to  as coverage probabilities, are compared to the expected level $1-\alpha$.

\subsection{Experiment 1}\label{sec:experiment1}
 
% simulator model

The simulator studied in this experiment is a multinomial distribution with $\ncls$ categories and observation probabilities $\pcls(\param)$ calculated as 
\begin{equation}
\pcls(\param)=\sigma(-\param(\icls-1)), \quad i = 1,\ldots,\ncls,
\end{equation}
where $\sigma(z_\icls)$ denotes the normalized exponential function output $\sigma(z_\icls)=\exp(z_\icls)/\sum_\icls\exp(z_\icls)$.
The expected observation counts either decrease with the category index $\icls$ when $\param>0$ or increase when $\param<0$.
The setup used in this example is $\ncls=5$ and $\param=0.2$.
The corresponding category probabilities are visualized in Figure~\ref{fig:event_probabilities_1}.
\begin{figure}
    \centering
    \includegraphics[width=0.35\textwidth]{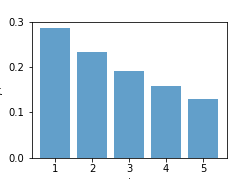}
    \caption{$P_o$ in the experiments carried out with the 1-parameter example model.}
    \label{fig:event_probabilities_1}
\end{figure}
We used the model to simulate 1000 observation sets with $\nobs=\{50, 100, 500, 1000\}$ samples and calculated the proposed test statistics based on the observation sets and simulations with $\nsms=\nobs$ and $\nsms=1000\nobs$ samples.

We start with experiments where the proposed test statistic values are calculated with $\avejsd$ estimated based on $\nrep=1000$ simulations carried out with the true parameter value.
A visual comparison between the expected null distribution and observed test statistic values indicates a reasonable fit in all test conditions (Figure~\ref{fig:chi_squared_visual_1}).
\begin{figure}[p]
    \centering
    \makebox[\textwidth][c]{(a)\vcenteredhbox{\includegraphics[width=\textwidth]{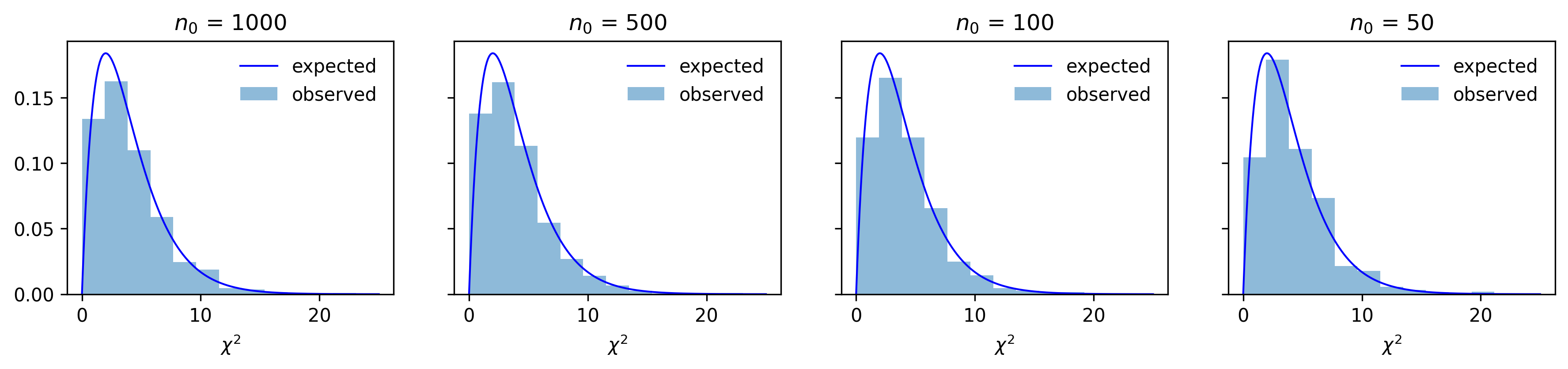}}} \-
    \makebox[\textwidth][c]{(b)\vcenteredhbox{\includegraphics[width=\textwidth]{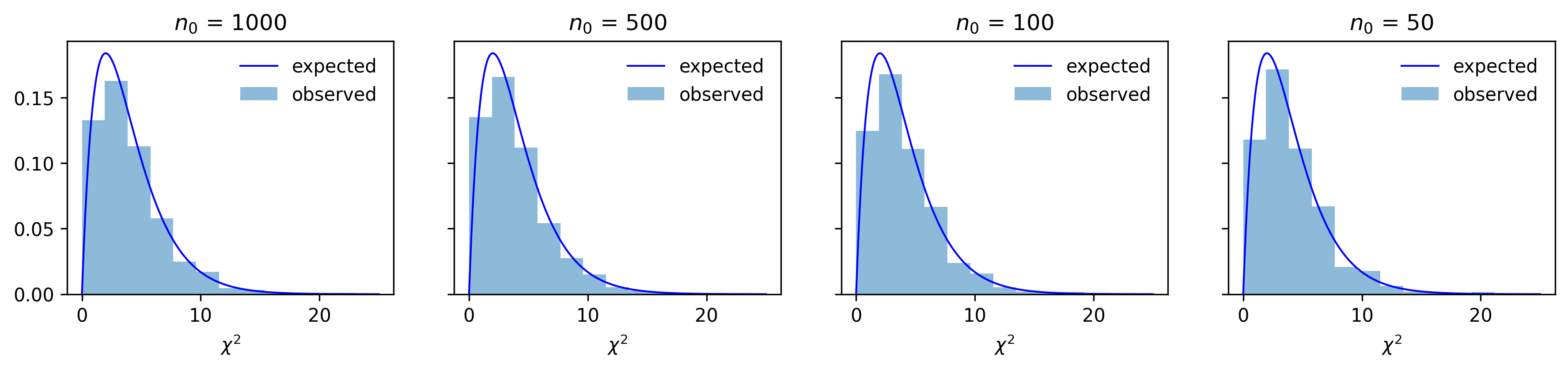}}}
    \makebox[\textwidth][c]{(c)\vcenteredhbox{\includegraphics[width=\textwidth]{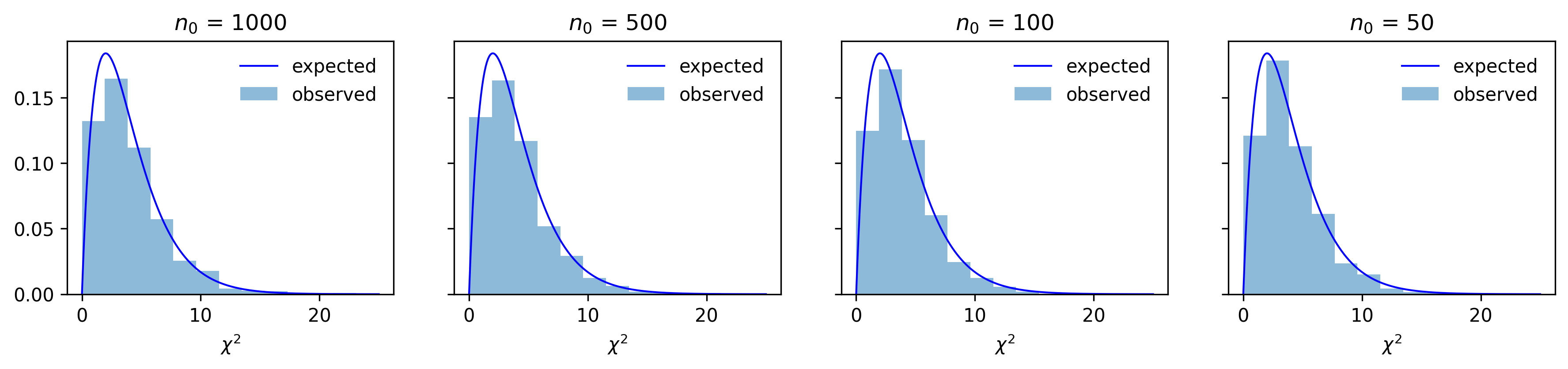}}}
    \caption{The expected asymptotic distribution and test statistics calculated based on observations simulated with 1-parameter example model. The proposed test statistic values were calculated based on $\nrep=1000$ simulations with (a) $\nsms=\nobs$ or (b) $\nsms=1000\nobs$ samples.
    For comparison, see (c) Pearson statistics calculated based on the true category probabilities.}
    \label{fig:chi_squared_visual_1}
\end{figure}
This is confirmed when we examine the coverage probabilities presented in Figure~\ref{fig:chi_squared_cov_probs_1}.
\begin{figure}[p]
    \centering
    \makebox[\textwidth][c]{\includegraphics[width=0.95\textwidth]{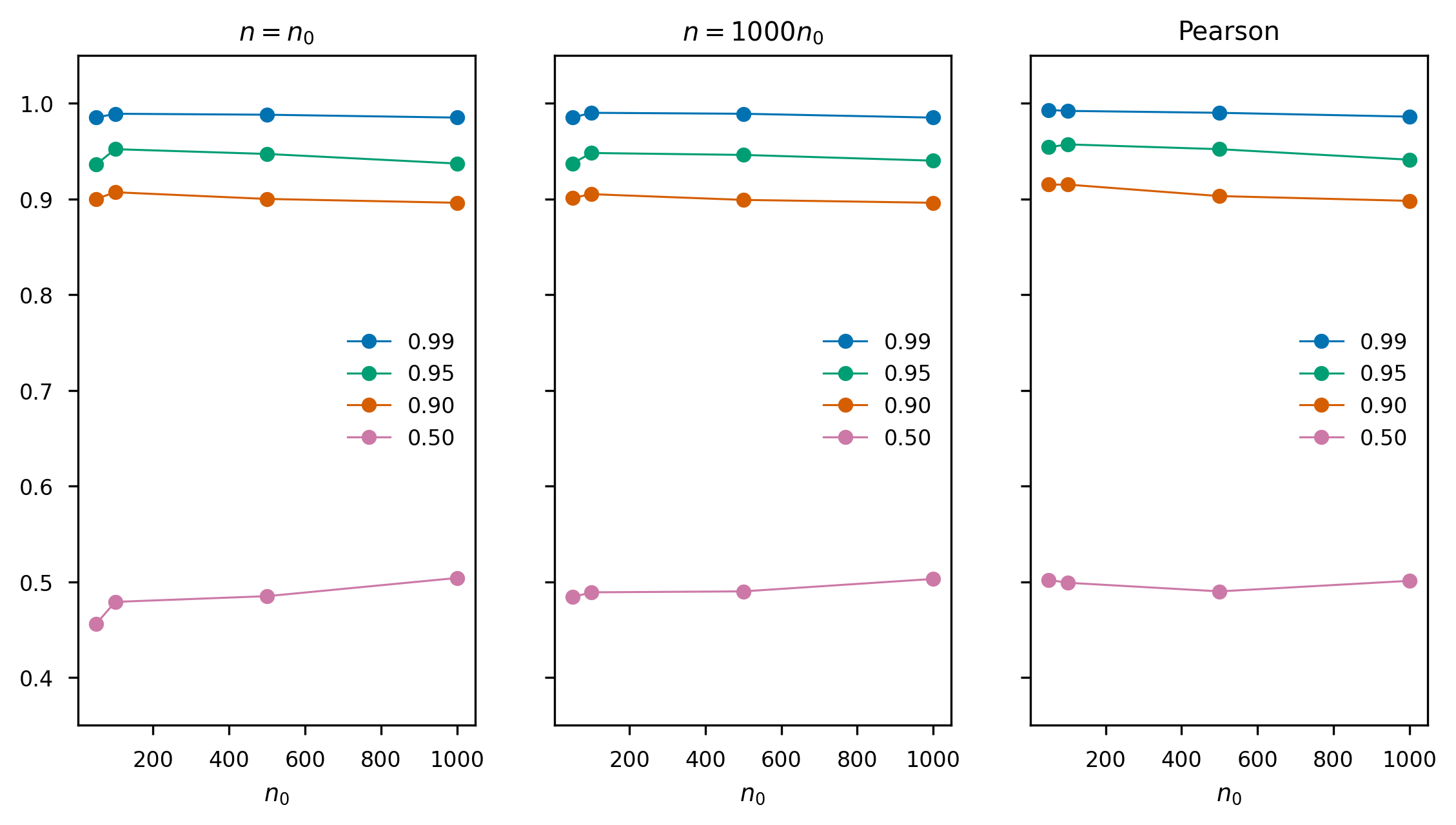}}
    \caption{Coverage probabilities calculated based on observations simulated with the 1-parameter example model. The proposed test statistic values were calculated based $\nrep=1000$ simulations with $\nsms=\nobs$ or $\nsms=1000\nobs$ samples. Coverage probabilities calculated with the Pearson statistic are presented for comparison.}
    \label{fig:chi_squared_cov_probs_1}
\end{figure}
We observe that the coverage probabilities are close to the nominal value $1-\alpha$ in all test conditions, and while the coverage probabilities are lower than the nominal value in some conditions when the observed and simulated sample are small, the coverage probabilities in all test conditions converge to the nominal values when $\nobs$ increases.

% test 2

We then proceed to evaluate test statistics calculated based on $\avejsd$ estimated using the surrogate model in BO.
BO was initialized with 20 simulations and the total simulation count was set to 1000.
The parameter values included in the initialization set were selected at random within $[-1, 1]$ and optimization was carried out in this range.
The coverage probabilities calculated based on the proposed test statistic are presented in Figure~\ref{fig:chi_squared_cov_probs_1_BOLFI}.
\begin{figure}
    \centering
    \makebox[\textwidth][c]{\includegraphics[width=0.65\textwidth]{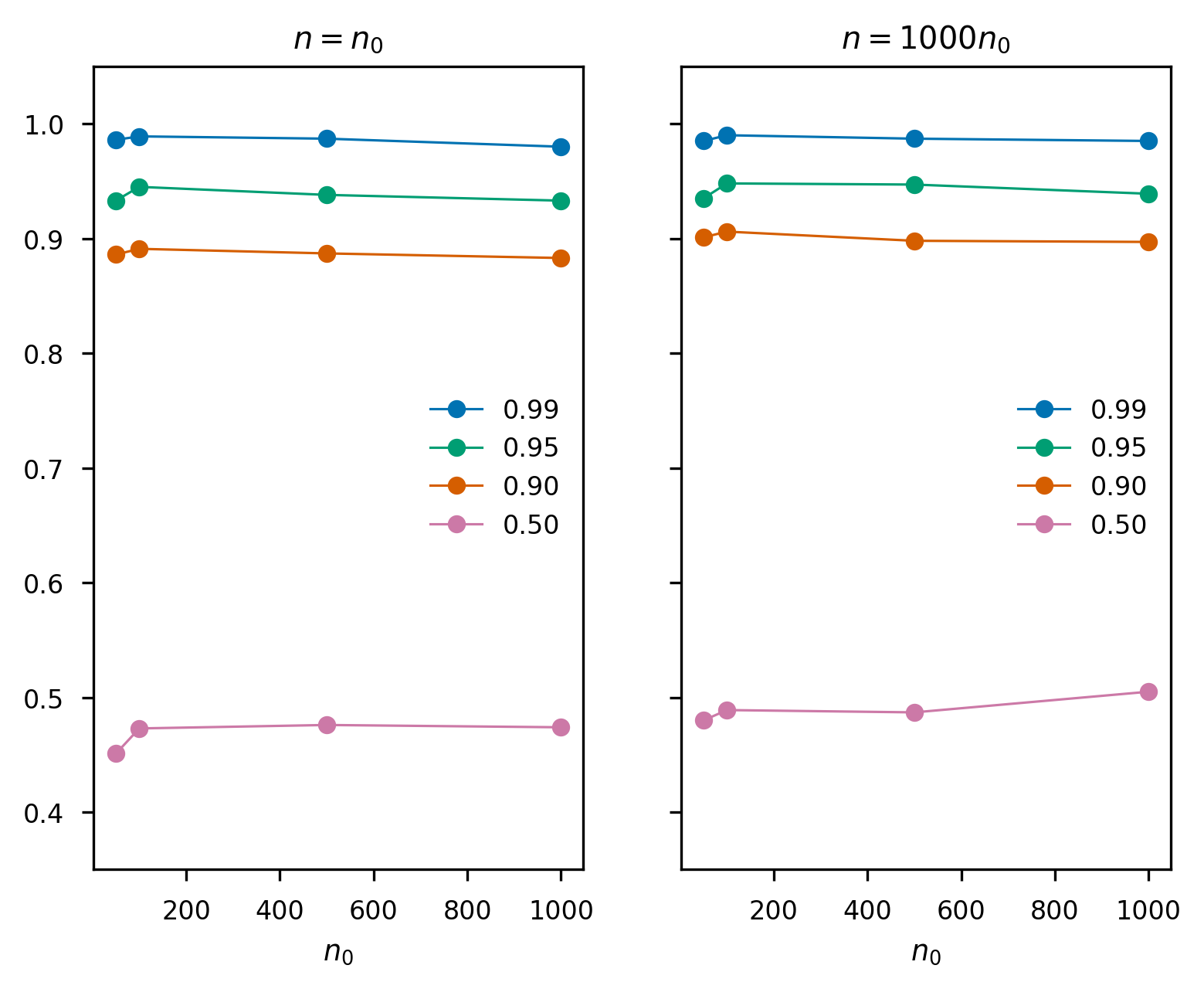}}
    \caption{Coverage probabilities calculated based on observations simulated with the 1-parameter example model when the test statistics were calculated based on the BOLFI model using simulations with $\nsms=\nobs$ or $\nsms=1000\nobs$ samples.}
    \label{fig:chi_squared_cov_probs_1_BOLFI}
\end{figure}
The coverage probabilities are close to $1-\alpha$ in all test conditions, but
in contrast to the results presented in Figure~\ref{fig:chi_squared_cov_probs_1}, %the coverage probabilities in Figure~\ref{fig:chi_squared_cov_probs_1_BOLFI} do not converge to the nominal value $1-\alpha$ when $\nobs$ increases if $\nsms=\nobs$.
do not converge to the nominal value when $\nobs$ increases if $\nsms=\nobs$.
This because the $\avejsd$ calculated based on the surrogate model is not as accurate as the estimate calculated based on $\nrep=1000$ simulations, and while the estimates become more accurate when $\nsms$ increases and there is less variation between simulations, the estimation errors are multiplied with $\nobs$ when we calculate the test statistic, meaning that an increase in sample size may not decrease the error in test statistic values when $\nsms=\nobs$ (Figure~\ref{fig:BOLFI_error_1}).
\begin{figure}
    \centering
    \makebox[\textwidth][c]{(a)\vcenteredhbox{\includegraphics[width=0.9\textwidth]{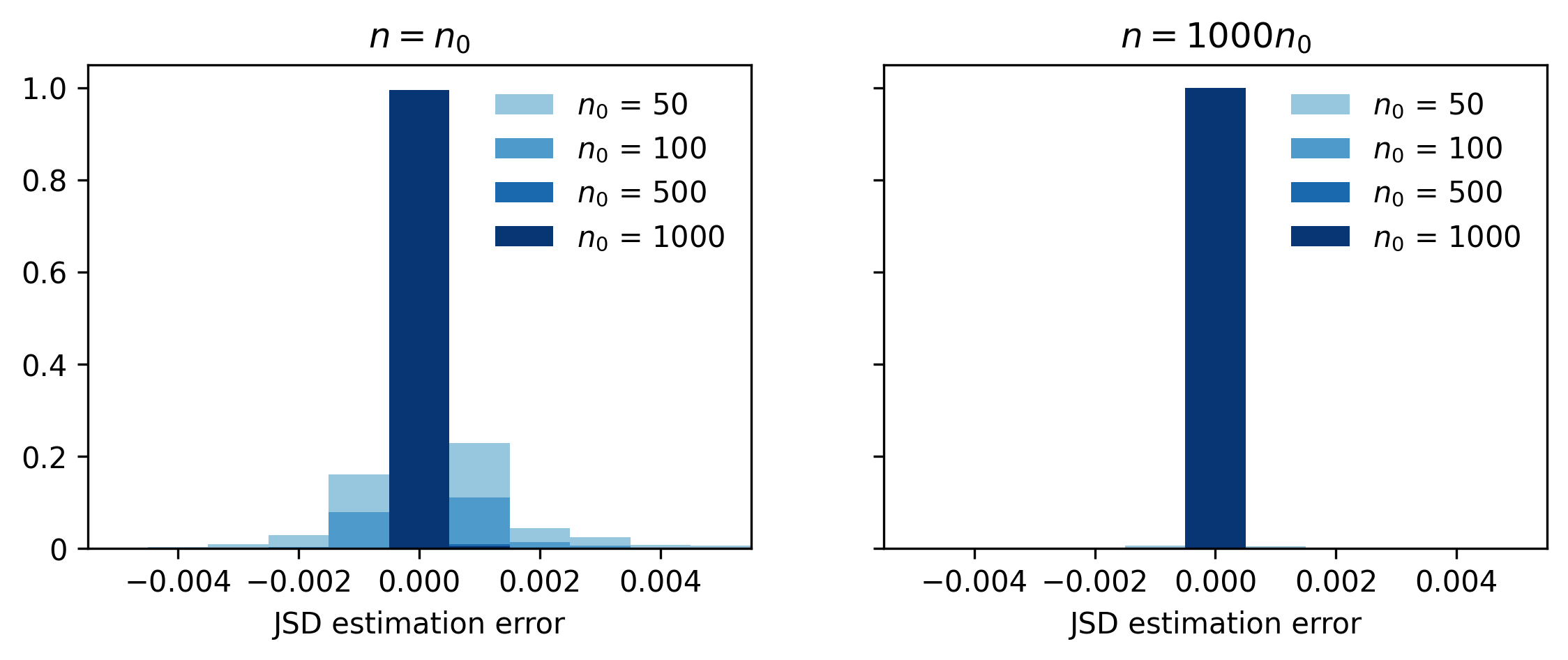}}} \-
    \makebox[\textwidth][c]{(b)\vcenteredhbox{\includegraphics[width=0.9\textwidth]{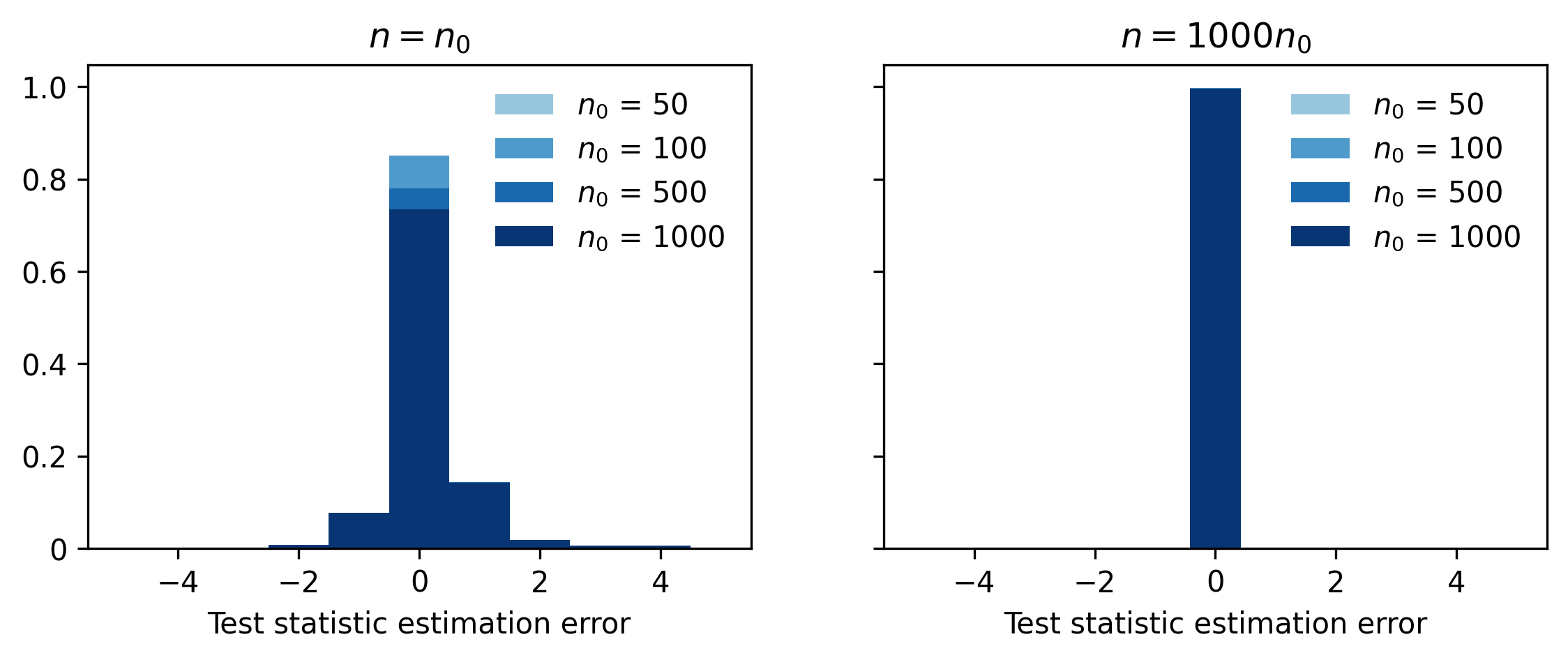}}}
    \caption{Estimation errors calculated as difference between (a)~simulator-based JSD and (b)~test statistic values calculated based on the BOLFI model or $\nrep=1000$ simulations.}
    \label{fig:BOLFI_error_1}
\end{figure}

While using a surrogate model can introduce estimation error in the test statistic values, it has the benefit that we can test large candidate sets without additional simulation cost.
%
%This would be the case when we want to estimate a confidence interval or confidence set over parameter values.
%
Figure~\ref{fig:confidence_intervals_1} shows coverage probabilities calculated with respect to the model parameter $\param$.
\begin{figure}
    \centering
    \makebox[\textwidth][c]{(a)\vcenteredhbox{\includegraphics[width=\textwidth]{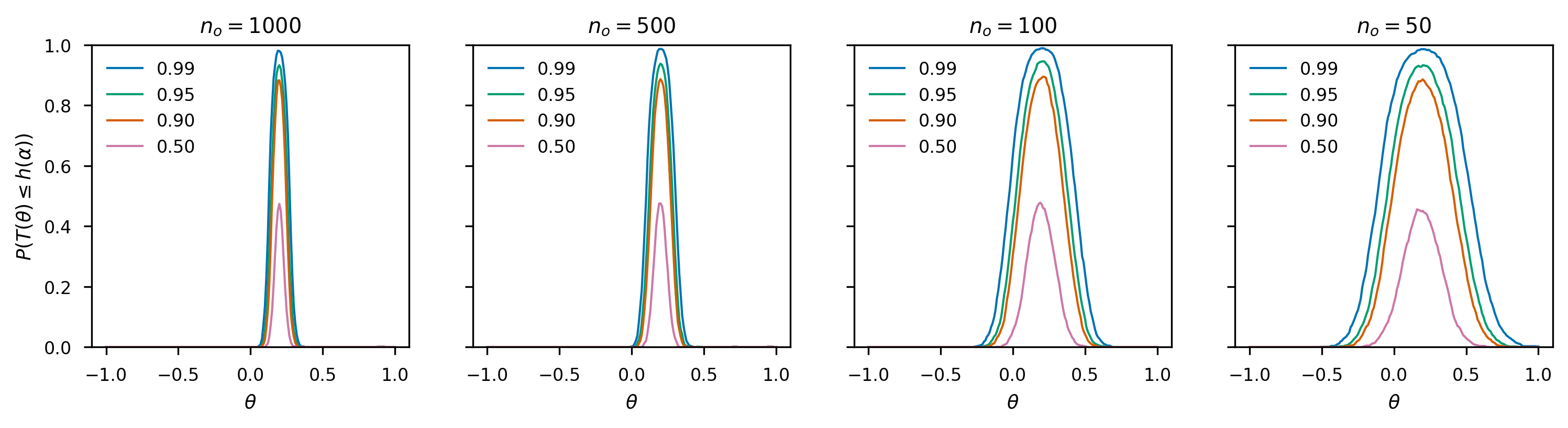}}} \-
    \makebox[\textwidth][c]{(b)\vcenteredhbox{\includegraphics[width=\textwidth]{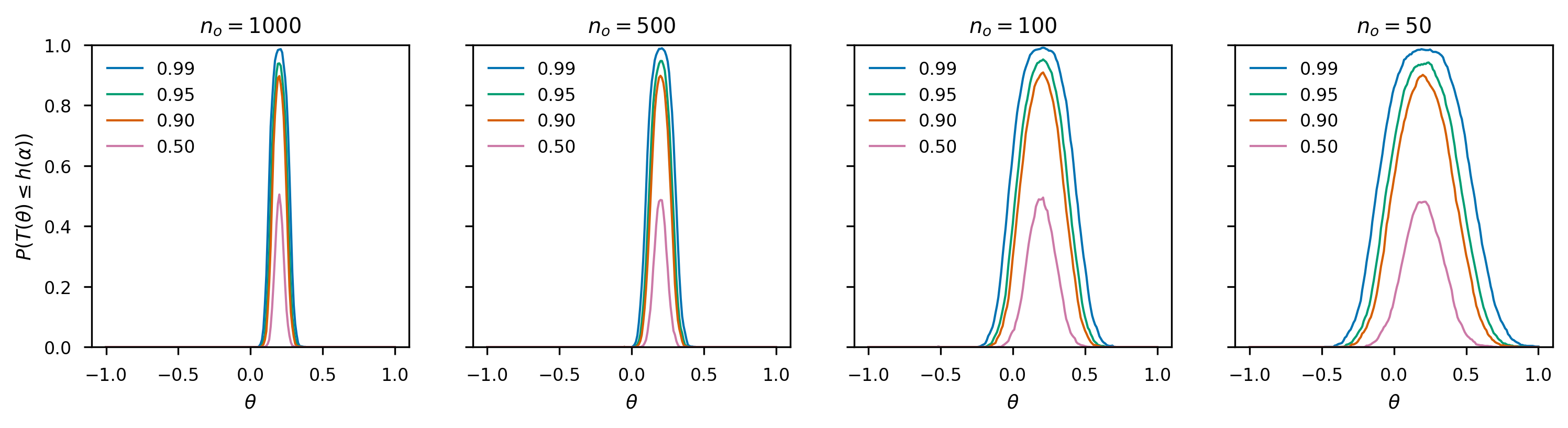}}}
    \caption{Proportion of confidence intervals that include the parameter value $\theta$ when observations were generated with the 1-parameter example model using $\theta=0.2$ and the proposed test statistic values were calculated based on the BOLFI model using simulations with (a) $\nsms=\nobs$ or (b) $\nsms=1000\nobs$ samples.}
    \label{fig:confidence_intervals_1}
\end{figure}
We see that the confidence intervals estimated based on the proposed test statistic are expected to cover a range around the true parameter value, and that the range becomes wider when we reduce the observed data set size $\nobs$.

\subsection{Experiment 2}\label{sec:experiment2}

\newcommand{\lx}{\lambda^X} 
\newcommand{\ly}{\lambda^Y} 
\newcommand{\lxy}{\lambda^{XY}} 

% simulator model

The second experiments are carried out with the standard log-linear model that is used to describe association and interaction patterns between two categorical random variables.
Here we model the counts in a two-way table as a sample from a multinomial distribution with $k=4$ categories.
The observation probabilities are calculated based on the normalized exponential function as
\begin{equation}
\pcls(\param)=\sigma(\log(\mu_i)) = \sigma(\lambda+ X_i\lx+ Y_i\ly+X_iY_i\lxy), \quad i = 1,2,3,4,
\end{equation}
where $\sigma(z_\icls)$ denotes the normalized exponential function output $\sigma(z_\icls)=\exp(z_\icls)/\sum_\icls\exp(z_\icls)$,
and $X_i$ and $Y_i$ are coded variable values and $\mu_i$ denote the expected observation counts.
We assume effect-coded variables that take values 1 or -1 as indicated in Table~\ref{tab:effectcode}.
\begin{table}
    \centering
    \begin{tabular}{ccccc}
    $\icls$ & 1 & 2 & 3 & 4  \\\hline
    $X_i$ & 1 & 1 & -1 & -1 \\ 
    $Y_i$ & 1 & -1 & 1 & -1 \\
    \end{tabular}
    \caption{\label{tab:effectcode}Effect coding in the log-linear example.}
\end{table}
The model parameters $\lx$ and $\ly$ then encode expected difference in the proportion between 1 and -1 values in variables $X$ and $Y$, and the parameter $\lxy$ encodes possible association between the two variable values.
Finally the constant $\lambda$ is calculated based on the other parameter values and total count $n$ so that the sum over expected counts equals $n$.

% setup

We run experiments with two model versions.
We use a two-parameter model where $\lxy=0$ and the model parameters $\param=(\lx,\ly)$ and a saturated model where the model parameters $\param=(\lx,\ly,\lxy)$.
The true parameter values are set to $\lx=-0.25$ and $\ly=0.15$ in the two-parameter version and to $\lx=-0.20$, $\ly=0.10$, and $\lxy=0.40$ in the saturated three-parameter version.
The corresponding category probabilities are visualized in Figure \ref{fig:event_probabilities_2}.
\begin{figure}
    \centering
    \begin{tabular}{cc}
    (a) & (b) \\
    \includegraphics[width=0.35\textwidth]{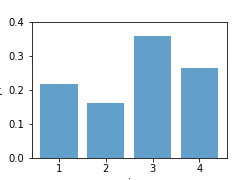} &
    \includegraphics[width=0.35\textwidth]{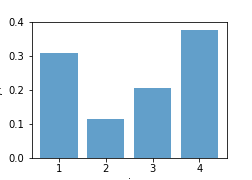}
    \end{tabular}
    \caption{$P_o$ in the experiments carried out with the (a) 2-parameter and (b) 3-parameter log-linear model.}
    \label{fig:event_probabilities_2}
\end{figure}
Both model versions were used to simulate 1000 observation sets with $\nobs=\{50,100,500,1000\}$ samples.

% test 1

We run the same experiments that were carried out with the 1-parameter example model studied in the previous section.
We start with the proposed test statistics calculated based on $\nrep=1000$ simulations carried out with the true parameter values with sample size $\nsms=\nobs$ and $\nsms=1000\nobs$.
Figure \ref{fig:chi_squared_cov_probs_2} shows coverage probabilities calculated based on comparison between the observed test statistic values and the expected null distribution.
\begin{figure}
    \centering
    \makebox[\textwidth][c]{(a)\vcenteredhbox{\includegraphics[width=0.95\textwidth]{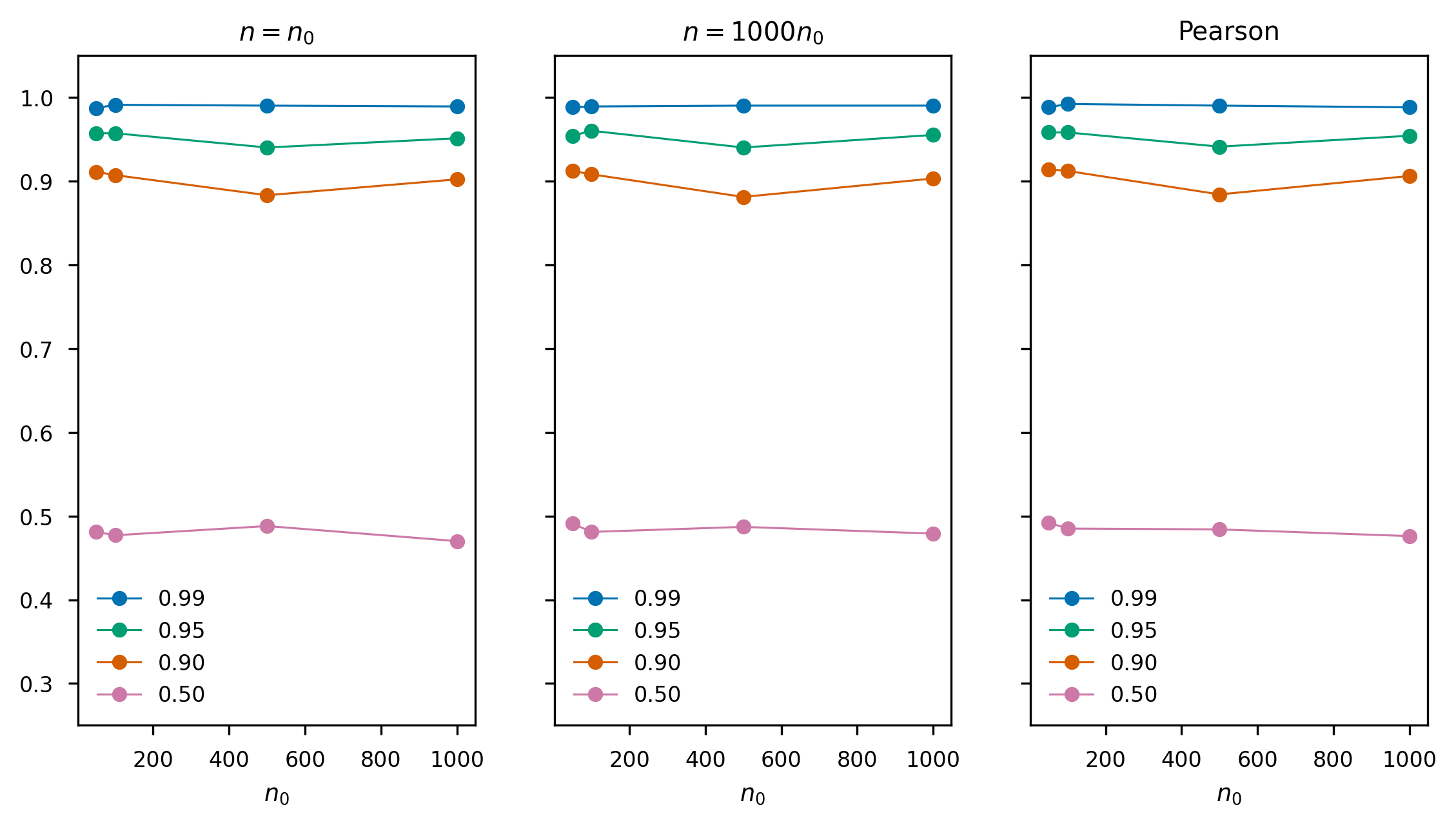}}}
    \makebox[\textwidth][c]{(b)\vcenteredhbox{\includegraphics[width=0.95\textwidth]{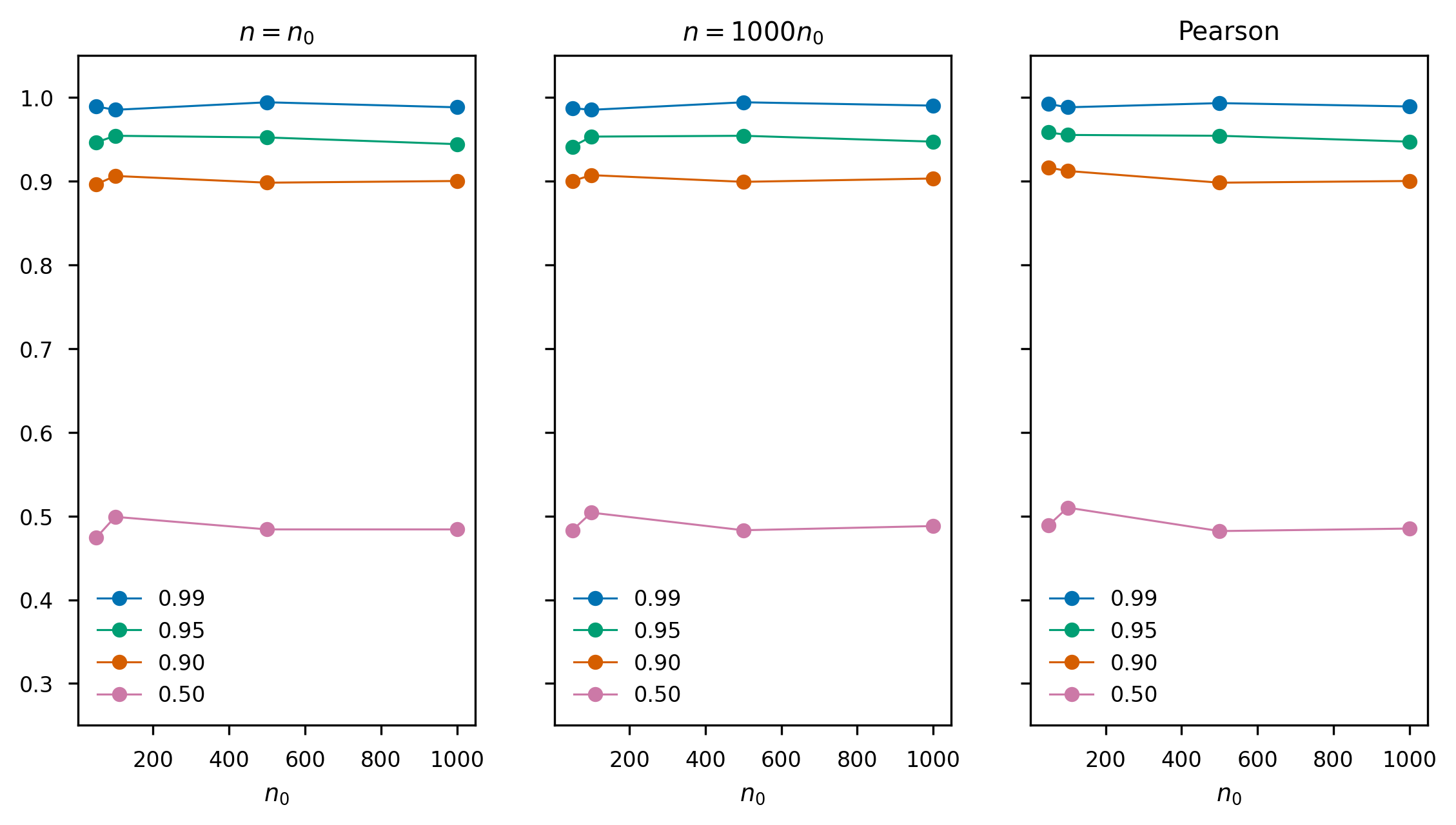}}}
    \caption{Coverage probabilities calculated based on observations simulated with the (a) 2-parameter or (b) 3-parameter log-linear model. The proposed test statistic values were calculated based $\nrep=1000$ simulations with $\nsms=\nobs$ or $\nsms=1000\nobs$ samples. Coverage probabilities calculated with the Pearson statistic are presented for comparison.
    }
    \label{fig:chi_squared_cov_probs_2}
\end{figure}
The coverage probabilities calculated based on the proposed test statistic follow the coverage probabilities calculated with the Pearson statistic and are close to the nominal value $1-\alpha$ in all test conditions.

% test 2

We also evaluate the proposed test statistic values calculated based on the surrogate model in BO.
In this experiment we initialized BO with 50 simulations and set the total simulation count to 2000. 
The parameter values included in the initialization set were selected at random within $[-1, 1]$ and optimization was carried out in this range.
Coverage probabilities calculated based on the proposed test statistic are visualized in Figure~\ref{fig:chi_squared_cov_probs_2_BOLFI}.
\begin{figure}
\centering
    (a) \vcenteredhbox{\includegraphics[width=0.65\textwidth]{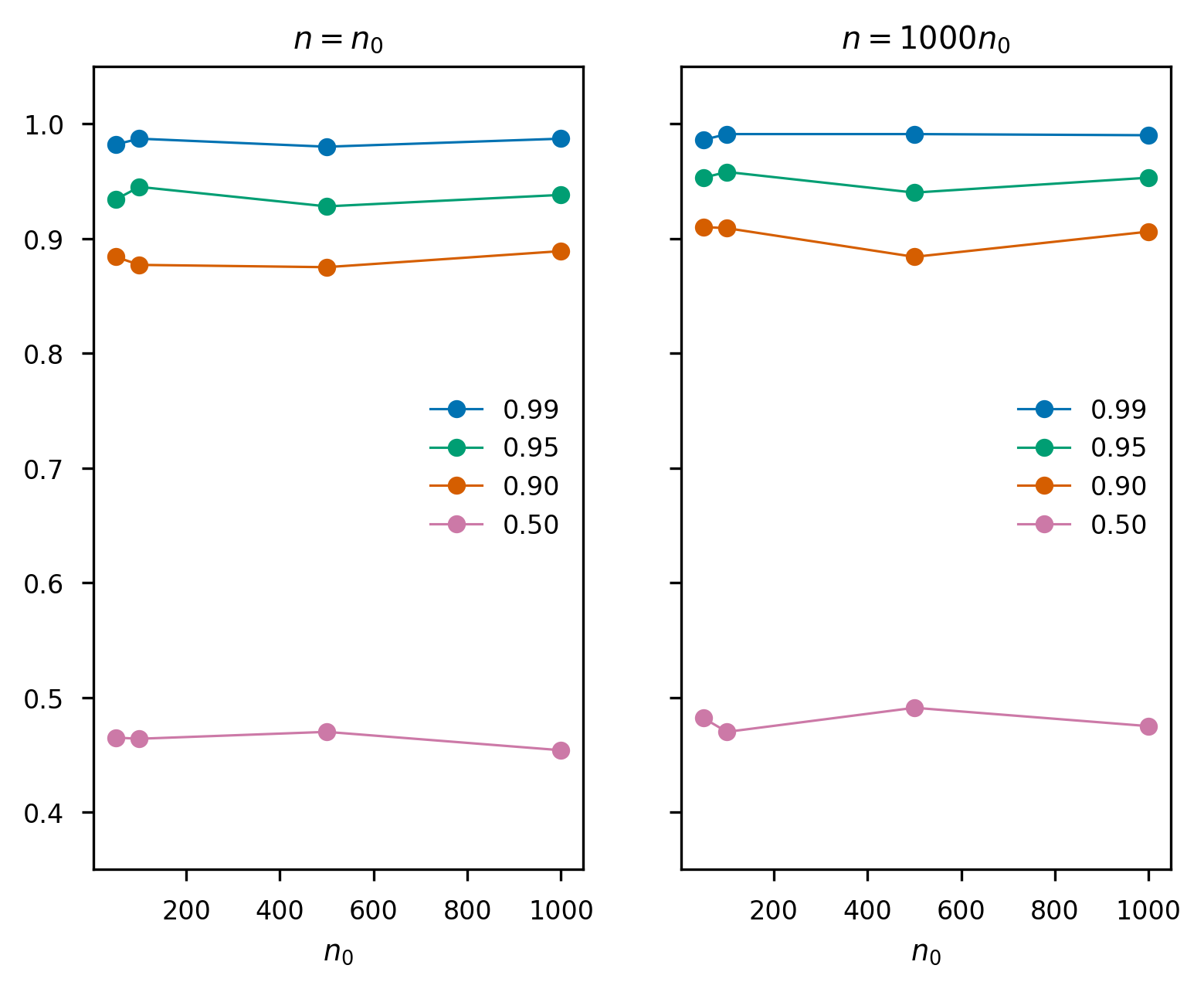}}
    
    (b) \vcenteredhbox{\includegraphics[width=0.65\textwidth]{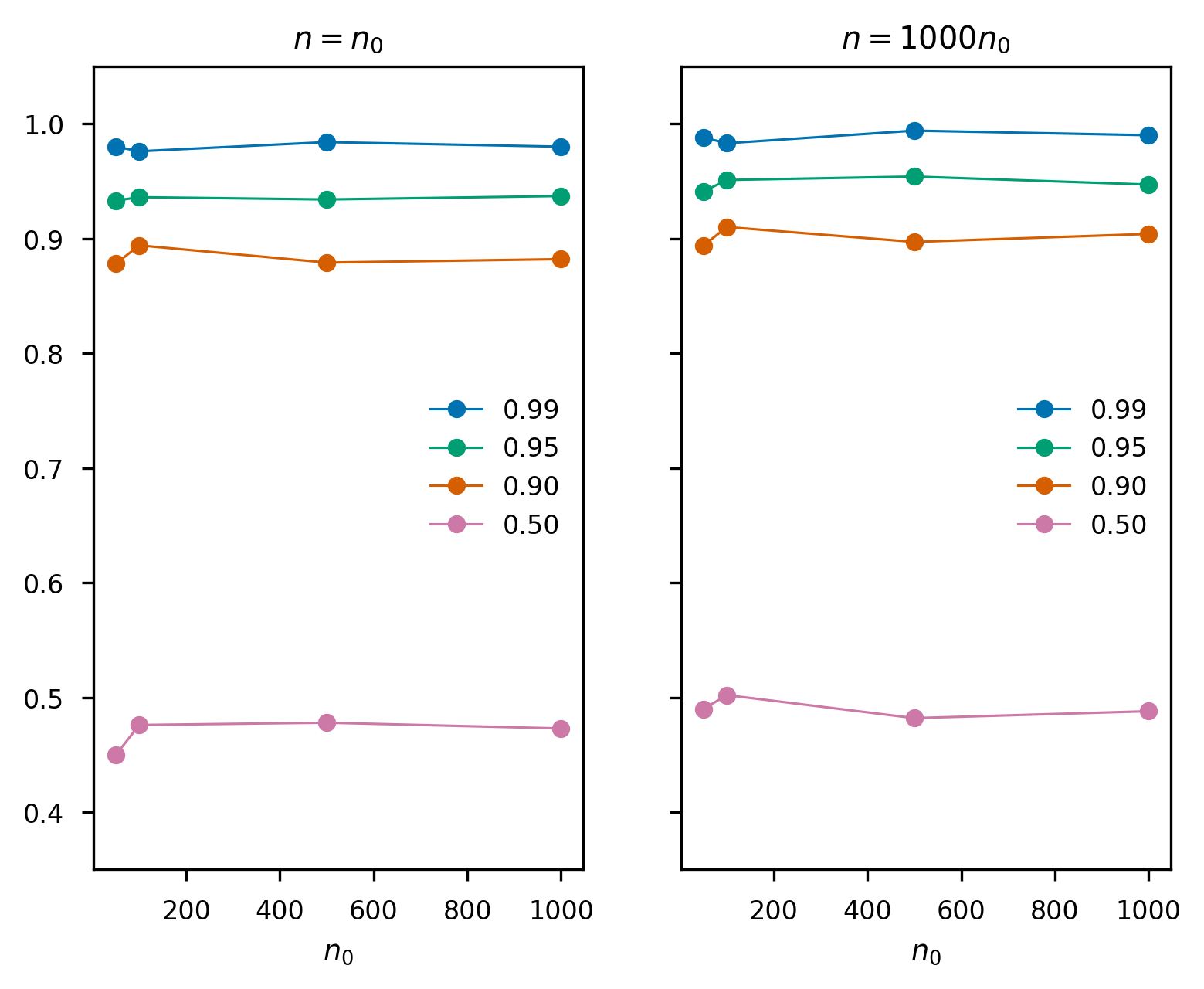}}
    \caption{Coverage probabilities calculated based on observations simulated with the (a) 2-parameter or (b) 3-parameter log-linear model when the test statistics were calculated based on the BOLFI model using simulations with $\nsms=\nobs$ or $\nsms=1000\nobs$ samples.}
    \label{fig:chi_squared_cov_probs_2_BOLFI}
\end{figure}
We observe that the coverage are close to the nominal value in all test conditions, but do not converge to $1-\alpha$ when $\nobs$ increases if $\nsms=\nobs$.

\subsection{Experiment 3}\label{sec:experiment3}

The last experiment is carried out with a model that simulates the evolution of genotype frequencies.
\citet{corander2017frequency} used the model to capture non-negative frequency dependent selection (NFDS) in the post-vaccine evolution of pneumococcal populations.
The simulator code is available online and the present experiments are carried out with the model version that simulates homogeneous-rate multilocus NFDS.
Evolution is modeled as a discrete-time process where the population at time $t+1$ is sampled with replacement from the population at time $t$ with observation counts
\begin{equation}
X_{i,t} \sim \mathrm{Poisson}\left(\frac{\kappa}{N_t}(1-m)(1-v_i)(1+\sigma_f)^{\pi_{i,t}} \right),
\end{equation}
where $i$ indexes the isolates in population at time $t$.
The first term accounts for general density-dependent selection where $\kappa$ denotes the carrying capacity and $N_t$ is the population size at time $t$. 
In the current experiment we assumed $\kappa=10^5$. 
The second term with migration rate $m$ models the pressure from migration into the population.
The third term describes negative selection pressure due to the vaccine: $v_i=v$ if the isolate $i$ has vaccine serotype and zero otherwise. 
Finally the last term describes the positive selection pressure associated with rare alleles under NFDS: $\pi_{i,t}$ measures the deviation between isolate $i$ and the equilibrium genotype at time $t$, and the selection pressure is  modeled with parameter $\sigma_f$.
Parameter estimation is carried out in the log-compressed parameter domain with $\param=(\ln(m), \ln(v), \ln(\sigma_f))$.

% observations

We use the model to simulate 1000 observation sets that are modeled on pneumococcal data studied by \citet{corander2017frequency}.
The data set studied in previous work includes a pre-vaccination ($t=0$) sample with 133 isolates and two post-vaccination samples with 203 isolates collected at $t=36$ and 280 isolates collected at $t=72$ \citep{mass}.
We use this data to create the simulated observation sets as follows.
We sample the pre-vaccination data to initialize the simulated population at $t=0$, simulate how the population evolves under selected model parameters, and then sample the simulated population at $t=36$ and $t=72$ to create simulated post-vaccination samples with $\nobs=250$ or $\nobs=1000$ isolates.
The parameter values used in the simulations were $\ln(m)=-5.3$, $\ln(v)=-2.5$, and $\ln(\sigma_f)=-5.3$.

% observation categories

The isolates in the simulated observation sets are arranged into 41 sequence clusters based on genetic content and typed as vaccine type (VT) or non-vaccine type (NVT) \citep{corander2017frequency}.
This creates $\ncls=82$ potential observation categories.
However some sequence clusters are exclusive to vaccine or non-vaccine types, which means that some categories are never observed.
To remove these categories, and to ensure that all categories have adequate observation counts, we collapsed the data into $\ncls=4$ observation categories as follows:
$\icls=1$ includes all VT isolates while $\icls=2$ includes NVT isolates in sequence clusters that do not include VT isolates at $t=0$, $\icls=3$ includes NVT isolates in sequence clusters that include VT isolates at $t=0$, and $\icls=4$ includes NVT isolates in sequence clusters that are not present in the observed data at $t=0$.
The negative selection pressure due to vaccine should then be observed as a decrease in category 1, while migration and NFDS control the balance between categories 2--4.
This is observed in the average category proportions calculated based on the simulated observations sets used in this experiment (Figure~\ref{fig:observed_nfds}).
\begin{figure}
    \centering
    \includegraphics[width=0.7\textwidth]{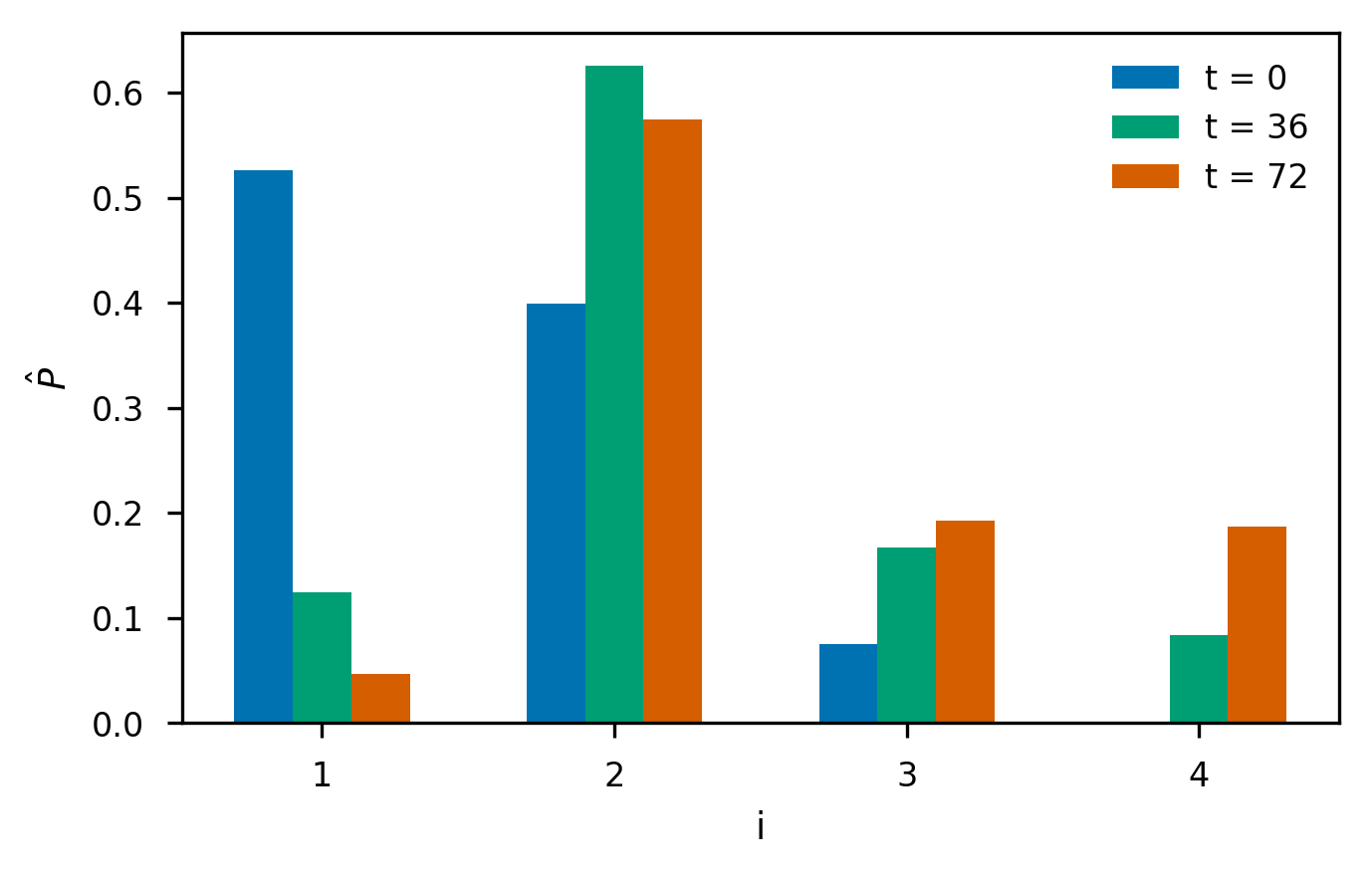}
    \caption{Average category proportions calculated based on 1000 observation sets simulated with the NFDS model.}
    \label{fig:observed_nfds}
\end{figure}

% comparison between (simulated) observed data and simulated data

To summarize, we simulated 1000 observation sets by sampling simulated populations at $t=36$ and $t=72$ and divided the isolates in each sample into $\ncls=4$ observation categories.
These are compared to simulated data based on a discrepancy measure calculated as the sum over JSD between the data collected at $t=36$ and JSD between the data collected at $t=72$.
We calculate the expected discrepancy based on repeated simulations or a BOLFI model, and use it to calculate test statistic values that correspond to the sum over proposed test statistic values calculated based on expected JSD at $t=36$ and $t=72$.

% 1) ave

We start with the proposed test statistic values calculated based on $\nrep=1000$ simulations carried out with the true parameter values and $\nsms=\nobs$.
%
%The observed and simulated data are compared based on the post-vaccination samples collected at $t=36$ and $t=72$, and the test statistic used in this experiment is calculated as a sum over test statistic values calculated at $t=36$ and $t=72$.
%
%This means that the expected distribution is a sum over two $\chi^2(k-1)$ distributions.
%
The expected and observed distributions are compared in Figure \ref{fig:chi_squared_visual_3}~(a) and coverage probabilities reported in Table \ref{tab:chi_squared_cov_probs_3}~(a).
\begin{figure}
    \centering
    (a) \vcenteredhbox{
    \includegraphics[width=0.35\textwidth]{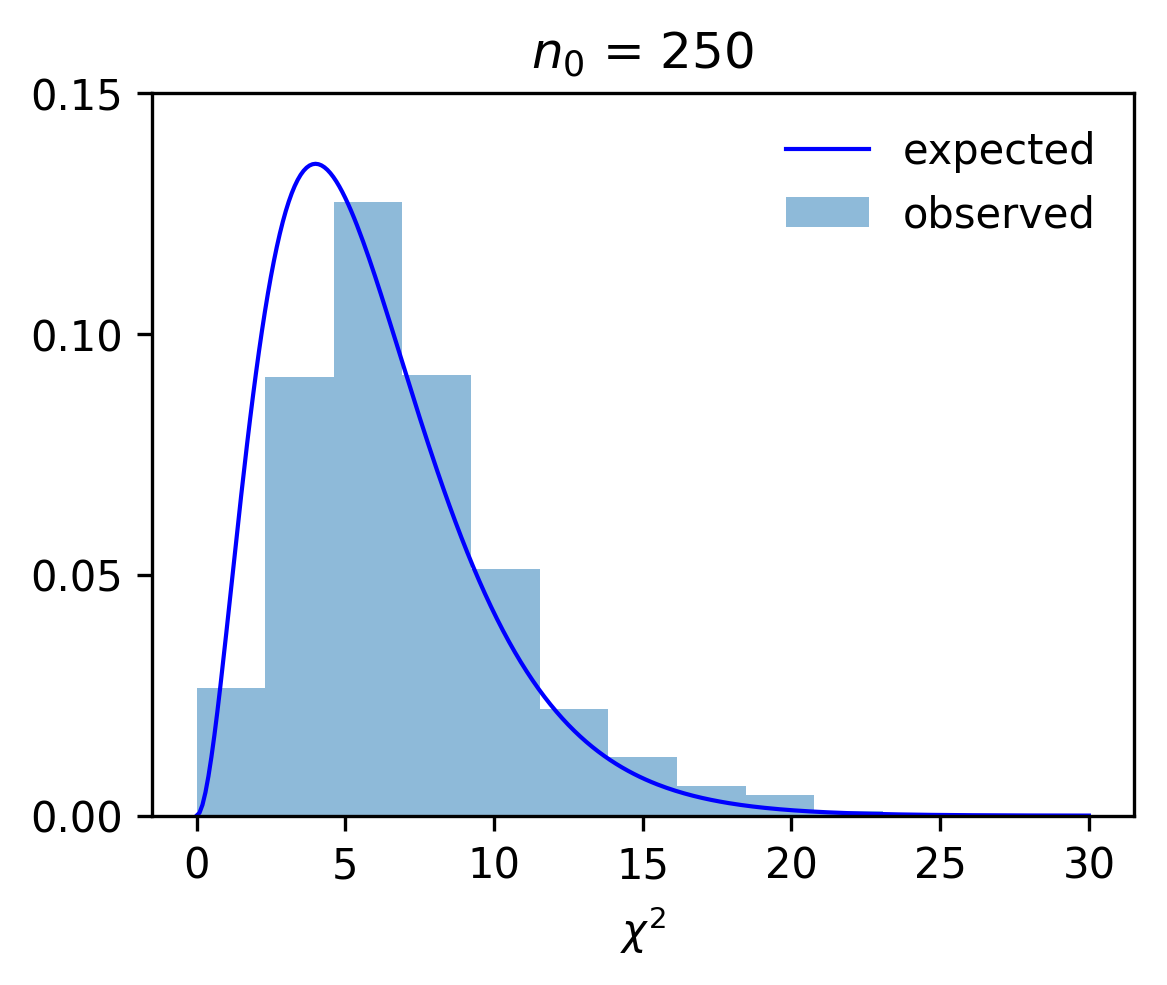}
    \includegraphics[width=0.35\textwidth]{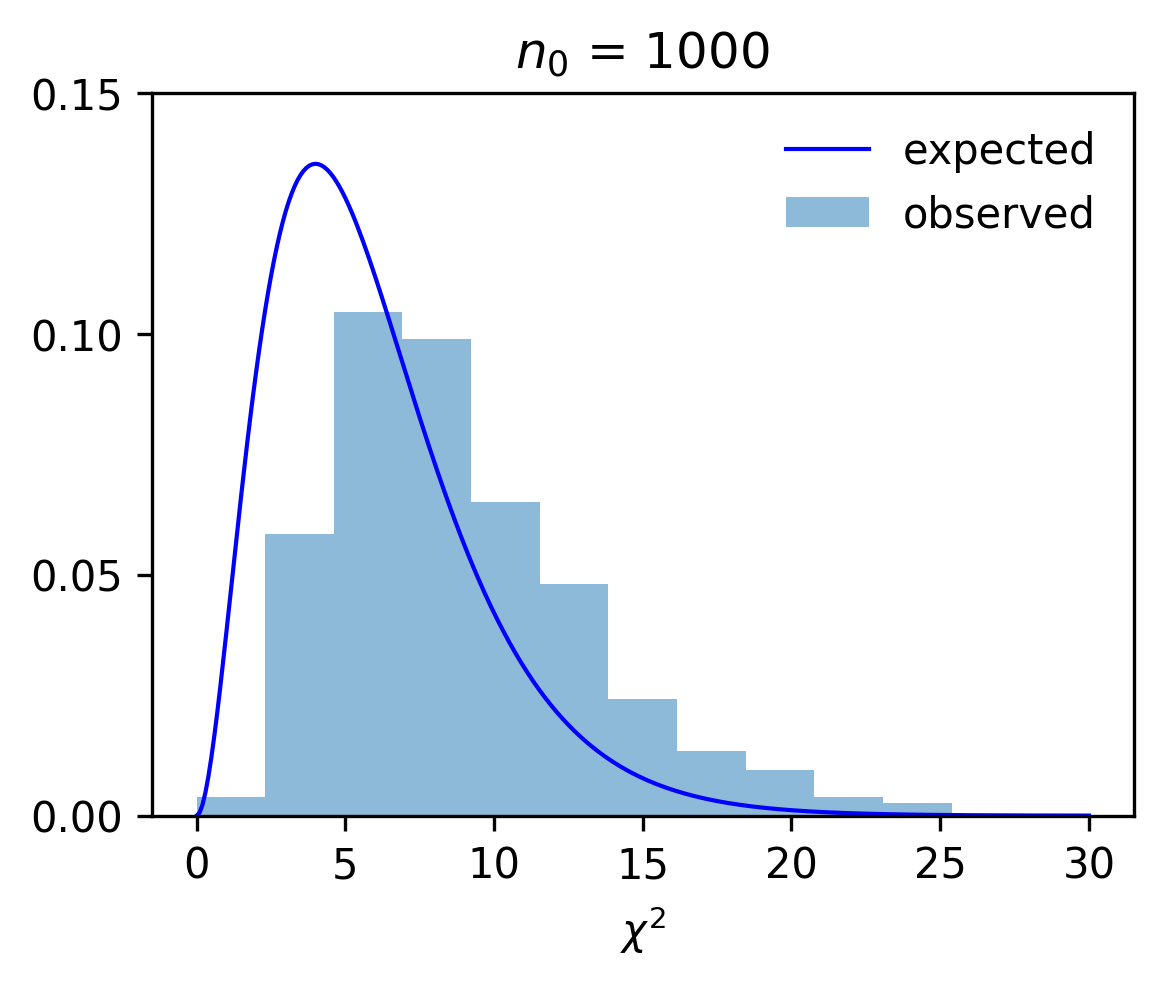}}
    
    (b) \vcenteredhbox{\includegraphics[width=0.35\textwidth]{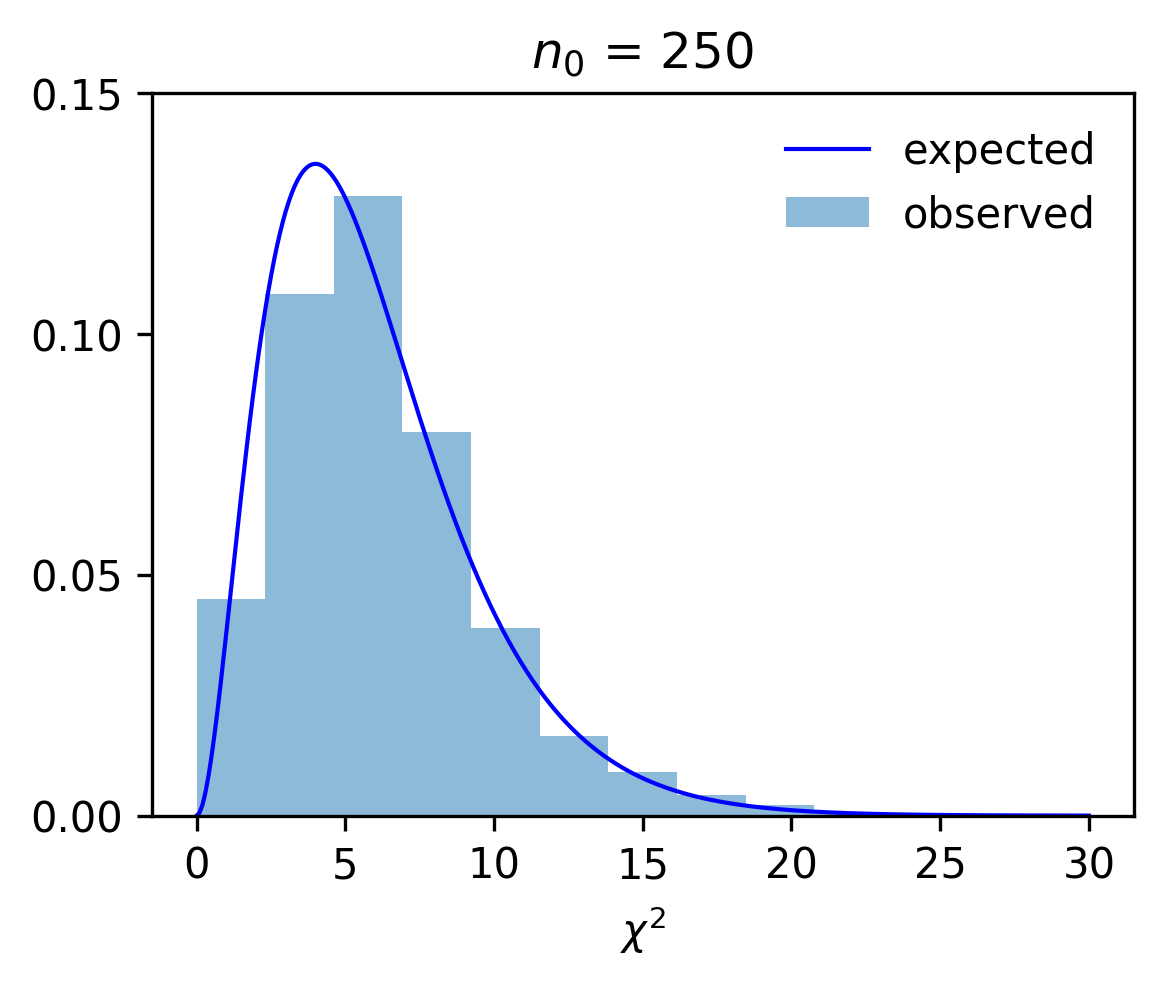}
    \includegraphics[width=0.35\textwidth]{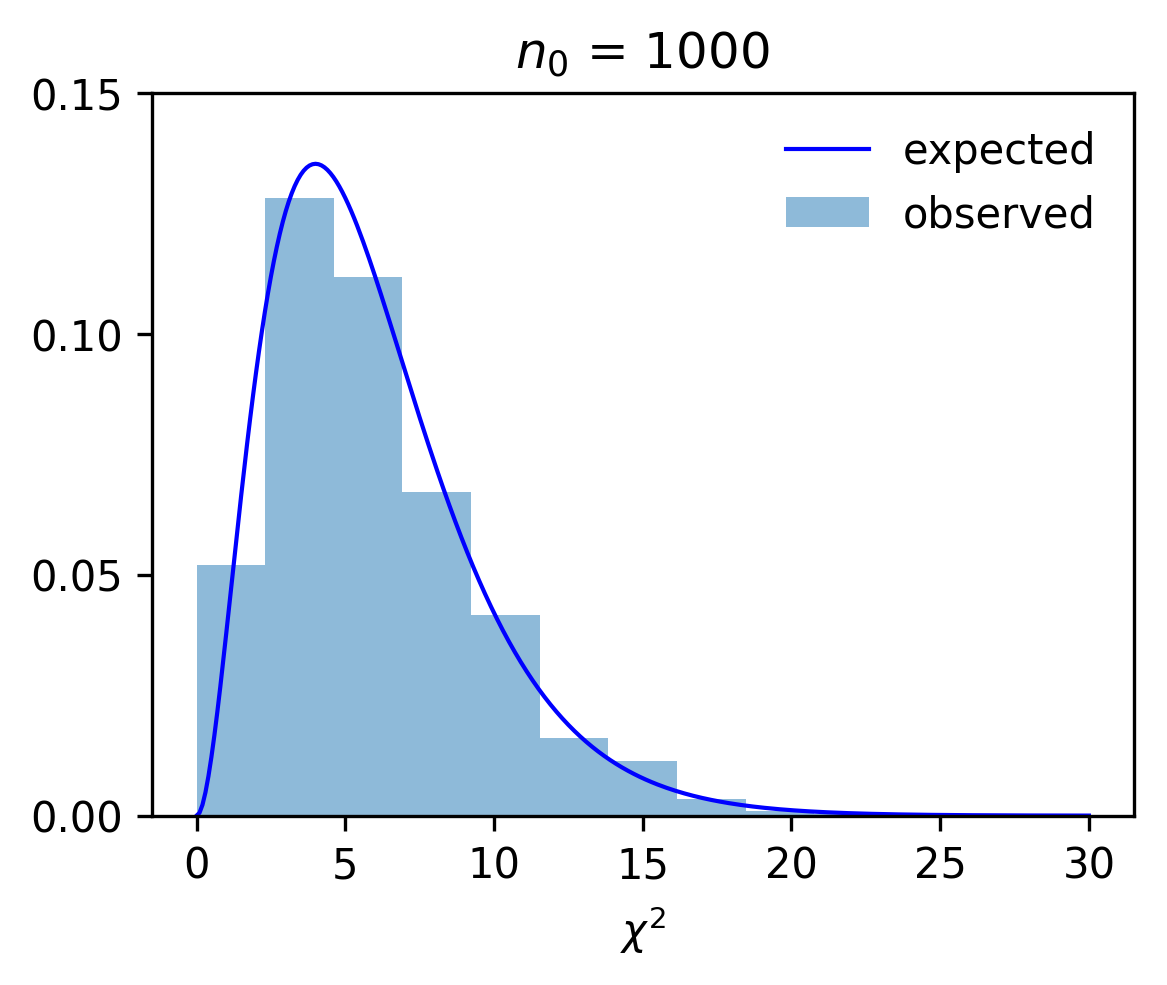}}
    \caption{The expected asymptotic distribution and test statistics calculated based on observations simulated with the NFDS model. The proposed test statistic values were calculated based on $\nrep=1000$ simulations with $\nsms=\nobs$ and based on sample size (a)~$\nobs$ or (b)~ESS.}
    \label{fig:chi_squared_visual_3}
\end{figure}
\begin{table}
    \centering
    \begin{tabular}{cccc}
    (a) &
    \begin{tabular}{ccc}
         $1-\alpha$ &  $\nobs=250$ & $\nobs=1000$ \\\hline
         0.99 & 0.98 & 0.94\\
         0.95 & 0.93 & 0.82\\
         0.90 & 0.86 & 0.71\\
         0.50 & 0.36 & 0.22\\
    \end{tabular} &
    (b) &
    \begin{tabular}{ccc}
         $1-\alpha$ &  $\nobs=250$ & $\nobs=1000$ \\\hline
         0.99 & 0.99 & 0.99\\
         0.95 & 0.94 & 0.94\\
         0.90 & 0.90 & 0.90\\
         0.50 & 0.46 & 0.51\\
    \end{tabular}
    \end{tabular}
    \caption{Coverage probabilities calculated based on observations simulated with the NFDS model. The proposed test statistic values were calculated based on $\nrep=1000$ simulations with $\nsms=\nobs$ and based on sample size (a)~$\nobs$ or (b)~ESS.}
    \label{tab:chi_squared_cov_probs_3}
\end{table}
We observe that the test statistic values are overestimated compared to the expected null distribution, and the coverage probabilities are lower than the nominal value.
This is because the observations in this example exhibit more variation than a multinomial sample with the same size.
The additional variation also does not depend on the sample size, which means that the error between observed and expected test statistic distribution increases when the sample size increases and less variation is expected.
However we can use the simulated data to estimate ESS that compensates for the overdispersion.
When the proposed test statistic is calculated based on ESS, the observed distribution follows the expected distribution well (Figure \ref{fig:chi_squared_visual_3}~(b)) and the coverage probabilities are close to the nominal value $1-\alpha$ in all test conditions (Table \ref{tab:chi_squared_cov_probs_3}~(b)).
The average ESS was 236 when $\nobs=250$ and 806 when $\nobs=1000$ (Figure~\ref{fig:ess_ave}).
\begin{figure}
    \centering
    \includegraphics[width=0.35\textwidth]{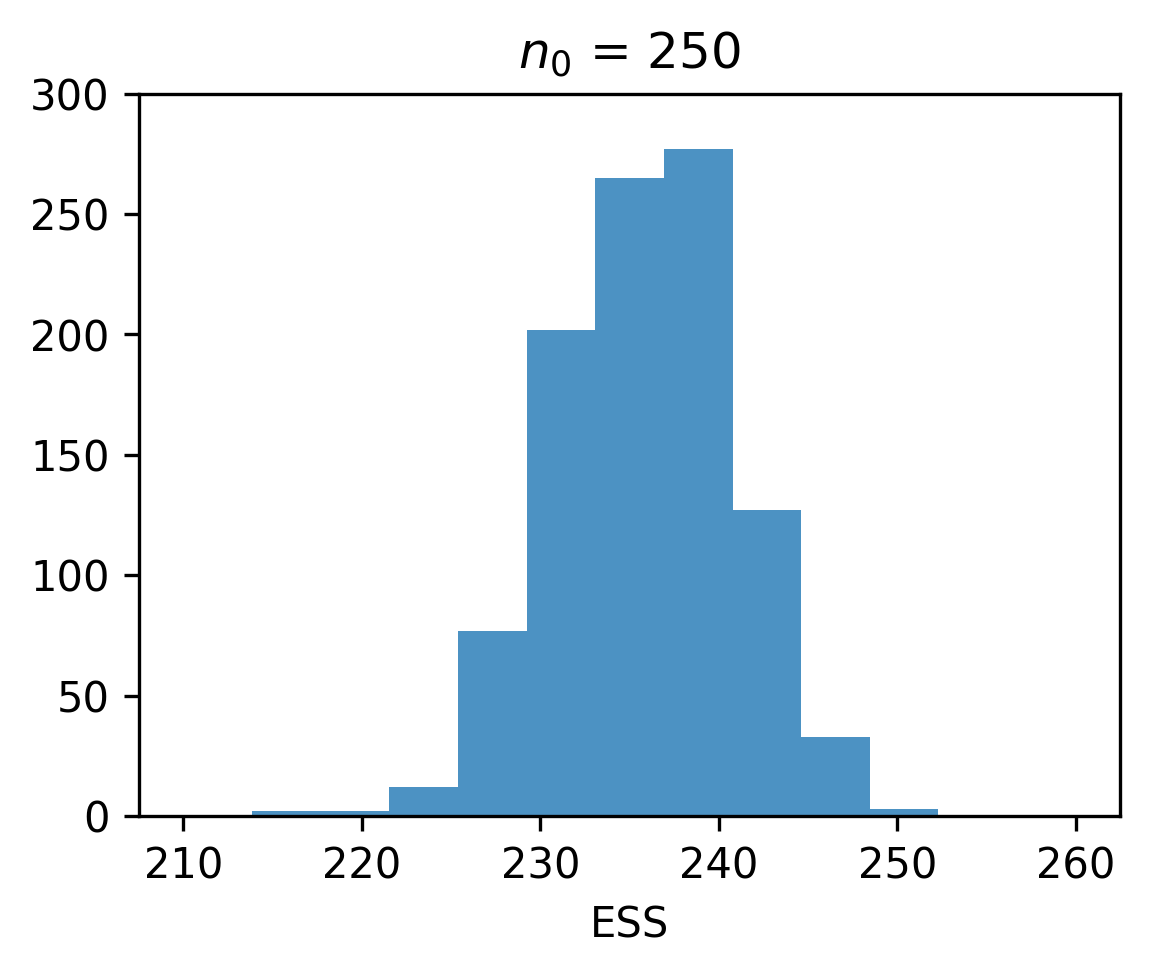}
    \includegraphics[width=0.35\textwidth]{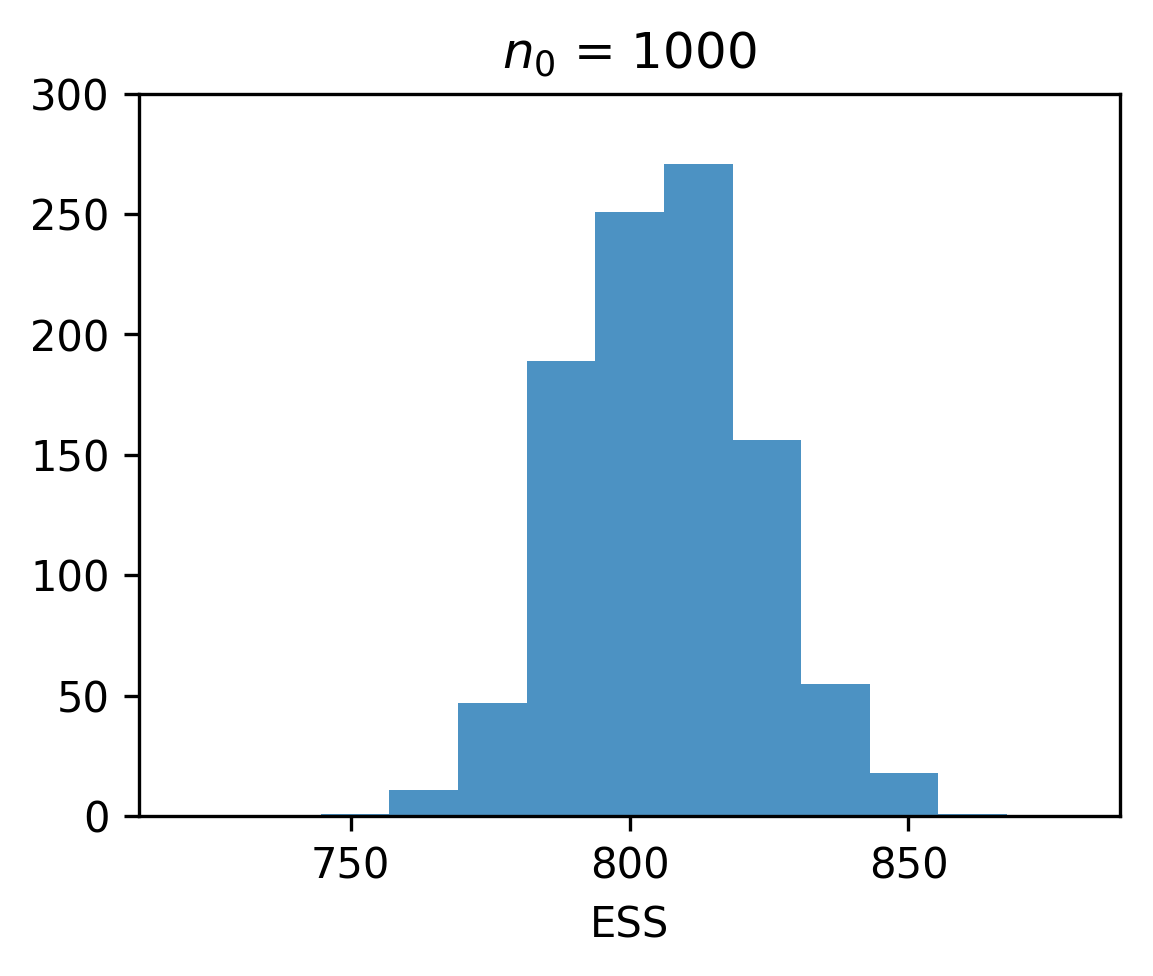}
    \caption{Effective sample size calculated based on the $\nrep=1000$ simulations carried out with the true simulator parameters.}
    \label{fig:ess_ave}
\end{figure}

% 2) BO

We also evaluate the proposed test statistic calculated based on the surrogate model in BO.
The parameter ranges considered in optimization were $\ln(m)\in [-7, -1.6]$, $\ln(v) \in [-7, -0.7]$, and $\ln(\sigma_f)\in [-7, -1.6]$.
We initialized optimization with 50 simulations and set the total simulation count to 2000.
The additional variation observed in this example is expected to have some dependence on the model parameters, but since the parameters that could be included in a confidence set are close to the parameters $\hat\param$ that minimize the expected simulator-based JSD, we approximate ESS at all parameter values with an estimate calculated based on $\nrep=1000$ simulations carried out with $\hat\param$.
The average ESS was 235 when $\nobs=250$ and 800 when $\nobs=1000$ (Figure~\ref{fig:ess_BOLFI}).
\begin{figure}
    \centering
    \includegraphics[width=0.35\textwidth]{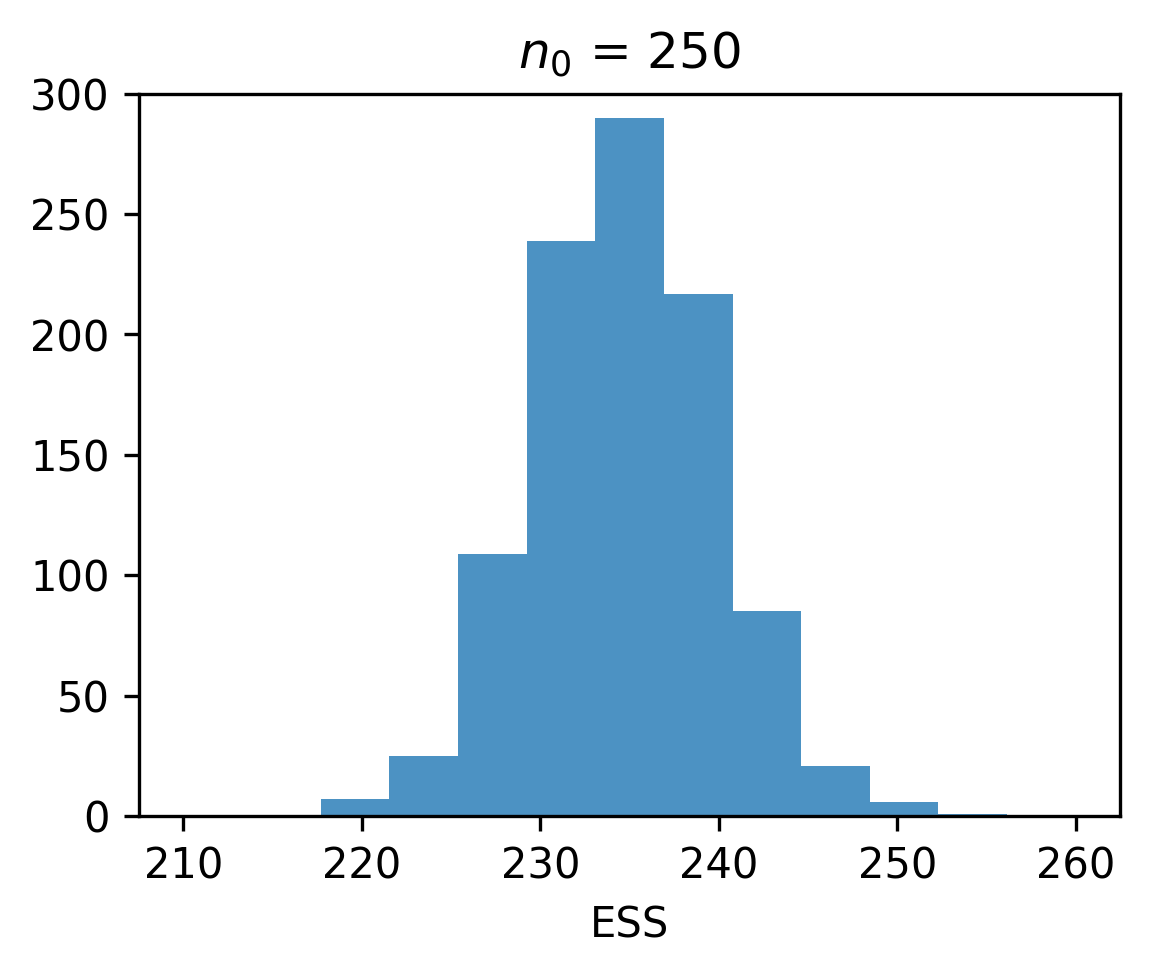}
    \includegraphics[width=0.35\textwidth]{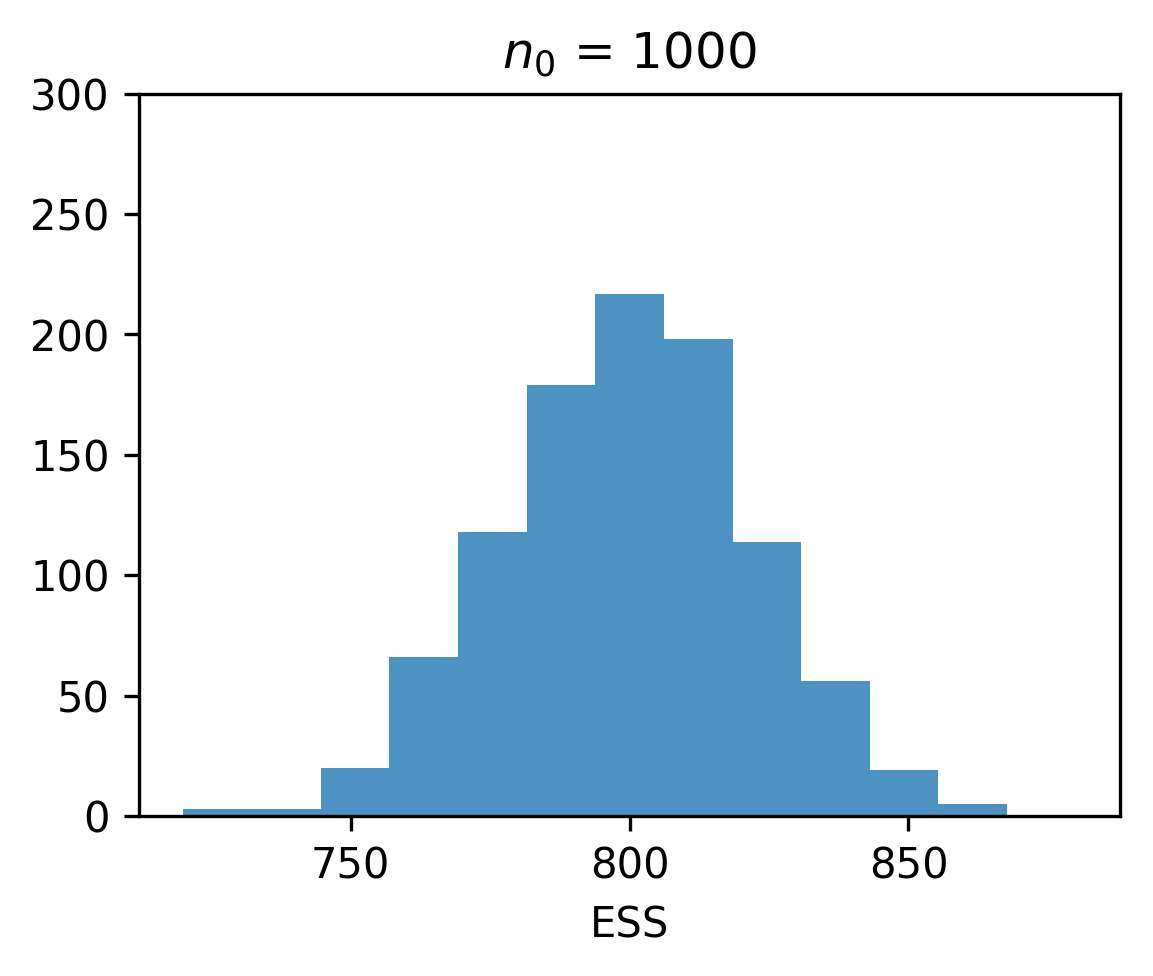}
    \caption{Effective sample size calculated based on $\nrep=1000$ simulations carried out with the parameters that minimize the expected JSD between observed and simulated data.}
    \label{fig:ess_BOLFI}
\end{figure}
Coverage probabilities calculated based on the proposed test statistic are reported in Table~\ref{tab:chi_squared_cov_probs_3_BOLFI}.
\begin{table}
    \centering
    \begin{tabular}{cccc}
    (a) &
    \begin{tabular}{ccc}
         $1-\alpha$ &  $\nobs=250$ & $\nobs=1000$ \\\hline
         0.99 & 0.96 & 0.73\\
         0.95 & 0.86 & 0.58\\
         0.90 & 0.75 & 0.48\\
         0.50 & 0.28 & 0.16\\
    \end{tabular} &
    (b) &
    \begin{tabular}{ccc}
         $1-\alpha$ &  $\nobs=250$ & $\nobs=1000$ \\\hline
         0.99 & 0.98 & 0.86\\
         0.95 & 0.89 & 0.75\\
         0.90 & 0.82 & 0.67\\
         0.50 & 0.35 & 0.32\\
    \end{tabular}
    \end{tabular}
    \caption{Coverage probabilities calculated based on observations simulated with the NFDS model. The proposed test statistic values were calculated based on the BOLFI model using simulations with $\nsms=\nobs$ and based on sample size (a)~$\nobs$ or (b)~ESS.}
    \label{tab:chi_squared_cov_probs_3_BOLFI}
\end{table}
We observe that the coverage probabilities are lower than the corresponding coverage probabilities reported in Table~\ref{tab:chi_squared_cov_probs_3}.
This is due to overestimation errors introduced in BOLFI (Figure~\ref{fig:BOLFI_error_3}).
\begin{figure}
    \centering
    \includegraphics[width=0.45\textwidth]{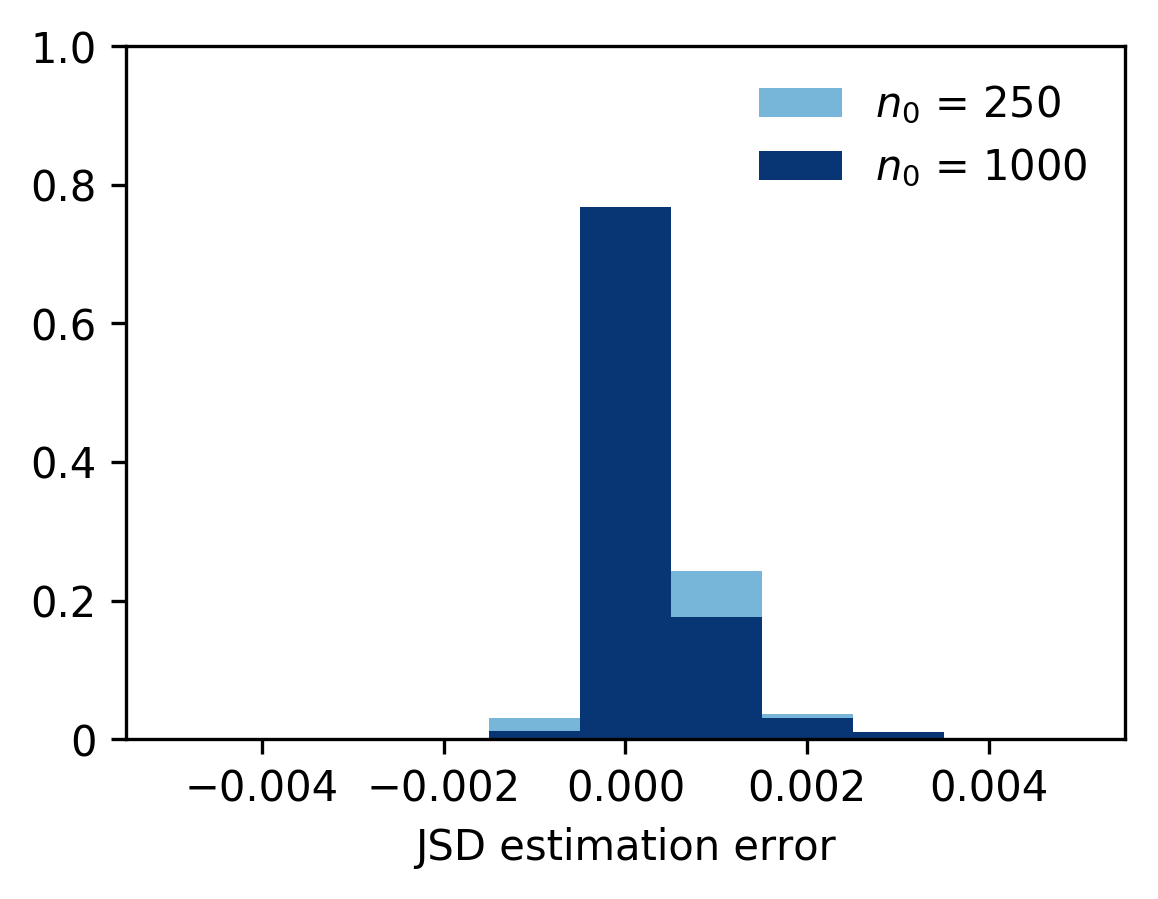}
    \caption{Estimation errors calculated as difference between the average over simulator-based JSD evaluated at $t=36$ and $t=72$ calculated based on the BOLFI model or $\nrep=1000$ simulations.}
    \label{fig:BOLFI_error_3}
\end{figure}

\section{Discussion and conclusions}

% overview and theoretical contributions TODO

Despite some early interest in the frequentist likelihood-free inference, more recently the primary focus of the research has been on approximating the posterior distribution of parameters with various sampling and surrogate model based approaches. As these methods can be very computation intensive and/or do require considerable expertise about training surrogate models, there is a clear need for further research on alternative approaches that would exploit asymptotic behavior of the approximate implicit likelihood. In the present work we considered this objective for the situation where the Jensen--Shannon divergence between observed and simulated data is used to measure the model fit and find model parameters that best explain the observed data.
We derived explicit expressions for the expectation and variance of the simulator-based JSD statistic (Section \ref{bernvantevarde}) and proposed a test statistic that can be used in hypothesis testing or confidence set estimation (Section \ref{secttjitvaa}--\ref{testinv}).

% about simulation results

We carried out simulation experiments to evaluate the proposed test statistic and to examine whether test statistic values can be calculated based on the surrogate model in BOLFI.
The coverage probabilities reported in Sections \ref{sec:experiment1} and \ref{sec:experiment2} indicate that the test statistic values follow the expected asymptotic distribution when the sample size is large and the observed and simulated data follow a multinomial distribution.
The experiments thus confirmed that the proposed approach can be used to evaluate parameter fit and calculate confidence sets in such conditions.
Meanwhile experiments carried out with the NFDS model (Section~\ref{sec:experiment3}) showed that the test statistic can be sensitive to the multinomial distribution assumption.
While this is a concern when we want to use the test statistic with simulator-based models, the experiments also showed that substituting sample size with ESS can compensate for deviations caused by extra heterogeneity in data.
In addition the example demonstrated collapsing observation categories to ensure adequate observation count in all categories.

% BOLFI experiments

%The test statistic studied in this work measures the fit between observed data and a hypothesized model represented as parameters $\param$.
%
%This is in contrast to test statistics that measure relative fit compared to parameters $\hat\param$ that best explain the observed data.
%
%In practice we do not need to solve $\hat\param$ and can calculate the proposed test statistic based on simulations carried out with the selected $\param$.
%
%However we also carried out experiments using BOLFI to find $\hat\param$.
%
%The motivation was that the surrogate model fitted in the optimization process could be used to calculate the proposed test statistic values for large candidate sets, for example when we want to estimate a confidence set over parameter values.

The BOLFI experiments indicated that the proposed test statistic is sensitive to errors in the estimated JSD value.
Experiments carried out with the multinomial example models suggested that the estimation error is related to variation between simulated data sets.
Since the variation in this case depended on the simulated sample size, using simulated samples that are larger than the observed sample allowed accurate test statistic calculation in these examples.
However increasing the simulated sample size is not expected work when the variation has other sources.
Hence rather than assume that the variation can be removed, we need to improve how variation between simulator-based JSD values is modelled in BOLFI.

While the current work is solely focused on the properties of JSD-based inference for models where the output can be summarized in terms of categorical distributions, there are multiple possible venues of future research to expand applicability of the theory. For example, quantization of continuous distributions can be used to enable application of JSD to more general model classes. An interesting question is then how the loss of information will depend on the chosen quantization and how accurate inference statements could be made about the model parameters as the sample size grows. In summary, we anticipate that a rich theory and applications can be found for more analytically oriented likelihood-free inference, where approximate likelihood quantities can be combined with efficient optimization algorithms.

% conclusions and future work

\section*{Acknowledgments}

The authors wish to acknowledge CSC – IT Center for Science, Finland, for computational resources.
 J.C. and U.R. are supported by  ERC grant 742158 and T.K.  is supported by  FCAI (=Finnish Center for Artificial Intelligence).

%\appendix

%\section*{Appendix }
%\label{app:theorem}

% Note: in this sample, the section number is hard-coded in. Following
% proper LaTeX conventions, it should properly be coded as a reference:

%In this appendix we prove the following theorem from
%Section~\ref{sec:textree-generalization}:

%In this appendix we prove the following theorem from

%\noindent
%{\bf Theorem} {\it  haa } \hfill\BlackBox

\appendix

\numberwithin{equation}{section}
\numberwithin{theorem}{section}
\numberwithin{example}{section}

\section{Asymptotics}\label{nasymptotiksec}

 \subsection{An Auxiliary Bound for   $D_{\rm JS}( P, Q)$}
The following bound is included  to make  the argument below more self-contained.  
\begin{lemma}\label{lemmakontbd2}
$P \in \mathbb{P}$, $Q \in \mathbb{P}$. Assume  $V(P,Q) < 1/2$.  Then 
\begin{equation}\label{cskorner2}
D_{\rm JS}(P,Q) \leq  -\pi  (1-\pi) \ln \left(  \pi (1-\pi) \right)    V(P,Q) -  2\pi  (1-\pi)  V(P,Q) \ln \left( \frac{ V(P,Q)}{k}\right). 
\end{equation}
\end{lemma}\begin{proof}: 
We have from Equation~(\ref{klinformationiudnet}) 
$$
 D_{\rm JS}( P, Q) =\pi \left(H(M) -  H(P) \right)  +(1-\pi)\left(H(M)- H(Q)\right)
$$
 and  then, since  $  D_{\rm JS}( P, Q) = \mid D_{\rm JS}( P, Q) \mid$,  
$$
 D_{\rm JS}( P, Q) \leq   \pi \left| H(M) -  H(P) \right|   +(1-\pi) \left|H(M)- H(Q)\right|.
$$
By definition 
$$
 V(M,P) =  \sum_{i=1}^{k} | (\pi p_{i} + (1-\pi)q_{i} ) - p_{i}|$$
$$
 =  \sum_{i=1}^{k} |  (1-\pi)q_{i}  -  (1-\pi) p_{i}| = (1-\pi) V(Q,P)=(1-\pi) V(P,Q),
$$
and  similarly $ V(M,Q) = \pi V(P,Q)$.  For $\pi \in [0,1]$ 
$(1-\pi) V(P,Q) \leq  V(P,Q)$  and $(1-\pi) V(P,Q) \leq  V(P,Q)$,    we get by assumption that  $V(M,P) < \frac{1}{2}$ and $V(M,Q) < \frac{1}{2}$.    Hence by Lemma \ref{lemmakontbd} 
$$
 \pi \left| H(M) -  H(P) \right|  \leq  - \pi (1-\pi)   V(P,Q) \cdot  \ln\left( \frac{(1-\pi) V(P,Q)}{k}\right)
$$
and 
$$
 (1-\pi) \left| H(M) -  H(Q) \right|  \leq  -\pi  (1-\pi)   V(P,Q) \cdot  \ln\left( \frac{ \pi V(P,Q)}{k}\right). 
$$ 
This gives 
$$
 \pi \left| H(M) -  H(P) \right|   +(1-\pi) \left|H(M)- H(Q)\right| \leq  
-\pi  (1-\pi)   V(P,Q) \ln \left(  \frac{\pi (1-\pi) V(P,Q)^{2}}{k^2}\right).  
$$
Hence 
the bound in  Equation~(\ref{cskorner}) follows, as claimed.
\end{proof}
For  $\pi= \frac{1}{2}$  the upper  bound   in  Equation~(\ref{cskorner})  assumes its smallest value and  the bound becomes 
\begin{equation}\label{cskorner3}
D_{\rm JS}(P,Q) \leq   \frac{1 }{2} V(P,Q) \left[\ln 2 -  \ln\left( \frac{V(P,Q)}{k}\right)\right].
\end{equation}
 As  $k  \geq 2$ and  $V(P,Q) \leq 2$, we have   $ 1 \geq V(P,Q)/k $.  Hence the bound in  
 Equation~(\ref{cskorner2}) is less sharp than   the bound   $D_{\rm JS}(P,Q) \leq   \frac{\ln 2 }{2} V(P,Q) $ in  \cite[Equation~(20), p.885]{corremkos} valid for any $V$  and $\pi \in (0,1)$, but suffices  for  the present purposes.

 \subsection{ $V\left(\widehat{Q}_{\theta}, P_{\theta}\right) \rightarrow 0$}\label{nasymptotik}

\begin{lemma}\label{totvarlemma}
 $\widehat{Q_{\theta}} \in \mathbb{M}_{n}(\theta)$. 
 \begin{equation}\label{totvarber1}
V\left(\widehat{Q}_{\theta}, P_{\theta}\right)   \rightarrow 0, \mbox{as $n \rightarrow +\infty$}   
 \end{equation}
 $P_{\theta}$-a.s..
\end{lemma}
\begin{proof}  By definition,  $V\left(\widehat{Q}_{\theta}, P_{\theta}\right) = \sum_{i=1}^{k} \left|  \frac{\xi_{i}}{n} - p_{i}(\theta) \right|$.  
We bound  for any $\epsilon >0$
$$
P_{\theta} \left( \sum_{i=1}^{k} \left|  \frac{\xi_{i}}{n} - p_{i}(\theta) \right|\geq  \epsilon \right) =
P_{\theta} \left( \sum_{i=1}^{k} \left|  \xi_{i} -n  p_{i}(\theta) \right|\geq  n \epsilon \right).
$$

If   the event $  n \epsilon\leq \sum_{i=1}^{k} \left| \xi_{i} -n p_{i}(\theta) \right|$ occurs, there is  
at least one $|\xi_{i} -n p_{i}(\theta)|  \geq  \frac{n\epsilon}{k}$. Thus  we bound upwards by the union bound
$$
\leq P_{\theta} \left( \cup_{i=1}^{k}\left\{ \left| \xi_{i}- n p_{i}(\theta) \right|\geq \frac{n\epsilon}{k}\right\} \right) 
\leq   \sum_{i=1}^{k} P_{\theta} \left( \left|  \xi_{i} - n p_{i}(\theta) \right|\geq  \frac{n\epsilon}{k}\right).
$$
Here $  \xi_{i} \sim {\rm Bin}(n, p_{i}(\theta))$, $  \xi_{i}$  is a sum of $n$ i.i.d. Bernoulli r.v.s  and the mean of the sum is  $ n p_{i}(\theta)$.  The Hoeffding inequality \citep[Thm 8.1]{devroye2013probabilistic} gives 
$$
 P_{\theta} \left( \left|   \xi_{i}   - np_{i}(\theta) \right|\geq  \frac{n\epsilon}{k} \right) \leq  2
e^{-2 n\epsilon^{2}/k^{3}}.
$$
Hence  
$$
 \sum_{i=1}^{k} P_{\theta} \left( \left|\xi_{i} - np_{i}(\theta) \right|\geq  \frac{n \epsilon}{k} \right) \leq   2k e^{-2 n\epsilon^{2}/k^{3}}
$$
and by the preceding chain of inequalities 
$$
\sum_{n=1}^{+\infty}P_{\theta} \left( \sum_{i=1}^{k} \left| \frac{\xi_{i}}{n} - p_{i}(\theta) \right|\geq   \epsilon\right)  < +\infty,
$$
and  Equation~(\ref{totvarber1}) follows by the  Borel-Cantelli lemma. \end{proof}
The proof  above  implements the  hint of proof for the case at hand   in  \citet{berend2012convergence}. 
The   proof  shows in fact complete convergence, see e.g., \citet[p.203]{gut2013probability}.
 In addition,   $ {\mathbf D}= (D_{1}, \ldots, D_{n_{o}}) $,  an i.i.d. $n_{o}$-sample  $\sim$ $P_{o}$ with $ \widehat{P}_{\mathbf D} $ that   
  \begin{equation}\label{totvarber2}
V\left(\widehat{P}_{\mathbf{D}}, P_{o}\right) \rightarrow 0, \mbox{as $n_{o} \rightarrow +\infty$}   
 \end{equation}
 $P_{o}$-a.s..
We get the following statement.  
\begin{theorem}
 $\widehat{Q_{\theta}} \in \mathbb{M}_{n}(\theta)$. 
 \begin{equation}\label{totvarber11}
D_{\rm JS}\left(\widehat{Q}_{\theta}, P_{\theta}\right)  \rightarrow 0, \mbox{as $n \rightarrow +\infty$}   
 \end{equation}
 $P_{\theta}$-a.s..
\end{theorem}
\begin{proof}: By  Equation~(\ref{totvarber2})  there exists $n_{1/2}$ such that $V\left(\widehat{Q}_{\theta}, P_{\theta}\right)  < 1/2$ 
for all $n > n_{1/2}$ with  probability one. Hence Lemma  \ref{lemmakontbd2}, i.e., Equation~(\ref{cskorner2}), entails  that 
$D_{\rm JS}\left(\widehat{Q}_{\theta}, P_{\theta}\right)  \rightarrow 0$, as $n \rightarrow +\infty$, since  $ V\ln V \rightarrow 0$, when $V\rightarrow  0$. 
\end{proof} 

 \subsection{Further  Asymptotics of  KLD  and Variation Distance}

\begin{lemma}\label{dklkonv}
\begin{description}
\item[(i)] Assume ${\rm supp}\left(P_{\theta} \right)= {\cal A}$ and  $\widehat{Q} _{\theta} \in \mathbb{M}_{n}(\theta)$. Then, as  $n \rightarrow +\infty$   
\begin{equation} \label{dklkonv1}
D_{\rm KL} \left( \widehat{Q}_{\theta}, P_{\theta} \right) \rightarrow 0, 
\end{equation}
$P_{\theta}$- a.s.. 
\item[(ii)] Assume ${\rm supp}\left(P_{\theta_{o}} \right)= {\cal A}$.  $ {\mathbf D}= (D_{1}, \ldots, D_{n_{o}}) $,  is  an i.i.d. $n_{o}$-sample  $\sim$ $P_{\theta_{o}}$. Then, as  $n_{o} \rightarrow +\infty$   
\begin{equation} \label{dklkonv2}
D_{\rm KL} \left( \widehat{P}_{\mathbf{D}}, P_{\theta_{o} }\right) \rightarrow 0, 
\end{equation}
$P_{o}$-a.s.. 
\end{description}
\end{lemma}
\begin{proof}  The assertions $\bf(i)$ and $\bf(ii)$  follow respectively from Equations~(\ref{totvarber1})  and (\ref{totvarber2}) by the reverse Pinsker inequality for $P \in \mathbb{P}$ and  $Q \in \mathbb{P}$, and $Q_{min}= \min_{1\leq  j \leq k}q_{j}$,
\begin{equation}\label{sasonpinsker}
D_{\rm KL} (P,Q) \leq \frac{1}{Q_{min}}V( P, Q)^{2}. 
\end{equation}
In  $\bf(i)$   take   $P= \widehat{Q}_{\theta}$ and  $ Q=P_{\theta} $ and analogously for  $\bf(ii)$. Note that 
$V(P,Q)=V(Q,P)$. 
 \end{proof}
For the reverse Pinsker inequality, see \citet{sason2015reverse},  and  
\citet{sason2016f}.  The results in the corollary are known, see \citet[Thm 11.2.1]{cover2012elements}, but are here
justified by  
 the reverse Pinsker inequality. 
 
\subsection{Asymptotics for  Increasing   $ {\mathbf D}$}\label{stokka}

The convergence in Equation~(\ref{dklkonv1}) of  Lemma \ref{dklkonv} requires that  there is a $\theta_{o}$  in the parameter space of the model such that observed data $\mathbf{D}$  are   $\sim  P_{\theta_{o}}$. Next we find a convergence that  does not require  this assumption, but reduces to
 Equation~(\ref{dklkonv1}),  if $P_{o}  \in \mathbb{M}$.
\begin{lemma}\label{misspec}
 Assume that Equation~(\ref{posass}) holds for
$P_{o} \in  \mathbb{P}$ and for  any $P_{\theta} \in \mathbb{M}$.  Let  $ {\mathbf D}= (D_{1}, \ldots, D_{n_{o}}) $ be    an i.i.d. $n_{o}$-sample  $\sim$ $P_{o}$. Then  it holds that 
\begin{equation}\label{misspec13}
\lim_{n_{o} \rightarrow +\infty} D_{\rm KL} \left( \widehat{P}_{\mathbf D}, P_{\theta}\right) =  D_{\rm KL} \left( P_{o}, P_{\theta}\right). 
\end{equation}
$P_{o}$-a.s.. 
\end{lemma}
\begin{proof}
From  \citet[Thm 11.1.2]{cover2012elements}  we have  the identity 
\begin{equation}\label{fshkldid}
\ln P_{\theta} (\mathbf{ D}) = -n_{o}  H\left(\widehat{P}_{\mathbf{ D}}\right) - n_{o}D_{\rm KL} \left( \widehat{P}_{\mathbf {D}}, P_{\theta}\right).
\end{equation} 
 By the same token  the probability of   $\mathbf{D}$ under $P_{o}$ is
\begin{equation} \label{take2}
\ln P_{o} (\mathbf {D}) = -n_{o}  H\left(\widehat{P}_{\mathbf {D}}\right) - n_{o}D_{\rm KL} \left( \widehat{P}_{\mathbf {D}}, P_{o}\right).
\end{equation}
If  the entropy expression from Equation~(\ref{fshkldid}) is substituted in Equation~(\ref{take2}) we get
$$
\ln P_{o} (\mathbf {D}) = \ln P_{\theta} (\mathbf{D}) +  n_{o}\left[D_{\rm KL} \left( \widehat{P}_{\mathbf {D}}, P_{\theta}\right)  -D_{\rm KL} \left( \widehat{P}_{\mathbf {D}}, P_{o}\right)\right]
$$
and then   
\begin{equation}\label{finident}
\frac{1}{n_{o}}\ln\frac{ P_{o} ({\mathbf D}) }{ P_{\theta} ({\mathbf D})} = D_{\rm KL} \left( \widehat{P}_{\mathbf D}, P_{\theta}\right) -D_{\rm KL} \left( \widehat{P}_{\mathbf D}, P_{o}\right).
\end{equation}
 In the left hand side  we have by the strong law of large numbers, that as $n_{o}\rightarrow +\infty$
$$
\frac{1}{n_{o}}\ln\frac{ P_{o} ({\mathbf D}) }{ P_{\theta} ({\mathbf D})} =\frac{1}{n_{o}} \sum_{i=1}^{n} \ln\frac{ P_{o} (D_{i}) }{ P_{\theta} (D_{i})} \rightarrow D_{\rm KL} \left(P_{o}, P_{\theta} \right)
$$
$P_{o}$-a.s.. By  Equation~(\ref{dklkonv2})
 $D_{\rm KL} \left( \widehat{P}_{\mathbf D}, P_{o}\right) \rightarrow 0$, $P_{o}$-a.s., as $n\rightarrow +\infty$. When we allow  these  limits in 
Equation~(\ref{finident}), the assertion in  Equation~(\ref{misspec13}) is established. \end{proof}

Now we give the proof of Proposition  \ref{misspec11}.
\begin{proof}
We use the calculation in  Equation~(\ref{plan}) to obtain 
\begin{equation}\label{plan2}
D_{\rm JS} \left(  \widehat{P}_{\mathbf{D}},  P_{\theta} \right) = \pi D_{\rm KL} \left(  \widehat{P}_{\mathbf{D}},  P_{\theta} \right) - D_{\rm KL} \left(  \widehat{M}_{\rm sp}, P_{\theta} \right),
 \end{equation}
 where  $  \widehat{M}_{\rm sp}:= \pi   \widehat{P}_{\mathbf{D}}  + (1- \pi) P_{\theta}$.  We prove first that 
$$
D_{\rm KL} \left(  \widehat{M}_{\rm sp}, P_{\theta} \right) \rightarrow D_{\rm KL} \left( M_{o}, P_{\theta} \right), 
$$
where $  M_{o}:= \pi   P_{o}  + (1- \pi) P_{\theta}$.  We   bound 
$
\left |D_{\rm KL} \left(  \widehat{M}_{\rm sp}, P_{\theta} \right) -D_{\rm KL} \left(  M_{o}, P_{\theta} \right)\right|
$
as follows. 
We have by definition of $D_{\rm KL}$  
\begin{eqnarray}\label{melcomp223} 
\left |  D_{\rm KL} \left( \widehat{M}_{\rm sp},  P_{\theta} \right)  -   D_{\rm KL}\left(M_{o},  P_{\theta}\right)\right | &\leq &          \left |   H\left( \widehat{M}_{\rm sp}\right)  - H\left( M_{o} \right)   \right | \nonumber \\ 
&  & \\
& +& \left| \sum_{j=1}^{k} m_{j,{\rm sp}}  \ln \frac{1}{p_{j}\left(\theta\right)} - \sum_{j=1}^{k}m_{j,o}  \ln \frac{1}{p_{j}\left(\theta\right)} \right|,  \nonumber
\end{eqnarray} 
where we used Equation~(\ref{posass}).   But from this inequality on  the proof of  Theorem  \ref{nskonv} shows that the  convergence sought for here depends on 
the variation distance between $ \widehat{M}_{\rm sp}$ and $  M_{o}$. We have 
 $$
V\left(  \widehat{M}_{\rm sp} ,  M_{o} \right)  = \pi  V \left(\widehat{P}_{\mathbf{D}}  , P_{o} \right).
$$
 By the remark in Equation~(\ref{totvarber2})  we see that  $V\left(\widehat{P}_{\mathbf{D}}, P_{o}\right) \rightarrow 0$,  as $n_{o} \rightarrow +\infty$.  
 Hence   it follows as in    the proof of  Theorem  \ref{nskonv}  that 
\begin{equation}\label{toka1}
\left |D_{\rm KL} \left(  \widehat{M}_{\rm sp}, P_{\theta} \right) -D_{\rm KL} \left(   M_{o}, P_{\theta} \right)\right| \rightarrow 0,
\end{equation}
$P_{o}$-a.s.,  as $n_{o} \rightarrow +\infty$. Now Equation~(\ref{misspec13}) in Lemma  \ref{misspec} and Equation~(\ref{toka1}) imply via 
 Equation~(\ref{plan2}) that
\begin{equation}\label{toka2}
D_{\rm JS} \left(  \widehat{P}_{\mathbf{D}},  P_{\theta} \right) \rightarrow   \pi D_{\rm KL} \left( P_{o},  P_{\theta} \right) - D_{\rm KL} \left( M_{o}, P_{\theta} \right)
\end{equation}
$P_{o}$-a.s.,  as $n_{o} \rightarrow +\infty$.  But the calculus underlying Equation~(\ref{plan2})  verifies that  the right  hand side of 
Equation~(\ref{toka2}) equals  $D_{\rm JS} \left(  P_{o},  P_{\theta} \right)$, as claimed. \end{proof}
The next result is written in terms of the symmetric JSD with $\pi=1/2$, i.e., $ D_{\rm JS, 1/2}$, but the results hold obviously for any $\pi \in (0,1)$. 
\begin{proposition}\label{misspec12}
 Assume that Equation~(\ref{posass}) holds for
$P_{o} \in  \mathbb{P}$ and for  any $P_{\theta} \in \mathbb{M}_{p}$. Let  $ {\mathbf D}= (D_{1}, \ldots, D_{n_{o}}) $ be    an i.i.d. $n_{o}$-sample  $\sim$ $P_{o}$. Then  it holds that 
\begin{equation}\label{misspec1}
\lim_{n_{o} \rightarrow +\infty} D_{\rm JS, 1/2} \left( \widehat{P}_{\mathbf D}, P_{\theta}\right) =  D_{\rm JS, 1/2} \left( P_{o}, P_{\theta}\right). 
\end{equation}
$P_{o}$-a.s.. 
\end{proposition}
\begin{proof}
We use the calculation in  Equation~(\ref{plan}) to obtain 
\begin{equation}\label{plan3}
D_{\rm JS, 1/2} \left(  \widehat{P}_{\mathbf{D}},  P_{\theta} \right) =  \frac{1}{2}   D_{\rm KL} \left(  \widehat{P}_{\mathbf{D}},  P_{\theta} \right) - D_{\rm KL} \left(  \widehat{M}_{\rm sp}, P_{\theta} \right),
 \end{equation}
 where  $  \widehat{M}_{\rm sp}:= \frac{1}{2}  \left[ \widehat{P}_{\mathbf{D}}  + P_{\theta} \right]$. In view of  Equation~(\ref{misspec13}) in Lemma  \ref{misspec} we need only to prove that  
$$
D_{\rm KL} \left(  \widehat{M}_{\rm sp}, P_{\theta} \right) \rightarrow D_{\rm KL} \left( M_{o}, P_{\theta} \right), 
$$
where $  M_{o}:=  \frac{1}{2}  \left[   P_{o}  + P_{\theta}\right]$.  We use the bound in Equation~(\ref{melcomp223}). 
  We have $
V\left(  \widehat{M}_{\rm sp} ,  M_{o} \right)  = \frac{1}{2}  V \left(\widehat{P}_{\mathbf{D}}  , P_{o} \right)$.   Hence  $ V\left(  \widehat{M}_{\rm sp} ,  M_{o} \right)   <1/2$    and  by Lemma \ref{lemmakontbd}, Equation~(\ref{cskorner})  entails  
\begin{equation}\label{bounda}
\left |   H\left( \widehat{M}_{\rm sp}\right)  - H\left( M_{o} \right)   \right | \leq  -\frac{1}{2}
   V \left(\widehat{P}_{\mathbf{D}}  , P_{o} \right) \ln\left( \frac{  \frac{1}{2} V \left(\widehat{P}_{\mathbf{D}}  , P_{o} \right) }{k}\right).
\end{equation}
For the second term in the right hand side of Equation~(\ref{melcomp223}) we have   by definitions of  $M_{\rm sp }$ and  $M_{o}$
$$
  \left| \sum_{j=1}^{k} m_{j,{\rm sp}}  \ln \frac{1}{p_{j}\left(\theta\right)} - \sum_{j=1}^{k}m_{j,o}  \ln \frac{1}{p_{j}\left(\theta\right)} \right| = 
\frac{1}{2}\left | \sum_{j=1}^{k}  \left(\widehat{p}_{j} - p_{o,j} \right)\ln \frac{ 1 }{p_{j}(\theta)}  \right|.
$$
The right hand side is bounded upwards ($|\ln  1/p_{j}\left(\theta\right)|=\ln  1/p_{j}\left(\theta\right) $) by 
\begin{equation}\label{boundb}
\leq   \frac{1}{2}\sum_{j=1}^{k} \left |  \left(\widehat{p}_{j} - p_{o,j} \right) \right| \ln \frac{ 1 }{p_{j}\left(\theta\right)} \leq   \frac{1}{ 2} \min_{ 1 \leq j \leq k} \ln \frac{ 1 }{p_{j}\left(\theta\right)} V \left(\widehat{P}_{\mathbf{D}}  , P_{o} \right).
\end{equation}
By the remark in Equation~(\ref{totvarber2})  we have  that  $V\left(\widehat{P}_{\mathbf{D}}, P_{o}\right) \rightarrow 0$,  as $n_{o} \rightarrow +\infty$.  
Hence Equations~(\ref{melcomp223}), (\ref{bounda}) and (\ref{boundb}) give 
 \begin{equation}\label{tokab}
\left |D_{\rm KL} \left(  \widehat{M}_{\rm sp}, P_{\theta} \right) -D_{\rm KL} \left(   M_{o}, P_{\theta} \right)\right| \rightarrow 0,
\end{equation}
$P_{o}$-a.s.,  as $n_{o} \rightarrow +\infty$. Now Equation~(\ref{misspec13}) in Lemma  \ref{misspec} and Equation~(\ref{toka1}) imply via 
 Equation~(\ref{plan2}) that
\begin{equation}\label{tokac}
D_{\rm JS, 1/2} \left(  \widehat{P}_{\mathbf{D}},  P_{\theta} \right) \rightarrow   \frac{1}{2} D_{\rm KL} \left( P_{o},  P_{\theta} \right) - D_{\rm KL} \left( M_{o}, P_{\theta} \right)
\end{equation}
$P_{o}$-a.s.,  as $n_{o} \rightarrow +\infty$.  But the calculus underlying Equation~(\ref{plan2})  verifies that  the right  hand side of 
Equation~(\ref{tokac}) equals  $ D_{\rm JS, 1/2} \left( P_{o}, P_{\theta}\right)$, as claimed. \end{proof}

\section{The Sum of Second Derivatives in the Voronovskaya   Expansion  }\label{voronder}
 We  need to obtain  Equation~(\ref{andraderivator}).  In order to simplify writing we set   $x= p_{i}(\theta)$ and $p=\widehat{p}_{i}$  and replace 
$u_{i}^{(1)}$ and  $u_{i}^{(1)}$ in Equation~(\ref{bernst2})   by
$$
f(x)=  \ln \left(\frac{ p}{ \pi p +(1-\pi)x}\right) = \ln p -  \ln  \left({ \pi p +(1-\pi)x}\right) 
$$
and
$$
g(x) =x\ln\left(\frac{x}{{ \pi p +(1-\pi)x}}\right) = x \ln x + xf(x) - x \ln p,
$$
respectively.  We have  the first and second derivatives  w.r.t $x$ 
\begin{equation}\label{fderiv}
f^{(1)}(x)= -\left(\frac{1-\pi }{ \pi p +(1-\pi)x}\right), f^{(2)}(x)= \frac{(1-\pi)^{2} }{( \pi p +(1-\pi)x)^{2}}
\end{equation}
and the first derivative  and the second derivative 
$$
g^{(1)}(x)= \ln x +1 +  f(x)  +  xf^{(1)}(x) -\ln p, 
g^{(2)}(x)=  \frac{1}{x} + 2 f^{(1)}(x) + xf^{(2)}(x),
$$
respectively. Let  us also put $m=\pi p +(1-\pi)x$.  When we  insert from Equation~(\ref{fderiv})  we get 
$$
g^{(2)}(x)= \frac{1}{x}  -\frac{2(1-\pi) }{ m} + \frac{(1-\pi)^{2} x}{ m^{2}}.
$$
 Then 
 $$
\pi  p f^{(2)}(x) + (1-\pi) g^{(2)}(x)=\frac{\pi (1-\pi)^{2} p}{m^{2}}+ \frac{(1-\pi) }{x}  -\frac{2(1-\pi)^{2} }{ m} + \frac{(1-\pi)^{3} x}{ m^{2}}
$$
$$
= (1-\pi)^{2}\left[ \frac{\pi p -2 m +(1-\pi)x}{m^{2}}     \right] + \frac{(1-\pi) }{x}=(1-\pi) \left[  \frac{1 }{x} -\frac{(1-\pi)}{m}     \right],
$$
which  with  $x= p_{i}(\theta)$ and $p=\widehat{p}_{i}$  produces Equation~(\ref{andraderivator}).  
 
\section{Proof of Lemma \ref{vantevarde}}\label{odotusarvo} 

\begin{proof} By  Equation~(\ref{Qhat2}), $ \triangle( \widehat{Q}_{\theta} )= (\xi_{1}/n, \ldots, \xi_{k}/n)$.   
We set   $\widehat{M}:= \pi\widehat{P}_{\mathbf{D}} +(1-\pi)  \widehat{Q}_{\theta}$, and   $\widehat{m}_{i}=\pi \widehat{p}_{i}  + (1-\pi)\frac{\xi_{i}}{n}$, $i=1, \ldots, k $. By the definition  of JSD  
 \begin{equation}\label{jsdivaver2} 
E_{P_{\theta}}\left[D_{\rm JS}(  \widehat{P}_{\mathbf{D}},    \widehat{Q}_{\theta})\right]=\pi   E_{P_{\theta}}\left[ D_{\rm KL}(\widehat{P}_{\mathbf{D}}, \widehat{M})\right] +(1-\pi)
E_{P_{\theta}}\left[ D_{\rm KL}(   \widehat{Q}_{\theta}, \widehat{M} )\right].
 \end{equation}
Since  $ \widehat{P}_{\mathbf{D}}$ is in this  not a random variable, we have
 \begin{eqnarray}\label{term1} 
 E_{P_{\theta}}\left[ D_{\rm KL}(  \widehat{P}_{\mathbf{D}}, \widehat{M} ) \right] & = & \sum_{i=1}^{k} \widehat{p}_{i}\ln \widehat{p}_{i} \nonumber \\
& &  \\
& -& \sum_{i=1}^{k} \widehat{p}_{i} E_{P_{\theta}}\left[ 
\ln\left(  \pi \widehat{p}_{i}  + (1-\pi) \frac{\xi_{i}}{n} \right)\right], \nonumber 
\end{eqnarray} 
and 
 \begin{eqnarray}\label{term2} 
 E_{P_{\theta}}\left[ D_{\rm KL}(\widehat{Q}_{\theta}, \widehat{M}) \right] & = & \sum_{i=1}^{k} E_{P_{\theta}}\left[  \frac{\xi_{i}}{n}\ln \frac{\xi_{i}}{n}\right]
 \nonumber \\
&- & \sum_{i=1}^{k} E_{P_{\theta}}\left[ \frac{\xi_{i}}{n}\ln\left( \pi\widehat{p}_{i}  +(1-\pi) \frac{\xi_{i}}{n}\right) \right].
\end{eqnarray}
The required expectations of the functions of $\xi_{i}$ above  are  next  computed with  the  binomial distribution ${\rm Bin}\left(n, p_{i}\left( \theta\right)\right)$. When 
$b_{i}(r)$ denotes the binomial probability of 
$r$ successes for a binomial r.v. $\sim {\rm Bin}\left(n, p_{i}\left( \theta\right)\right)$, and $ \widehat{m}_{i}(r) :=  \pi\widehat{p}_{i}  +(1-\pi) \frac{r}{n}$  we have 
$$
 E_{P_{\theta}}\left[ 
\ln\left( \pi \widehat{p}_{i}  + (1-\pi)\widehat{Q}_{i} \right)\right]= \sum_{r=0}^{n} \ln\left(  \widehat{m}_{i}(r) \right)
b_{i}(r).
$$
Thus we get  
 \begin{eqnarray}\label{term12} 
 E_{P_{\theta}}\left[ D_{\rm KL}( \widehat{P}, \widehat{M})\right] & = & \sum_{i=1}^{k} \widehat{p}_{i}\ln \widehat{p}_{i} \nonumber \\
& &  \\
& -& \sum_{i=1}^{k} \widehat{p}_{i}\sum_{r=0}^{n} \ln\left( \widehat{m}_{i}(r) \right)
b_{i}(r). \nonumber 
\end{eqnarray} 
By the same argument as above 
$$
 \sum_{i=1}^{k} E_{P_{\theta}}\left[\frac{\xi_{i}}{n}\ln  \frac{\xi_{i}}{n}\right] = \sum_{i=1}^{k} \sum_{r=1}^{n} \frac{r}{n} \ln\left(\frac{r}{n}  \right)b_{i}(r).
$$
The term corresponding to $r=0$ vanished above due to  $0 \ln (0/n)=0$. Furthermore 
$$
 \sum_{i=1}^{k}E_{P_{\theta}}\left[ \frac{\xi_{i}}{n}\ln\left( (\pi \widehat{p}_{i}  + (1-\pi) \frac{\xi_{i}}{n}\right) \right] = 
\frac{1}{n} \sum_{i=1}^{k}\sum_{r=1}^{n} r  \ln\left( \widehat{m}_{i}(r) \right)b_{i}(r).
$$
Hence we have 
 \begin{eqnarray}\label{term22} 
 E_{P_{\theta}}\left[ D_{\rm KL}( \widehat{Q}, \widehat{M})\right] & = &  \sum_{i=1}^{k} \sum_{r=1}^{n} \frac{r}{n} \ln\left(\frac{r}{n}  \right)b_{i}(r).
 \nonumber \\
&- & \frac{1}{n} \sum_{i=1}^{k}\sum_{r=1}^{n} r  \ln\left( \widehat{m}_{i}(r)\right)b_{i}(r).
\end{eqnarray}
Hence
\begin{eqnarray}\label{proberror23}
E_{P_{\theta}}\left[D_{\rm JS}( \widehat{P}, \widehat{Q})\right]   =  \pi \sum_{i=1}^{k} \widehat{p}_{i}\ln \widehat{p}_{i} &-&  \pi\sum_{i=1}^{k} \widehat{p}_{i}\sum_{r=0}^{n} \ln\left( \widehat{m}_{i}(r) \right)
b_{i}(r) \nonumber \\
& & \\
&+&   (1-\pi)\sum_{i=1}^{k} \sum_{r=1}^{n} \frac{r}{n} \ln\left(\frac{\frac{r}{n}}{ \widehat{m}_{i}(r) }  \right)b_{i}(r).
\nonumber
\end{eqnarray} 
 Consider the second term 
in the right hand side of   in Equation~(\ref{proberror23}), namely 
$$
-\pi\sum_{i=1}^{k} \widehat{p}_{i}\sum_{r=0}^{n} \ln\left( \widehat{m}_{i}(r) \right)b_{i}(r). 
$$
  Consider the second term 
in the right hand side of   in Equation~(\ref{proberror23}), namely 
$$
-\pi\sum_{i=1}^{k} \widehat{p}_{i}\sum_{r=0}^{n} \ln\left( \widehat{m}_{i}(r) \right)b_{i}(r). 
$$
We  have the compensated identity 
$$
= -\pi\sum_{i=1}^{k} \widehat{p}_{i}\sum_{r=0}^{n} \ln \left(\frac{\widehat{m}_{i}(r)}{ p_{i}\left( \theta\right)}\right)  
b_{i}(r) -\pi\sum_{i=1}^{k} \widehat{p}_{i}\sum_{r=0}^{n} \ln p_{i}\left( \theta\right)
b_{i}(r)
$$
and  rightmost term can be rewritten as
$$
\sum_{i=1}^{k} \widehat{p}_{i}\sum_{r=0}^{n} \ln p_{i}\left( \theta\right) 
b_{i}(r)=\sum_{i=1}^{k} \widehat{p}_{i} \ln p_{i}\left( \theta\right)\sum_{r=0}^{n}  b_{i}(r)
= \sum_{i=1}^{k} \widehat{p}_{i} \ln p_{i}\left( \theta\right). 
$$
Thus we have obtained 
$$
-\pi\sum_{i=1}^{k} \widehat{p}_{i}\sum_{r=0}^{n} \ln\left(\widehat{m}_{i}(r)\right)
b(r; p_{i}\left( \theta\right))  = -\pi\sum_{i=1}^{k} \widehat{p}_{i}\sum_{r=0}^{n} \ln \left(\frac{\widehat{m}_{i}(r)}{ p_{i}\left( \theta\right)}\right) b(r; p_{i}\left( \theta\right)) 
$$
$$
- \pi \sum_{i=1}^{k} \widehat{p}_{i} \ln p_{i}\left( \theta\right). 
$$
W.r.t.  Equation~(\ref{proberror23}) we note that 
$
\pi \sum_{i=1}^{k}\widehat{p}_{i} \ln \widehat{p}_{i}- \pi \sum_{i=1}^{k} \widehat{p}_{i} \ln p_{i}\left( \theta\right) = \pi D_{\rm KL}( \widehat{P}, P_{\theta}),
$     
and we have obtained  
\begin{eqnarray}\label{proberror233}
E_{P_{\theta}}\left[D_{\rm JS}( \widehat{P}, \widehat{Q}) \right]   =  \pi D_{\rm KL}( \widehat{P}, P_{\theta}) &-& \pi\sum_{i=1}^{k} \widehat{p}_{i}\sum_{r=0}^{n} \ln \left(\frac{\widehat{m}_{i}(r)}{ p_{i}\left( \theta\right)}\right) b_{i}(r) \nonumber \\
& & \\
&+ &    (1-\pi)\sum_{i=1}^{k} \sum_{r=1}^{n} \frac{r}{n} \ln\left(\frac{\frac{r}{n}}{ \widehat{m}_{i}(r) }  \right) b_{i}(r).
\nonumber
\end{eqnarray} 
We take a closer look at the second term in the right hand side, i.e.,  
$$
S:=\pi\sum_{i=1}^{k} \widehat{p}_{i}\sum_{r=0}^{n} \ln \left(\frac{\widehat{m}_{i}(r)}{ p_{i}\left( \theta\right)}\right) b_{i}(r). 
$$
A compensated equality for this is 
$$
\ln \left(\frac{\widehat{m}_{i}(r)}{ p_{i}\left( \theta\right)}\right)= \ln \left(\frac{\widehat{m}_{i}(r)}{ \widehat{p}_{i}}\right) + \ln \left(\frac{\widehat{p}_{i}}{ p_{i}\left( \theta\right)}\right).
$$
This yields 
$$
S= \pi\sum_{i=1}^{k} \widehat{p}_{i}\sum_{r=0}^{n} \ln \left(\frac{\widehat{m}_{i}(r)}{ \widehat{p}_{i}}\right) b_{i}(r) + \pi \sum_{i=1}^{k} \widehat{p}_{i}\ln \left(\frac{\widehat{p}_{i}}{ p_{i}\left( \theta\right)}\right)\sum_{r=0}^{n} b_{i}(r)
$$
$$
= \pi\sum_{i=1}^{k} \widehat{p}_{i}\sum_{r=0}^{n} \ln \left(\frac{\widehat{m}_{i}(r)}{ \widehat{p}_{i}}\right) b_{i}(r) +\pi D_{\rm KL}(\widehat{P},P_{\theta}). 
$$
As soon as  this  is substituted  in the right hand side of Equation~(\ref{proberror233}), we  have
\begin{eqnarray}\label{slututryckperl}
E_{P_{\theta}}\left[D_{\rm JS}( \widehat{P}_{\mathbf{D}}, \widehat{Q}_{\theta})\right]    & =& \pi\sum_{i=1}^{k} \widehat{p}_{i}\sum_{r=0}^{n} \ln 
\left(\frac{ \widehat{p}_{i}}{\widehat{m}_{i}(r)}\right) b(r;p_{i}\left( \theta\right)) \nonumber \\
& & \\
&+ &    (1-\pi)\sum_{i=1}^{k} \sum_{r=1}^{n} \frac{r}{n} \ln\left(\frac{\frac{r}{n}}{ \widehat{m}_{i}(r) }  \right)b(r;p_{i}\left( \theta\right)).
\nonumber
\end{eqnarray} 
As soon as  the Bernstein operators 
(Equation~\ref{bernst1}) on  the functions in  Equation~(\ref{bernst2}) are  identified  in Equation~(\ref{slututryckperl}), we get Equation~(\ref{slututryck}), as claimed. \end{proof}

\section{Proof of  Lemma    \ref{jsdvantevardetaylor}  }\label{mse} 

\begin{proof} We   start with    Lemma \ref{jsdtaylorlem} 
by  evaluating the expectations in   the three terms in Equation~(\ref{jsdstatmse}). Set $\sigma_{j}^{2}(\theta):= p_{j}(\theta)(1-p_{j}(\theta))$.  By simulator-modeling  $ \widehat{q}_{i,\theta} - p_{j}(\theta)= 
\frac{1}{n}\left( \xi_{j} -n p_{j}(\theta) \right)$, where $ \xi_{j} \sim {\rm Bin}\left(n,  p_{j}(\theta) \right)$.
Hence  $E_{P_{\theta}} \left( \xi_{j} -n p_{j}(\theta) \right)^{2} = n \sigma_{j}^{2}(\theta) $. 
  By Equation~(\ref{multinompid})  we have 
$$E_{P_{\theta}}\left[\left(  \xi_{j} - np_{i}(\theta) \right) \left( \xi_{j} - np_{j}(\theta) \right)\right]= -n p_{i}(\theta)p_{j}(\theta).$$ 
 Hence we substitute  in  Equation~(\ref{taylorjsdt21})  to  obtain 
\begin{eqnarray}\label{taylorjsdt211} 
M_{1} = \frac{1}{n}\sum_{j=1}^{k}\left[\frac{\partial}{\partial q_{j}} F_{j}\left( p_{j}(\theta) \right)\right]^{2} 
 \sigma_{j}^{2}(\theta)
& & \nonumber \\
& & \\
 -  \frac{2}{n}\sum_{i=1}^{k-1} \sum_{j=i+1}^{k}\frac{\partial}{\partial q_{i}} F_{i}\left( p_{i}(\theta) \right) \frac{\partial}{\partial q_{j}} F_{j}\left( p_{j}(\theta) \right)p_{i}(\theta)p_{j}(\theta).& & \nonumber 
\end{eqnarray} 
By  \citet[Eq. (80) p.~25]{ouimet2021general}
$$E_{P_{\theta}}\left[\left(  \xi_{i} - np_{i}(\theta) \right) \left(  \xi_{j} -n p_{j}(\theta)\right)^{2}\right] = n  
p_{i}(\theta) p_{j}(\theta) \left( 2p_{j}(\theta)-1 \right)$$ and thus in Equation~(\ref{taylorjsdt41})
 \begin{equation}\label{taylorjsdt411}
M_{2}=
\frac{1}{n^{2}} \sum_{i=1}^{k}  \sum_{j=1}^{k} \frac{\partial}{\partial q_{}} F_{j}\left( p_{i}(\theta) \right) \cdot\frac{\partial^{2}}{\partial q_{j}^{2}} F_{j}\left( p_{j}(\theta)\right) p_{i}(\theta) p_{j}(\theta) \left( 2p_{j}(\theta)-1 \right). 
\end{equation}
For reasons of space we set  $ M_{3}= T_{3}+T_{4} $ in  Equation~(\ref{taylorjsdt31}),  where 
$$T_{3}=\frac{1}{4}\sum_{j=1}^{k}\left[\frac{\partial^{2}}{\partial q_{j}^{2}} F_{j}\left( p_{j}(\theta) \right)\right]^{4}E_{P_{\theta}}\left[ \left( \widehat{q}_{i,\theta} - p_{j}(\theta) \right)^{4} \right]$$
 and  $$T_{4}=\frac{1}{2} \sum_{i=1}^{k-1} \sum_{j=i+1}^{k}\frac{\partial^{2}}{\partial q_{i}^{2}} F_{i}\left( p_{i}(\theta) \right) \frac{\partial^{2}}{\partial q_{j}^{2}} F_{j}\left( p_{j}(\theta) \right)E_{P_{\theta}}\left[ \left( \widehat{q}_{i,\theta} - p_{i}(\theta) \right)^{2}\left( \widehat{q}_{j,\theta} - p_{j}(\theta) \right)^{2}\right].$$    

Next, by the algorithms of  \citet[p.270]{griffiths2013raw},  \citet[Eq. (82) p.~25]{ouimet2021general}  and \citet[Table 2., p.5]{skorski2020handy},
$
E_{P_{\theta}} \left( \xi_{j} -n p_{j}(\theta) \right)^{4} = n \sigma_{j}^{2}(\theta)\left ( 1+ 3 \sigma_{j}^{2}(\theta) (n-2) \right)$.
This can, by some effort,  be checked by evaluating the  fourth derivative of the moment generating function of the binomial distribution  at zero. Then 
\begin{equation}\label{taylorjsdt311a} 
T_{3}= \frac{1}{4n^{3}}\sum_{j=1}^{k}\left[\frac{\partial^{2}}{\partial q_{j}^{2}} F_{j}\left( p_{j}(\theta) \right)\right]^{4}
 \sigma_{j}^{2}(\theta)\left ( 1+ 3 \sigma_{j}^{2}(\theta) (n-2) \right).
\end{equation}
By  \citet[Eq. (84) p.~25]{ouimet2021general}
 \begin{eqnarray}\label{ouimet5}
E_{P_{\theta}}\left[\left(  \xi_{i} - np_{i}(\theta) \right)^{2} \left(  \xi_{j} -n p_{j}(\theta)\right)^{2}\right] &=&
3n(n-2)p_{i}(\theta)^{2}p_{j}(\theta)^{2} \nonumber \\
& &  +n(n-2)p_{i}(\theta)p_{j}(\theta)\left[1-(p_{i}(\theta)+p_{j}(\theta))\right] \nonumber \\
& &  +n p_{i}(\theta)p_{j}(\theta). 
\end{eqnarray}
Hence 
 \begin{eqnarray}\label{taylorjsdt311b}
T_{4} &=&\frac{3(n-2)}{n^{3}}  \sum_{i=1}^{k-1} \sum_{j=i+1}^{k}\frac{\partial^{2}}{\partial q_{i}^{2}} F_{i}\left( p_{i}(\theta) \right) \frac{\partial^{2}}{\partial q_{j}^{2}} F_{j}\left( p_{j}(\theta) \right)p_{i}(\theta)^{2}p_{j}(\theta)^{2} \nonumber \\
& & +\frac{(n-2)}{n^{3}}\sum_{i=1}^{k-1} \sum_{j=i+1}^{k}\frac{\partial^{2}}{\partial q_{i}^{2}} F_{i}\left( p_{i}(\theta) \right) \frac{\partial^{2}}{\partial q_{j}^{2}} F_{j}\left( p_{j}(\theta) \right)p_{i}(\theta)p_{j}(\theta)\left[1-(p_{i}(\theta)+p_{j}(\theta))\right] \nonumber \\
& &  +\frac{1}{n^{3}} \sum_{i=1}^{k-1} \sum_{j=i+1}^{k}\frac{\partial^{2}}{\partial q_{i}^{2}} F_{i}\left( p_{i}(\theta) \right) \frac{\partial^{2}}{\partial q_{j}^{2}} F_{j}\left( p_{j}(\theta) \right)p_{i}(\theta)p_{j}(\theta). 
\end{eqnarray}
We find from Equation~(\ref{forstaderivjsd}) that
$$\frac{\partial}{\partial q_{j}} F_{j}\left( p_{j}(\theta) \right)  = \frac{\partial}{\partial q_{j}} D_{\phi_{\rm JS}}\left( \widehat{{\bf p}},    E_{P_{\theta}}\left[  \widehat{{\bf q}}_{\theta} \right]    \right)= 
(1-\pi) 
  \ln \left( p_{j}(\theta)/m_{j}\right),  $$ 
and from Equation~(\ref{andraderivjsd}) 
$$
\frac{\partial^{2}}{\partial q_{j}^{2}} F_{j}\left( p_{j}(\theta)\right)  =\frac{\partial^{2}}{\partial q_{j}^{2}}D_{\phi_{\rm JS}}\left( \widehat{{\bf p}},    E_{P_{\theta}}\left[  \widehat{{\bf q}}_{\theta} \right]\right)= \frac{ \widehat{p}_{j}}{p_{j}(\theta)} \frac{\pi (1-\pi)}{m_{j}}.
$$
When these two partial derivatives  are substituted in the right hand sides of Equations~(\ref{taylorjsdt211}), (\ref{taylorjsdt411}),  (\ref{taylorjsdt311a}) and (\ref{taylorjsdt311b}),  the expression displayed in 
Equation~(\ref{taylorjsdt211fin})  follows after some simplification as claimed. 
\end{proof}

\newpage

\vskip 0.2in

\bibliography{jsdkas15.bib}
\end{document}